\documentclass[11pt,fleqn]{article}

\usepackage{amsmath,amssymb,amsthm}
\usepackage[margin=1in]{geometry}
\usepackage{xcolor}
\usepackage{url}
\usepackage{array}
\usepackage{comment}
\usepackage{centernot}
\usepackage{graphicx}
\usepackage{makecell}
\usepackage{booktabs}
\usepackage{subcaption}
\usepackage{enumitem}
\usepackage[hypertexnames=false,bookmarksnumbered=true,final]{hyperref}
\usepackage{algorithm,algpseudocode}
\usepackage[capitalize,sort]{cleveref}

\def\colorschemesepia{sepia}
\def\colorschemedark{dark}
\def\colorschemelight{light}

\ifx\colorscheme\undefined
\let\colorscheme\colorschemelight
\fi

\ifx\colorscheme\colorschemelight
\colorlet{textColor}{black}
\colorlet{bgColor}{white}
\fi

\ifx\colorscheme\colorschemesepia
\definecolor{textColor}{HTML}{433423}
\definecolor{bgColor}{HTML}{fbf0da}
\fi

\ifx\colorscheme\colorschemedark
\definecolor{textColor}{HTML}{bdc1c6}
\definecolor{bgColor}{HTML}{202124}
\definecolor{textRed}{HTML}{ff968c}  %
\definecolor{textGreen}{HTML}{70cc70}  %
\definecolor{textBlue}{HTML}{8cbcff}  %
\definecolor{textCyan}{HTML}{70cccc}  %
\definecolor{textMagenta}{HTML}{d982d9}  %
\definecolor{textYellow}{HTML}{bfbf69}  %
\else
\colorlet{textRed}{red!50!black}
\colorlet{textGreen}{green!50!black}
\colorlet{textBlue}{blue!50!black}
\colorlet{textCyan}{cyan!80!black}
\colorlet{textMagenta}{magenta!80!black}
\colorlet{textYellow}{yellow!60!black}
\definecolor{textPurple}{HTML}{681da8}
\fi

\ifx\colorscheme\colorschemelight\else
\pagecolor{bgColor}
\color{textColor}
\fi

\hypersetup{colorlinks,linkcolor=textRed,citecolor=textRed,urlcolor=textBlue}
\let\eps\varepsilon
\newcommand*{\defeq}{:=}

\newcommand*{\WLoG}{Without loss of generality}
\newcommand*{\wLoG}{without loss of generality}

\newcommand*{\boolone}{\mathbf{1}}  %
\newcommand*{\bigfloor}[1]{\left\lfloor #1 \right\rfloor}
\newcommand*{\bigceil}[1]{\left\lceil #1 \right\rceil}
\newcommand*{\floor}[1]{\lfloor #1 \rfloor}
\newcommand*{\ceil}[1]{\lceil #1 \rceil}

\newcommand*{\Sum}{\operatorname{sum}}
\DeclareMathOperator*{\E}{\mathbb{E}}

\DeclareMathOperator*{\argmin}{argmin}
\DeclareMathOperator*{\argmax}{argmax}
\newcommand*{\fimplies}{\Longrightarrow}  %
\newcommand*{\nfimplies}{\centernot{\Longrightarrow}}  %
\newcommand*{\eqEnt}{-}

\newcommand*{\Acal}{\mathcal{A}}
\newcommand*{\Fcal}{\mathcal{F}}
\newcommand*{\Ical}{\mathcal{I}}
\newcommand*{\Icalhat}{\widehat{\Ical}}
\newcommand*{\Scal}{\mathcal{S}}

\newcommand*{\Ahat}{\widehat{A}}
\newcommand*{\Bhat}{\widehat{B}}
\newcommand*{\Chat}{\widehat{C}}
\newcommand*{\Ghat}{\widehat{G}}

\newcommand*{\Shat}{\widehat{S}}

\newcommand*{\chat}{\widehat{c}}
\newcommand*{\fhat}{\widehat{f}}
\newcommand*{\ghat}{\widehat{g}}
\newcommand*{\jhat}{\widehat{\jmath}}
\newcommand*{\phat}{\widehat{p}}

\newcommand*{\uhat}{\widehat{u}}
\newcommand*{\vhat}{\widehat{v}}
\newcommand*{\what}{\widehat{w}}

\newcommand*{\MMS}{\mathrm{MMS}}
\newcommand*{\MXS}{\mathrm{MXS}}
\newcommand*{\WMMS}{\mathrm{WMMS}}
\newcommand*{\pessShare}{\mathrm{pessShare}}
\newcommand*{\APS}{\mathrm{APS}}
\newcommand*{\pAPS}{\mathrm{pAPS}}
\newcommand*{\dAPS}{\mathrm{dAPS}}
\newcommand*{\MEFS}{\mathrm{MEFS}}

\newcommand*{\PROP}{\mathrm{PROP}}
\DeclareMathOperator{\minFS}{minFS}
\DeclareMathOperator{\restrict}{restrict}
\DeclareMathOperator{\improve}{improve}
\DeclareMathOperator{\optBP}{optBP}
\newcommand*{\EFXZero}{EFX$_0$}
\newcommand*{\MXSZero}{MXS$_0$}
\newcommand*{\sorted}{sorted}
\newcommand*{\loodM}{\ell\textrm{-out-of-}d\textrm{-share}}

\DeclareMathOperator{\PMRF}{PMRF}

\newcommand*{\fdInst}[4]{(#1,\allowbreak #2,\allowbreak #3,\allowbreak #4)}

\makeatletter
\g@addto@macro{\UrlBreaks}{%
\do\/%
\do\a\do\b\do\c\do\d\do\e\do\f\do\g\do\h\do\i\do\j\do\k\do\l\do\m%
\do\n\do\o\do\p\do\q\do\r\do\s\do\t\do\u\do\v\do\w\do\x\do\y\do\z%
\do\A\do\B\do\C\do\D\do\E\do\F\do\G\do\H\do\I\do\J\do\K\do\L\do\M%
\do\N\do\O\do\P\do\Q\do\R\do\S\do\T\do\U\do\V\do\W\do\X\do\Y\do\Z%
\do\0\do\1\do\2\do\3\do\4\do\5\do\6\do\7\do\8\do\9%
}
\makeatother

\allowdisplaybreaks

\newif\ifColsOne
\newif\ifColsTwo
\newif\ifVerbose
\newif\ifAnonymous

\newcommand*{\fairDefLater}{\footnote{\label{foot:fairDefLater}%
We formally define EF1 and other fairness notions in \cref{sec:notions}.}}
\newcommand*{\fairDefAgain}{}

\usepackage{xparse}
\let\realItem\item  %
\makeatletter
\NewDocumentCommand\myItem{ o }{%
   \IfNoValueTF{#1}%
      {\realItem}%
      {\realItem[#1]\def\@currentlabel{#1}\phantomsection}%
}
\makeatother
\setlist[enumerate]{before=\let\item\myItem}%

\newcolumntype{H}{>{\setbox0=\hbox\bgroup}c<{\egroup}@{}}

\newcounter{tabSerial}
\newcommand{\tabSn}{\stepcounter{tabSerial}\textcolor{textColor!50!bgColor}{\thetabSerial}}

\let\shortcite\cite
\let\citet\cite
\let\citep\cite
\newenvironment*{tightenum}{\enumerate[noitemsep]}{\endenumerate}

\let\extRef\ref
\let\extCref\cref
\let\extCrefCof\cref

\let\bigTableSize\footnotesize
\let\tableHeadSize\tiny

\ColsOnetrue
\ColsTwofalse
\Verbosetrue
\ifColsTwo\newcommand{\wrapIfTwoCols}{\\ &}\else\let\wrapIfTwoCols\relax\fi

\newtheorem{theorem}{Theorem}
\newtheorem{definition}{Definition}
\newtheorem{example}[theorem]{Example}

\newtheorem{lemma}[theorem]{Lemma}
\newtheorem{observation}[theorem]{Observation}
\newtheorem{remark}[theorem]{Remark}

\newcommand*{\dagsRef}{Figs.~\ref{fig:additive-general-nay}--\ref{fig:general-any-nny}
(pages \pageref{fig:additive-general-nay}--\pageref{fig:general-any-nny})}

\newcommand*{\tabsRefNp}{Tables~\ref{table:impls1}--\ref{table:cex-nonadd}}
\newcommand*{\tabsRef}{\tabsRefNp{} (pages \pageref{table:impls1}--\pageref{table:cex-nonadd})}

\newcommand{\defaultFloatPlacement}{}
\title{Exploring Relations among Fairness Notions\texorpdfstring{\\}{ }in Discrete Fair Division%
\thanks{This work was supported by \href{https://www.nsf.gov/awardsearch/show-award/?AWD_ID=2334461}{NSF Grant CCF-2334461}.}}
\ifAnonymous
\author{\empty}
\else
\author{
Jugal Garg%
\thanks{Department of Industrial \& Enterprise Engineering, University of Illinois at Urbana-Champaign, USA}
\\ \texttt{\small jugal@illinois.edu}
\and
Eklavya Sharma\footnotemark[2]
\\ \texttt{\small eklavya2@illinois.edu}
}
\fi
\date{\empty}

\begin{document}

\maketitle

\begin{abstract}
Fair allocation of indivisible items among agents is a fundamental and extensively studied problem. However, \emph{fairness} does not have a single universally accepted definition, leading to many competing fairness notions. Some of these notions are considered stronger or more desirable, but they are also more difficult to guarantee.
In this work, we examine 22 different fairness notions and organize them into a hierarchy. Formally, we say that a notion $F_1$ \emph{implies} another notion $F_2$ if every $F_1$-fair allocation is also $F_2$-fair. We give a near-complete picture of implications among fairness notions: for almost every pair of notions, we either prove an implication or give a counterexample demonstrating that the implication does not hold. Although some of these results are already known, many are new.
We examine multiple settings, including the allocation of goods, chores, and mixed manna, and different valuation classes like additive, submodular, and subadditive. We believe this work clarifies the relative strengths and applicability of these notions, providing a foundation for future research in fair division.
Moreover, we develop an \emph{inference engine} to automate part of our work. It is available as a user-friendly web application and may have broader applications beyond fair division.
\end{abstract}

\setcounter{tocdepth}{2}
\tableofcontents

\section{Introduction}
\label{sec:intro}

The problem of fairly allocating items among multiple agents has garnered significant attention
in multi-agent systems and game theory.
It has many real-world applications like dividing inheritance,
distributing natural resources among countries or states, allocating public housing,
divorce settlements, and assigning research papers to reviewers.

Research in fair division began with the study of \emph{divisible} resources, such as land, water, and cake
\citep{steinhaus1940sur,stromquist1980how,varian1974equity}.
\emph{Fairness} was formally defined primarily in two ways:
\emph{envy-freeness} (EF) and \emph{proportionality} (PROP).
EF ensures that each agent believes they received the best bundle compared to others,
while PROP guarantees that each agent's value for their own bundle is at least
$1/n$ times their value for the entire set of items, where $n$ is the number of agents.
By the 1980s, EF and PROP allocations were shown to exist
and be efficiently computable
\citep{steinhaus1940sur,stromquist1980how,varian1974equity}.

These positive results no longer hold for indivisible items.
For instance, when 5 identical goods must be divided among two agents, an EF or PROP allocation is not possible.
However, one can still aim for \emph{approximate} fairness; e.g.,
an allocation where one agent receives 3 goods and the other receives 2 goods is intuitively as fair as possible.
However, when we move beyond such simple examples to the general setting where
items are heterogeneous and agents have different preferences,
formally defining (approximate) fairness becomes much more complex.
As a result, many fairness notions have been proposed for the indivisible setting.

The concept of EF was relaxed to a notion called \emph{EF1}\fairDefLater{} (envy-freeness up to one item),
and algorithms were designed to guarantee EF1 allocations \citep{budish2011combinatorial,lipton2004approximately}.
A stronger relaxation of EF, called \emph{EFX}\fairDefAgain{} (envy-free up to any item),
was introduced in \citet{caragiannis2019unreasonable}.
Despite significant efforts, the existence of EFX allocations remains an open problem.
Consequently, relaxations of EFX \citep{caragiannis2023new,amanatidis2020multiple,chaudhury2021little}
as well as the existence of EFX in special cases \citep{chaudhury2024efx,plaut2020almost,amanatidis2021maximum}
have been explored.

A similar story played out for relaxations of PROP:
The existence of PROP1\fairDefAgain{} (proportional up to one item) was easy to prove \citep{aziz2021fair},
but stronger notions like MMS\fairDefAgain{} (maximin share) \citep{kurokawa2018fair}
and PROPx\fairDefAgain{} (proportional up to any item) \citep{aziz2020polynomial}
were shown to be infeasible. Thus, many relaxations of these notions have been studied
\citep{kurokawa2018fair,akrami2024breaking,akrami2023randomized,baklanov2021propm,caragiannis2023new,kobayashi2025proportional}.
Additionally, enforcing polynomial-time computability,
compatibility with efficiency notions (such as Pareto optimality),
or other constraints \citep{gourves2014near,bouveret2017fair,biswas2018fair,equbal2024fair}
makes these problems more challenging, often necessitating the use of weaker fairness concepts.

Hence, for the problem of fairly allocating indivisible items, a variety of fairness notions have been proposed,
each offering a different level of perceived fairness.
We believe that a systematic comparison of these fairness notions is essential for
guiding both practical applications and future research in this area.
To contribute to this objective, we present a comprehensive analysis of 22 different fairness notions,
examined through the lens of \emph{implications}.

Formally, we say that a fairness notion $F_1$ \emph{implies} another notion $F_2$ if
every $F_1$-fair allocation is also $F_2$-fair.
Conversely, if we can identify an $F_1$-fair allocation that is not $F_2$-fair,
we have a \emph{counterexample}, demonstrating that $F_1$ does not imply $F_2$.
For many well-studied settings in fair allocation (e.g., additive valuations),
we give a near-complete picture of the implications among fairness notions.
For almost every pair of notions, we either prove that one notion implies the other, or we give a counterexample.
These results establish a hierarchy of fairness notions,
and the counterexamples highlight the strengths and weaknesses of each notion.
See \cref{fig:additive-nny} for the implications in the context of additive valuations over goods and over chores,
and \dagsRef{} for other fair division settings.

\let\oldFp\defaultFloatPlacement
\renewcommand{\defaultFloatPlacement}{tb}
\begin{figure*}[tb]
\centering
\begin{subfigure}{0.45\textwidth}
    \centering
    \includegraphics[scale=0.55]{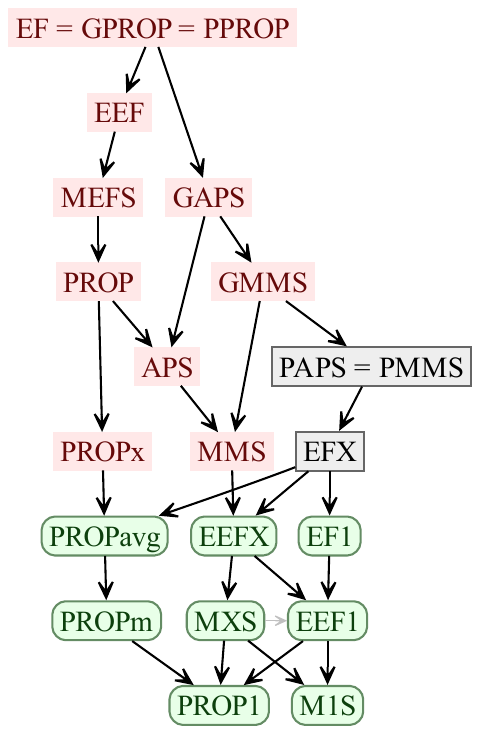}
    \caption{Goods}
    \label{fig:additive-nny:goods}
\end{subfigure}
\hfill
\begin{subfigure}{0.5\textwidth}
    \centering
    \includegraphics[scale=0.55]{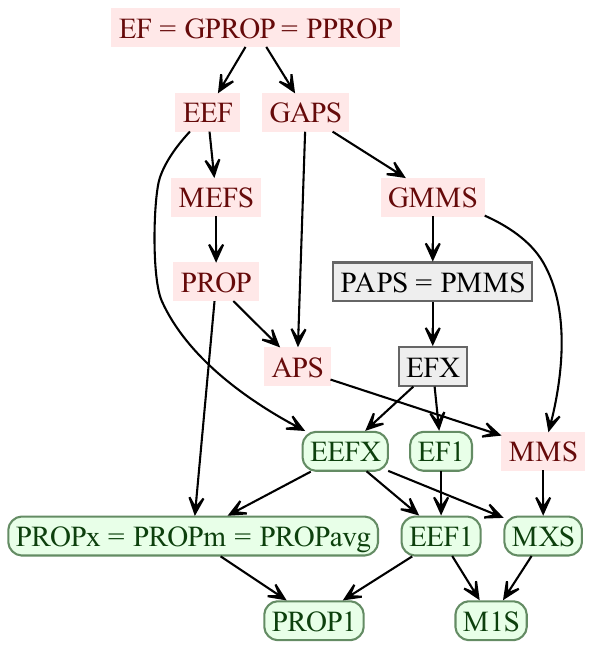}
    \caption{Chores}
    \label{fig:additive-nny:chores}
\end{subfigure}
\caption[Implications between fairness notions for additive goods and chores with equal entitlements]{%
Implications between fairness notions for additive valuations over goods and over chores
when agents have equal entitlements. There is a vertex for each fairness notion.
Notion $F_1$ implies notion $F_2$ iff there is a path from $F_1$ to $F_2$ in the graph
(except that it is not known whether MXS implies EEF1 for goods).
Borderless vertices (red) are infeasible notions,
vertices with rounded corners (green) are feasible notions,
and the feasibility of the remaining vertices (gray) are open problems.
Note that goods and chores have some key differences, e.g.,
for goods, PROP $\fimplies$ MMS $\fimplies$ EEF1 $\fimplies$ PROP1,
but for chores, MMS $\nfimplies$ PROP1, and PROP $\nfimplies$ EEF1.}
\label{fig:additive-nny}
\end{figure*}
\let\defaultFloatPlacement\oldFp

The literature on fair division covers a wide range of settings,
including distinctions between goods, chores, and mixed manna,
as well as varying entitlements (equal vs. unequal) and different classes of valuation functions.
Special cases, such as identical valuations or fair division among only two agents, have also been explored.
In this paper, we consider all combinations of these aspects of fair division.
At first, this leads to a combinatorial explosion of possible settings.
However, we address this challenge by encoding our results on
implications and non-implications in a machine-readable format,
and by implementing an \emph{inference engine} that uses these results to
automatically deduce new implications and non-implications.

For instance, if we query the inference engine with
``Does epistemic envy-freeness (EEF\fairDefAgain{}) imply maximin share (MMS) for goods with additive valuations?",
it would answer `yes' based on these three results we show in the paper:
\begin{tightenum}
\item EEF implies minimum-EF-share fairness (MEFS\fairDefAgain{}).
\item MEFS implies PROP under subadditive valuations
    (\extCref{thm:impl:mefs-to-prop} in \extCrefCof{sec:impls-extra}).
\item PROP implies MMS for superadditive valuations
    (\extCref{thm:impl:prop-to-wmms} in \extCrefCof{sec:impls-extra}).
\end{tightenum}
In \cref{sec:cpig}, we give two more examples of inference that are much more involved,
highlighting the engine's usefulness.

\subsection{Our Contributions}

We establish several implications and counterexamples between fairness notions,
leading to near-complete pictures of implications among fairness notions.
For example, in \cref{fig:additive-nny:goods} (additive valuations over goods),
there are $22 \times 21 = 462$ ordered pairs of notions,
but our work leaves only one problem (MXS $\fimplies$ EEF1) unresolved.
We obtain such near-complete pictures for 76 different fair division settings,
such as additive valuations over goods, chores, and mixed manna,
and submodular and general valuations over goods.
\Cref{sec:summary:auto} precisely enumerates these 76 settings.

\Cref{fig:additive-nny} and \dagsRef{} depict the near-complete pictures for 13 important settings.
Similar figures for other settings can be generated using the companion web application \citep{cpigjsEku}.

To achieve these tight results, we first curated a long list of implications and counterexamples, which we list in \tabsRef{}.
These tables list a total of 113 results, of which 46 are either trivial or already known,
and the remaining 67 are new results proved in this paper.
Moreover, the new counterexamples in this paper were carefully selected to be the simplest possible,
so as to clearly illustrate the distinctions between fairness notions.
Given the large number of results we prove, this was a challenging task.

Next, we developed a computer program, called the \emph{inference engine}.
It takes \tabsRefNp{} as input, and chains subsets of those results together to generate several additional results.
The hundreds of results unearthed this way give us near-complete pictures of implications
among fairness notions for many different settings.
Our inference engine is implemented as a web application in JavaScript \citep{cpigjsEku}.
The engine can be extended beyond fair division (see \cref{sec:cpig}),
and may have broader applications of independent interest.

We emphasize that the results in \tabsRefNp{} were proved manually,
without using the inference engine.

We also wrote a Python program to automate the verification of counterexamples \citep{fd-cex-checker}.
It contains functions to check whether an allocation satisfies a fairness notion,
and a representation of this paper's counterexamples (\cref{table:cex-add,table:cex-nonadd}) as Python objects.
Thus, to verify the correctness of our counterexamples,
instead of manually verifying each counterexample,
one can just run the code and verify the correctness of the fairness-checking functions.

Some fairness notions were originally defined for very specific settings;
for example, EFX was introduced in \citet{caragiannis2019unreasonable} only for additive goods.
We extend all fairness notions to the most general fair division setting we consider:
mixed manna with non-additive valuations and unequal entitlements.
In some cases, selecting an appropriate definition proved non-trivial,
and we explain the insights that motivated our choices.

\subsection{Related Work}
\label{sec:intro:related-work}

\citet{amanatidis2023fair} survey recent progress and open problems in fair division,
where many different fair division settings and fairness notions are considered.
\citet{suksompong2025weighted} gives a similar survey for unequal entitlements.

The most common setting in fair division is equally-entitled agents
having additive valuations over goods. For this setting, \citet{bouveret2016characterizing}
studied implications among 5 fairness notions: CEEI, EF, PROP, MMS, and min-max-share
(also called minimum EF share).
\citet{amanatidis2018comparing} study implications between
multiplicative approximations of fairness notions.
\citet{aziz2021fair} consider implications between EF, PROP, EF1, and PROP1 for mixed manna instead.
\citet{chakraborty2024weighted} study implications for the weighted setting
among EF1, PROP1, APS, MNW, and other notions.
Over time, as new fairness notions were proposed,
their connections with other well-established notions were studied
\citep{caragiannis2023new,babaioff2023fair,barman2018groupwise,aziz2018knowledge,aziz2024almost}.
However, the above works only consider a few fairness notions and fair division settings.
Our work, on the other hand, aims to be more exhaustive and thus have broader applicability.

\subsection{Structure of the Paper}

In \cref{sec:prelims}, we formally define the fair division problem
and describe different fair division settings.
In \cref{sec:notions}, we describe the fairness notions we consider in this paper.
In \cref{sec:summary}, we present a summary of our results.
In \cref{sec:cpig}, we describe our inference engine.
\Cref{sec:conclusion} contains concluding remarks and open problems.

\ifVerbose
Appendices \extRef{sec:settings-extra} and \extRef{sec:notions-extra}
contain details on fair division settings and fairness notions.
Appendices \extRef{sec:impls-extra}, \extRef{sec:cex-add-extra}, \extRef{sec:cex-nonadd-extra}
contain proofs of our implication and non-implication results.
\extCref{sec:feas} contains results on (in)feasibility of fairness notions.
\fi

\section{Preliminaries}
\label{sec:prelims}

In the fair division problem, there is a finite set $M$ of items
that must be distributed among a finite set $N$ of agents fairly.
Formally, we are given as input a \emph{fair division instance} $\Ical \defeq \fdInst{N}{M}{V}{w}$.
Here $w \defeq (w_i)_{i \in N}$ is a collection of positive numbers that sum to 1,
and $V \defeq (v_i)_{i \in N}$ is a collection of functions,
where $v_i: 2^M \to \mathbb{R}$ and $v_i(\emptyset) = 0$ for each $i \in N$.
$v_i$ is called agent $i$'s \emph{valuation function},
and $w_i$ is called agent $i$'s \emph{entitlement}.
Our task is to find a fair allocation. An \emph{allocation} $A \defeq (A_i)_{i \in N}$ is
a collection of pairwise-disjoint subsets of $M$ such that $\bigcup_{i \in N} A_i = M$.
The set $A_i$ is called agent $i$'s \emph{bundle} in $A$.

For any integer $k \ge 0$, define $[k] \defeq \{1, 2, \ldots, k\}$.
We generally assume \wLoG{} that $N = [n]$ and $M = [m]$.
For an agent $i$ and item $j$, we often write $v_i(j)$ instead of $v_i(\{j\})$ for notational convenience.
We denote a fair division instance by $\fdInst{N}{M}{V}{\eqEnt}$ when entitlements are equal.
For any function $u: 2^M \to \mathbb{R}$ and sets $S, T \subseteq M$, the \emph{marginal value}
of $S$ over $T$ is defined as $u(S \mid T) \defeq u(S \cup T) - u(T)$.

\subsection{Fairness Notions}
\label{sec:prelims:fairness-notions}

A \emph{fairness notion} $F$ is a function that takes as input a fair division instance $\Ical$,
an allocation $A$, and an agent $i$, and outputs either true or false.
When $F(\Ical, A, i)$ is true, we say that allocation $A$ is $F$\emph{-fair} to agent $i$,
or that agent $i$ is $F$\emph{-satisfied} by allocation $A$.
Allocation $A$ is said to be $F$-fair if it is $F$-fair to every agent.

A fairness notion $F$ is said to be \emph{feasible} if for every fair division instance,
there exists an $F$-fair allocation.
We say that a notion $F_1$ of fairness \emph{implies} another notion $F_2$ of fairness if
every $F_1$-fair allocation is also an $F_2$-fair allocation.
An allocation $A$ is $(F_1+F_2)$-fair to an agent $i$ if it is both $F_1$-fair
and $F_2$-fair to agent $i$.

\subsection{Fair Division Settings}
\label{sec:settings}

We study many fair division settings in this paper.
A setting is given by multiple \emph{features}.
By picking different values of these features, we get many different settings.
We consider 5 features in this paper:
(i) whether entitlements are equal,
(ii) whether there are only two agents,
(iii) whether agents have identical valuations,
(iv) valuation function type,
(v) marginal values.
The first three are self-explanatory. We give an overview of the last two,
and defer the details to \extCref{sec:settings-extra}.

\textbf{Valuation Function Type}:
This feature indicates how values of different sets of items are related to each other.
We consider many popular function types like additive, subadditive,
submodular, and general functions.

\textbf{Marginal values}:
This feature indicates the possible marginal values items can have.
Agent $i$'s marginal value for item $j$ over set $S$
is given by $v_i(j \mid S) \defeq v_i(S \cup \{j\}) - v_i(S)$.
We consider several marginal value types, e.g., non-negative (goods), non-positive (chores),
bivalued ($\{a, b\}$), binary ($\{0, 1\}$), negative binary ($\{0, -1\}$).

\section{Fairness Notions}
\label{sec:notions}

We now list all the fairness notions we consider in this paper.
Due to the large number of notions we consider, we describe them briefly;
details, motivation, and examples can be found in the papers we cite for each notion,
or in the survey papers cited in \cref{sec:intro:related-work}.

\subsection{Envy-Based Notions}

\begin{definition}[EF]
\label{defn:ef}
Let $\Ical \defeq \fdInst{[n]}{[m]}{(v_i)_{i=1}^n}{w}$ be a fair division instance.
In an allocation $A$, an agent $i \in [n]$ \emph{envies} another agent
$j \in [n] \setminus \{i\}$ if $v_i(A_i)/w_i < v_i(A_j)/w_j$.
Agent $i$ is \emph{envy-free} in $A$ (or $A$ is EF-fair to $i$) if
she does not envy any other agent in $A$.
\end{definition}

For unequal entitlements, most papers use the term WEF (weighted EF),
but we use the term EF instead to emphasize that unequal entitlements
is a property of the setting, not the notion.
It is easy to see that EF allocations may not exist, so several relaxations have been studied.
Two of the most popular relaxations of EF are
EF1 (envy-free up to one item) \citep{budish2011combinatorial,lipton2004approximately},
and EFX (envy-free up to any item) \citep{caragiannis2019unreasonable}.

\begin{definition}[EF1]
\label{defn:ef1}
Let $\Ical \defeq \fdInst{[n]}{[m]}{(v_i)_{i=1}^n}{w}$ be a fair division instance.
In $\Ical$, an allocation $A$ is EF1-fair to agent $i$ if for every other agent $j$,
either $i$ does not envy $j$,
or $v_i(A_i)/w_i \ge v_i(A_j \setminus \{g\})/w_j$ for some $g \in A_j$,
or $v_i(A_i \setminus \{c\})/w_i \ge v_i(A_j)/w_j$ for some $c \in A_i$.
Equivalently, $A$ is EF1-fair to $i$ if for every other agent $j$
and some $S \subseteq A_i \cup A_j$ such that $|S| \le 1$, we have
$v_i(A_i \setminus S)/w_i \ge v_i(A_j \setminus S)/w_j$.
\end{definition}

EFX was originally defined \citep{caragiannis2019unreasonable} for additive valuations over goods:
an agent $i$ \emph{strongly envies} an agent $j$ if,
even after removing a good of non-zero value from $j$'s bundle, $i$ still envies $j$.
An allocation is EFX-fair to agent $i$ if she does not strongly envy anyone.
We extend this definition to non-additive valuations over mixed manna.

\begin{definition}[EFX]
\label{defn:efx}
Let $\Ical \defeq \fdInst{[n]}{[m]}{(v_i)_{i=1}^n}{w}$ be a fair division instance.
In $\Ical$, an allocation $A$ is EFX-fair to agent $i$ if for each $j \in [n] \setminus \{i\}$,
either $i$ does not envy $j$, or both of the following hold:
\begin{enumerate}
\item Removing any positively-valued subset from $j$'s bundle eliminates $i$'s envy towards $j$:
\ifColsTwo
\else
\[ \frac{v_i(A_i)}{w_i} \ge \frac{\max(\{v_i(A_j \setminus S): S \subseteq A_j
    \textrm{ and } v_i(S \mid A_i) > 0\})}{w_j}. \]
\fi
\item Removing any negatively-valued subset from $i$'s bundle eliminates $i$'s envy towards $j$:
\ifColsTwo
\else
\[ \frac{\min\left(\left\{v_i(A_i \setminus S): S \subseteq A_i
    \textrm{ and } v_i(S \mid A_i \setminus S) < 0 \right\}\right)}{w_i} \ge \frac{v_i(A_j)}{w_j}. \]
\fi
\end{enumerate}
\end{definition}

\Cref{defn:efx} looks different from the original definition by \citet{caragiannis2019unreasonable},
and also differs from \EFXZero, an alternative definition of EFX studied by many papers
\citep{plaut2020almost,chaudhury2021little,chaudhury2024efx}.
However, we show in \extCref{sec:notions:efx} that \cref{defn:efx} is equivalent to the original definition of EFX
for submodular valuations over goods and submodular valuations over chores,
and equivalent to \EFXZero{} when marginals are (strictly) positive or negative.
Thus, \cref{defn:efx} is a generalization of known definitions.
In \extCref{sec:notions:efx}, we give more details on our motivation for defining EFX this way.

\subsection{Proportionality-Based Notions}

\begin{definition}[PROP]
\label{defn:prop}
Let $\Ical \defeq \fdInst{[n]}{[m]}{(v_i)_{i=1}^n}{w}$ be a fair division instance.
For $\Ical$, agent $i$'s \emph{proportional share} is $w_i \cdot v_i([m])$.
Allocation $A$ is \emph{proportional} (PROP) to agent $i$ if $v_i(A_i) \ge w_i \cdot v_i([m])$.
\end{definition}

A popular relaxation of PROP is PROP1 (proportional up to one item),
where each agent believes her bundle is better than the proportional share
after taking some good or giving away some chore.

\begin{definition}[PROP1, \shortcite{conitzer2017fair}]
\label{defn:prop1}
Let $\Ical \defeq \fdInst{[n]}{[m]}{(v_i)_{i=1}^n}{w}$ be a fair division instance.
In $\Ical$, an allocation $A$ is PROP1-fair to agent $i$ if
either $v_i(A_i) \ge w_i \cdot v_i([m])$,
or $v_i(A_i \cup \{g\}) > w_i \cdot v_i([m])$ for some $g \in [m] \setminus A_i$,
or $v_i(A_i \setminus \{c\}) > w_i \cdot v_i([m])$ for some $c \in A_i$.
\end{definition}

Note that \cref{defn:prop1} uses strict inequalities, whereas most papers do not.
This definition is from \citet{feige2025low}, and it makes PROP1
a slightly stronger notion without altering its fundamental properties.
See \extCref{sec:notions:prop1} for details.

PROPx (PROP up to any item) \citep{aziz2020polynomial,li2022almost} is another relaxation of PROP.
It was originally defined for additive valuations over goods:
an agent $i$ is PROPx-satisfied if transferring any good to her bundle makes her proportionally-satisfied.
We extend this definition to non-additive valuations over mixed manna.

\begin{definition}[PROPx]
\label[definition]{defn:propx}
For a fair division instance $\Ical \defeq \fdInst{[n]}{[m]}{(v_i)_{i=1}^n}{w}$,
an allocation $A$ is said to be PROPx-fair to agent $i$ iff
either $v_i(A_i) \ge w_i \cdot v_i([m])$ or both of these conditions hold:
\begin{tightenum}
\item $v_i(A_i \cup S) > w_i \cdot v_i([m])$ for every $S \subseteq [m] \setminus A_i$
    such that $v_i(S \mid A_i) > 0$.
\item $v_i(A_i \setminus S) > w_i \cdot v_i([m])$ for every $S \subseteq A_i$
    such that $v_i(S \mid A_i \setminus S) < 0$.
\end{tightenum}
\end{definition}

\Cref{defn:propx} differs from \citet{aziz2020polynomial} and \citet{li2022almost} in two ways.
First, we use strict inequalities in \cref{defn:propx}, similar to our definition of PROP1.
Second, we replaced `transferring a good' by `transferring a positive-valued subset',
and `transferring out a chore' by `transferring out a negative-valued subset',
similar to our definition of EFX. We prove in \extCref{sec:notions:propx} that,
just like with EFX, our definition of PROPx simplifies for submodular valuations over goods,
for submodular valuations over chores, and for (strictly) positive or negative marginals.

PROPavg (PROP up to the min-avg item) \citep{kobayashi2025proportional},
and PROPm (PROP up to the minimax item) \citep{baklanov2021achieving}
are strengthenings of PROP1 and relaxations of PROPx.
We formally describe them in \extCref{sec:notions:propm-propavg}.
Just like PROPx, our formal definitions for PROPm and PROPavg differ from the original definitions,
and our definitions simplify for special cases. See \extCref{sec:notions:propm-propavg} for details.

\subsection{Maximin Share and AnyPrice Share}

For equal entitlements, an agent's maximin share (MMS) \citep{budish2011combinatorial}
is the maximum value she can obtain by partitioning the goods into $n$ bundles and picking the worst one.
An allocation is MMS-fair to her if her bundle's value is at least her maximin share.

\begin{definition}[MMS, \citet{budish2011combinatorial}]
\label[definition]{defn:mms-simple}
Let $\fdInst{[n]}{[m]}{(v_i)_{i=1}^n}{\eqEnt}$ be a fair division instance.
Let $\Pi_n([m])$ be the set of all $n$-partitions of $[m]$.
Define agent $i$'s \emph{maximin share} to be
\[ \MMS_i \defeq \max_{P \in \Pi_n([m])} \min_{j=1}^n v_i(P_j). \]
An allocation $A$ is \emph{MMS-fair} to agent $i$ if $v_i(A_i) \ge \MMS_i$.
\end{definition}

Weighted MMS (WMMS) \citep{farhadi2019fair}
\ifVerbose
and pessimistic share (pessShare) \citep{babaioff2023fair} are well-known extensions
\fi
of MMS to the unequal entitlements setting.
See \extCref{sec:notions:mms} for details.

AnyPrice Share (APS) \citep{babaioff2023fair} is another fairness notion that is inspired by MMS.

\begin{definition}[APS, \citet{babaioff2023fair}]
\label[definition]{defn:aps-simple}
For a fair division instance $\Ical \defeq \fdInst{[n]}{[m]}{(v_i)_{i=1}^n}{w}$,
agent $i$'s AnyPrice Share (APS) is defined as
\[ \APS_i \defeq \min_{p \in \mathbb{R}^m}\;\max_{S \subseteq [m]: p(S) \le w_i \cdot p([m])} v_i(S). \]
Here $p$ is called the \emph{price vector}.
An allocation $A$ is APS-fair to agent $i$ if $v_i(A_i) \ge \APS_i$.
\end{definition}

\citet{babaioff2023fair} define APS for goods and for chores,
whereas \cref{defn:aps-simple} works for mixed manna.
They also give an alternative definition of APS, called the \emph{dual definition},
which does not involve prices at all.
See \extCref{sec:notions:aps} for a comparison of these definitions.

\subsection{Derived Notions}

New fairness notions can be obtained by systematically modifying existing notions.
We start with two related concepts, epistemic fairness \citep{aziz2018knowledge}
and minimum fair share \citep{caragiannis2023new}.

\begin{definition}[epistemic fairness]
\label{defn:epistemic}
Let $F$ be a fairness notion.
An allocation $A$ is \emph{epistemic-$F$-fair} to an agent $i$ if
there exists another allocation $B$ that is $F$-fair to agent $i$ and $B_i = A_i$.
$B$ is called agent $i$'s \emph{epistemic-$F$-certificate} for $A$.
\end{definition}

\begin{definition}[minimum fair share]
\label{defn:minfs}
For a fair division instance $\Ical \defeq \fdInst{[n]}{[m]}{(v_i)_{i=1}^n}{w}$
and fairness notion $F$, let $\Acal(\Ical, F, i)$ be the set of allocations
that are $F$-fair to agent $i$.
Then $A$ is minimum-$F$-share-fair to agent $i$ if there exists
an allocation $B \in \Acal(\Ical, F, i)$ such that $v_i(A_i) \ge v_i(B_i)$.
Here $B$ is called agent $i$'s \emph{minimum-$F$-share-certificate} for $A$.
Equivalently, an allocation $A$ is \emph{minimum-$F$-share-fair} to agent $i$ if
$v_i(A_i)$ is at least her minimum-$F$-share, defined as
\[ \minFS(\Ical, F, i) \defeq \min_{X \in \Acal(\Ical, F, i)} v_i(X_i). \]
\end{definition}

We now describe pairwise \citep{caragiannis2019unreasonable}
and groupwise \citep{barman2018groupwise} fairness.

\begin{definition}[restricting, pairwise and groupwise fairness]
\label{defn:restricting}
\label{defn:pairwise}
\label{defn:groupwise}
Let $\Ical \defeq \fdInst{N}{M}{(v_i)_{i \in N}}{w}$ be a fair division instance and $A$ be an allocation.
For a subset $S \subseteq N$ of agents, where $|S| \ge 2$,
let $\restrict(\Ical, A, S)$ be the pair $(\Ical^{(S)}, A^{(S)})$, where
allocation $A^{(S)} \defeq (A_j)_{j \in S}$,
instance $\Ical^{(S)} \defeq \fdInst{S}{\bigcup_{j \in S} A_j}{(v_j)_{j \in S}}{\what}$,
and weights $\what_j \defeq w_j / \sum_{k \in S} w_k$.

$A$ is \emph{pairwise-$F$-fair} to agent $i$ if for all $j \in N \setminus \{i\}$,
$A^{(\{i, j\})}$ is $F$-fair to $i$ in the instance $\Ical^{(\{i, j\})}$.
$A$ is \emph{groupwise-$F$-fair} to agent $i$ if for all non-empty $S \subseteq N \setminus \{i\}$,
$A^{(\{i\} \cup S)}$ is $F$-fair to $i$ in the instance $\Ical^{(\{i\} \cup S)}$.
\end{definition}

In this paper, we consider these derived notions:
Epistemic envy-freeness (EEF), epistemic EFX (EEFX), epistemic EF1 (EEF1),
minimum EF share (MEFS), minimum EFX share (MXS), minimum EF1 share (M1S),
pairwise proportionality (PPROP), pairwise MMS (PMMS), pairwise APS (PAPS),
groupwise proportionality (GPROP), groupwise MMS (GMMS), groupwise APS (GAPS).

\section{Summary of Results}
\label{sec:summary}

\subsection{Manually-Proved Results}
\label{sec:summary:manual}

We first prove several implications and non-implications among fairness notions manually,
i.e., without using the inference engine.
We summarize implications in \cref{table:impls1,table:impls-tribool}
(pages \pageref{table:impls1} and \pageref{table:impls-tribool}),
and defer the proofs to \extCref{sec:impls-extra}.
We state several counterexamples in \cref{table:cex-add,table:cex-nonadd}
(pages \pageref{table:cex-add} and \pageref{table:cex-nonadd}),
and defer the proofs to Appendices \extRef{sec:cex-add-extra} and \extRef{sec:cex-nonadd-extra}, respectively.
All of these counterexamples, except \cref{cex:paps-pprop-not-propm-binary-subadd,cex:gmms-not-aps-binary-subadd},
can also be found in the checker program \citep{fd-cex-checker}, where they are written as Python objects.
(\cref{cex:paps-pprop-not-propm-binary-subadd,cex:gmms-not-aps-binary-subadd}
are very large instances where agents have subadditive valuations,
so standard fairness-checking algorithms are computationally intractable for them.)

\ifVerbose
In \extCref{sec:impls-extra:unit-demand}, we prove additional implications for unit-demand valuations.
\fi
We list results regarding the feasibility and infeasibility of fairness notions in \extCref{sec:feas}.

\subsection{Automatically-Inferred Results}
\label{sec:summary:auto}

After proving implications and counterexamples manually,
we feed them into our inference engine, which uses them to infer many more results.
These results give us a near-complete picture of implications,
similar to \cref{fig:additive-nny}, for many different fair division settings.
In particular, for each of the following 76 settings, we get at most 3 unresolved implications,
i.e., there are at most three pairs $(F_1, F_2)$ where it is not known whether $F_1$ implies $F_2$
(whereas there are $22 \times 21 = 462$ ordered pairs of notions):
\begin{itemize}
\ifColsOne\raggedright\fi
\item $\overset{\text{\textcolor{gray}{entitlements}}}{\{\text{equal}, \text{unequal}\}}
    \times \overset{\text{\textcolor{gray}{no.\ of agents}}}{\{n \ge 2, n = 2\}}
    \times \overset{\text{\textcolor{gray}{identical valuations}}}{\{\text{ident}, \text{non-ident}\}}
    \times \overset{\text{\textcolor{gray}{valuation type}}}{\{\text{additive}\}}
    \times \overset{\text{\textcolor{gray}{marginal values}}}{\{\text{goods}, \text{chores}, \text{mixed}, \{0, 1\}, \{-1, 0\}\}}$
\item $\overset{\text{\textcolor{gray}{entitlements}}}{\{\text{equal}\}}
    \times \overset{\text{\textcolor{gray}{no.\ of agents}}}{\{n \ge 2, n = 2\}}
    \times \overset{\text{\textcolor{gray}{identical valuations}}}{\{\text{ident}, \text{non-ident}\}}
    \times \overset{\text{\textcolor{gray}{valuation type}}}{\{\text{submod}, \text{subadditive}, \text{general}\}}
    \times \overset{\text{\textcolor{gray}{marginal values}}}{\{\text{goods}, \mathbb{R}_{>0}, \{0, 1\}\}}$
\end{itemize}

The almost-complete resolution of the above settings can be verified by running
the \texttt{countOpenProblems.js} script included with the inference engine \citep{cpigjsEku}.

We present implication diagrams for some important settings in \dagsRef{}.
Dashed lines represent unresolved implications.
Diagrams for other settings can be viewed in the inference engine web app \citep{cpigjsEku}.

\section{Inference Engine}
\label{sec:cpig}

We wrote a computer program, called the \emph{inference engine}.
We initialize it with a list of implications and counterexamples (\cref{sec:summary:manual}),
and then we repeatedly query it with a fair division setting.
For each query, it infers additional implications and counterexamples using a method similar to transitive closure.
Although the inference procedure is simple, the engine can still make non-trivial inferences.
We present two such examples:
\begin{enumerate}
\item EEFX does not imply PROPavg (for equal entitlements over additive goods).
    This follows from a chain of (non-)implications:
    EEFX is implied by APS (Lemmas \extRef{thm:impl:aps-to-pess} and \extRef{thm:impl:mms-to-eefx}),
    APS does not imply PROPm (\extCref{cex:aps-not-propm}), and PROPm is implied by PROPavg.
    Finding such a chain can be difficult if one is not
    intimately aware of all implications and counterexamples.
    The engine helps uncover many such non-obvious insights.
\item APS does not imply EEF1 for unequal entitlements (over additive goods).
    To prove this, one cannot find a chain as in the previous example, since it does not exist.
    However, for the simpler case of binary valuations among two agents,
    APS is equivalent to PROP1 (Lemmas \extRef{thm:impl:tribool:aps} and \extRef{thm:impl:tribool:prop1}),
    EEF1 is equivalent to M1S (Lemmas \extRef{thm:impl:tribool:ef1-to-efx} and \extRef{thm:impl:mxs-to-ef1-n2}),
    and PROP1 does not imply M1S for this simpler case (\extCref{cex:prop1-not-m1s-n2}).
\end{enumerate}

For each fair division setting that we query the engine with,
it summarizes its inferred results as a directed acyclic graph (DAG).
See \cref{fig:additive-nny} and \dagsRef{} for examples of the program's output
for various important fair division settings.
The program is implemented as a web application;
see \cref{fig:cpigjs} for the program's screenshot.
Its source code is available on Github \citep{cpigjsEku}.

\begin{figure}[htb]
\centering
\ifColsOne
\includegraphics[width=0.8\textwidth]{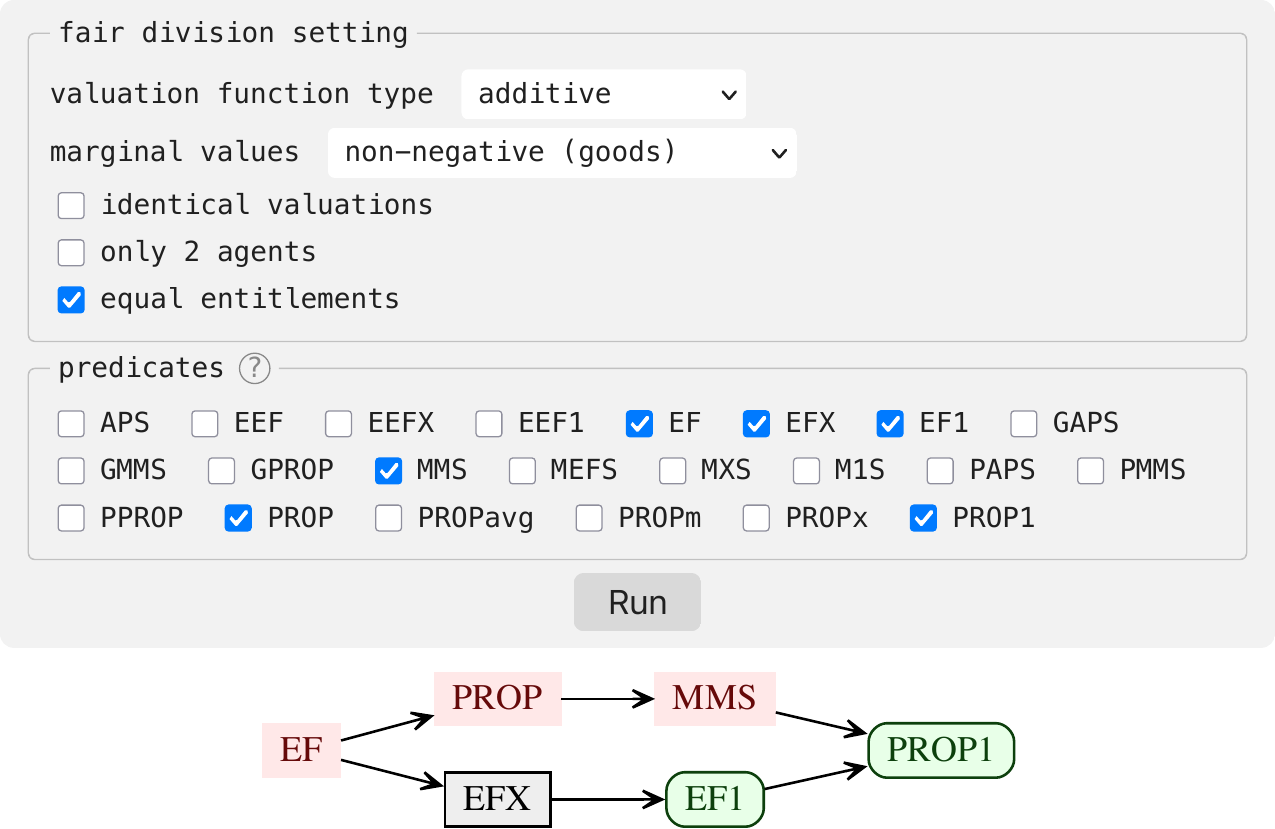}
\else
\includegraphics[width=0.9\columnwidth]{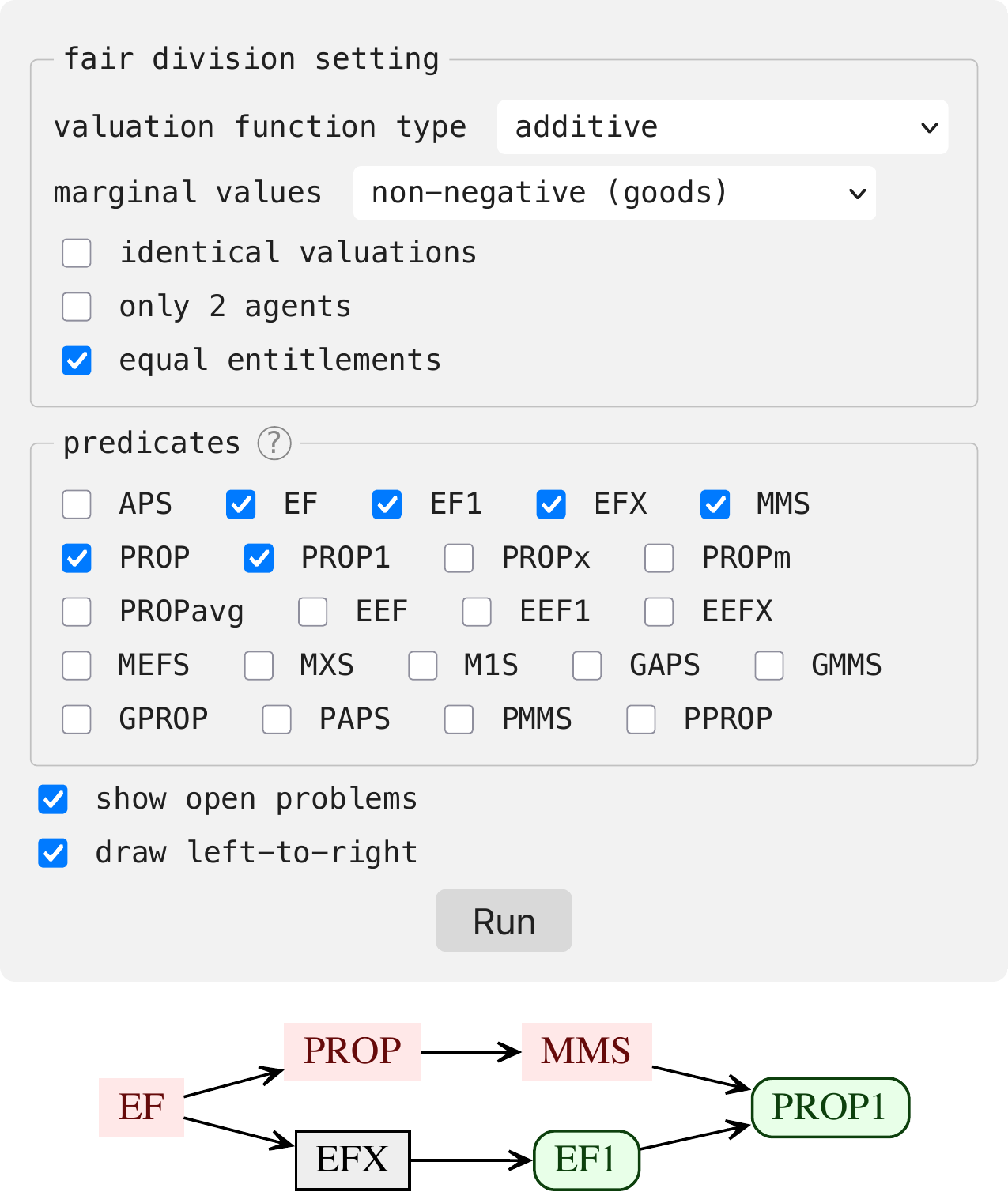}
\fi
\caption[Screenshot from cpigjs]{
Screenshot from the inference engine's web interface for fair division.}
\label{fig:cpigjs}
\end{figure}

Our program is not limited to just fair division.
It can be used more broadly for \emph{conditional predicate implications}.
A \emph{predicate} is a function whose co-domain is $\mathbb{B} \defeq \{\mathtt{true}, \mathtt{false}\}$.
Given two predicates $\phi_1, \phi_2: \Omega \to \mathbb{B}$,
we say that $\phi_1$ \emph{implies} $\phi_2$ conditioned on $S \subseteq \Omega$,
denoted as $\phi_1 \fimplies_S \phi_2$,
if $\phi_1(x) \fimplies \phi_2(x)$ for all $x \in S$.
In fair division, $\Omega$ is the set of all pairs $(\Ical, A)$,
where $\Ical$ is a fair division instance and $A$ is an allocation for $\Ical$.
A fair division setting is a subset of $\Omega$,
and fairness notions are predicates.

The inference engine's input is a tuple $(\Fcal, \Phi, I, C)$.
$\Fcal$ is a set family over a ground set $\Omega$.
    In fair division, each set in $\Fcal$ represents a setting.
    Since $\Omega$ can be uncountable, we represent sets in $\Fcal$ implicitly
    (see \extCref{sec:fd-set-family}).
    Moreover, given $S_1, S_2 \in \Fcal$, we should be able to efficiently tell whether $S_1 \subseteq S_2$.
$\Phi$ is a set of predicates over $\Omega$.
$I$ is a set of \emph{conditional implications}, i.e., a set of triples
    $(\phi_1, \phi_2, S) \in \Phi \times \Phi \times \Fcal$
    where $\phi_1 \fimplies_S \phi_2$.
$C$ is a set of \emph{conditional counterexamples}, i.e., a set of triples
    $(\phi_1, \phi_2, S) \in \Phi \times \Phi \times \Fcal$,
    where $\phi_1(x) \nfimplies \phi_2(x)$ for some $x \in S$.

We repeatedly query the engine with a set $S \in \Fcal$,
and it outputs all implications and counterexamples conditioned on $S$,
even those not explicitly present in $I$ and $C$.

The inference engine works in two steps.
In step 1, we find all implications conditioned on $S$.
To do this, we simply select implications from $I$ that are conditioned on
supersets of $S$, and compute their transitive closure.
In step 2, we find all counterexamples conditioned on $S$.
To do this, for each $(\phi_1, \phi_2, T) \in C$,
we first find all implications conditioned on $T$ like in step 1.
Next, if $\phi_1 \fimplies_T \phi'_1$ and $\phi'_2 \fimplies_T \phi_2$,
then we can infer that $\phi'_1 \nfimplies_{\!\!\!T\,\,\,} \phi'_2$, because otherwise,
by transitivity, we get $\phi_1 \fimplies_T \phi_2$.
Using this technique, we expand the set of all counterexamples.
Then we select counterexamples conditioned on subsets of $S$.

We can extend the inference engine to also make inferences about
feasibility and infeasibility of fairness notions using data from the tables in \cref{sec:summary:manual}.
Specifically, if $F_1 \fimplies_S F_2$ and $F_1$ is feasible for setting $S$,
then $F_2$ is also feasible for $S$.
Contrapositively, if $F_1 \fimplies_S F_2$ and $F_2$ is infeasible for setting $S$,
then $F_1$ is infeasible for setting $S$.

\section{Conclusion and Open Problems}
\label{sec:conclusion}

We prove several (non-)implications between fairness notions,
and for many settings, we give an almost complete picture of implications.
We believe our results would help inform further research in fair division.
This would be especially useful if one wants to extend a fair division result
to a stronger notion or a more general setting,
or study a weaker notion or a simpler setting for a hard problem.

Our framework can easily accommodate new fairness notions.
One just needs to prove a few key implications and counterexamples, and the rest can be inferred.
In fact, we originally started with only 18 fairness notions, and gradually expanded the list
to 22 notions as we found out about them.
When we added PROPavg \citep{kobayashi2025proportional}, for example, we only had to add the following 5 results:
PROPx $\fimplies$ PROPavg, PROPavg $\fimplies$ PROPm, EFX $\fimplies$ PROPavg,
PROPavg $\nfimplies$ PROPx, and PROPm $\nfimplies$ PROPavg.
The inference engine inferred PROPavg's relationship to the remaining notions.

In \cref{fig:additive-nny}, the only notions having unresolved feasibility are EFX and PMMS.
For mixed manna, even the existence of MXS allocations is open.
For equally-entitled agents having additive valuations over goods or over chores,
EF1 and Pareto optimal (PO) allocations are known to exist
\citep{caragiannis2019unreasonable,barman2018finding,mahara2025existence},
but their efficient computation remains open.
Relaxing the problem to EEF1+PO can be a helpful first step.

Here are four interesting open problems regarding implications that we could not resolve:
\begin{tightenum}
\item For additive goods (equal entitlements), does MXS imply EEF1?
    Note that the implication holds for the special case of two agents
    (\extCref{thm:impl:mxs-to-ef1-n2} in \extCref{sec:impls-extra}).
\item For additive goods (unequal entitlements), does APS imply PROP1?
    This is open even for two agents.
\item For submodular goods (equal entitlements), does MXS imply PROP1?
    This is true for binary marginals
    (\extCref{thm:impl:m1s-to-propx-binary-subadd} in \extCref{sec:impls-extra}).
\item For submodular goods with binary marginals (equal entitlements),
    does pairwise-MMS imply MMS?
\end{tightenum}
For less-studied settings, like non-additive valuations over chores,
many implications are still open.

Another interesting direction is to study implications of the form $F_1$+PO $\fimplies$ $F_2$+PO.
For additive valuations, we have already resolved most questions of this form.
This is because if $F_1 \fimplies F_2$, then $F_1$+PO $\fimplies$ $F_2$+PO.
On the other hand, most of our counterexamples use identical valuations,
where every allocation is trivially PO.

We omitted some fairness notions from our work because
they are fundamentally different from the notions we consider.
Equitability (EQ) \citep{brams1996fair}, and its relaxations like EQ1 and EQX
\citep{amanatidis2023fair,gourves2014near}, compare different utility functions with each other.
Notions like CEEI \citep{varian1974equity} and maximum Nash welfare \citep{caragiannis2019unreasonable}
include an aspect of efficiency in addition to fairness.
We also didn't study fair division of divisible items \citep{steinhaus1940sur,stromquist1980how,varian1974equity},
or a mix of divisible and indivisible items \citep{liu2024mixed}.
Nevertheless, we believe that these notions and settings offer an interesting line of research,
and our inference engine (\cref{sec:cpig}) can be readily adapted for them.

Some fairness notions and settings don't fit our model (\cref{sec:prelims}),
so it is unclear how to systematically represent results about them
and extend the inference engine (\cref{sec:cpig}) for them.
Examples of such settings include constrained fair division
\citep{gourves2014near,bouveret2017fair,biswas2018fair,equbal2024fair},
notions based on social graphs \citep{aziz2018knowledge},
parametrized notions like $\mathrm{WEF}(x, y)$ \citep{chakraborty2024weighted},
and multiplicative approximations of fairness notions \citep{amanatidis2018comparing}.
Specifically, studying multiplicative approximations presents the following challenges:
\begin{tightenum}
\item For all $\alpha \in (0,1]$, $\alpha$-EF1 implies $\alpha/(1+\alpha)$-PMMS,
    and $\alpha$-PMMS implies $\alpha/(2-\alpha)$-EF1
    (Propositions 3.8 and 4.6 in \citet{amanatidis2018comparing}).
    It's unclear how to properly visually depict such results, since we can have
    infinite chains like EF1 $\fimplies$ $1/2$-PMMS $\fimplies$ $1/3$-EF1 $\fimplies$ $1/4$-PMMS \ldots.
\item To infer new implications using existing ones, we need to be able to
    represent, compose, and evaluate arbitrary functions of approximation ratios.
    For some results, such representation is non-trivial.
    For example, MMS $\fimplies$ EEFX (\extCref{thm:impl:mms-to-eefx} in \extCref{sec:impls-extra:mms-vs-efx}),
    but $(1-\eps)$-MMS does not imply $\eps$-EEFX for $n=2$ and any $\eps > 0$
    (Proposition 4.8 in \citet{amanatidis2018comparing}).
\end{tightenum}
Extending our techniques to these settings and notions would be an interesting line of research.

\let\oldFp\defaultFloatPlacement
\renewcommand{\defaultFloatPlacement}{p}
\begin{table*}[p]
\centering
\caption[Implications among fairness notions.]{%
Implications among fairness notions.
For conciseness, we write ep-$F$ instead of epistemic-$F$,
min-$F$-sh instead of minimum-$F$-share,
g-$F$ instead of groupwise-$F$, and p-$F$ instead of pairwise-$F$.
For known results, citations can be found in the corresponding lemma statements.}
\label{table:impls1}
\bigTableSize
\setcounter{tabSerial}{0}
\begin{tabular}{ccccccccc}
\toprule & & \tableHeadSize valuation & \tableHeadSize marginals & \tableHeadSize identical & \tableHeadSize $n$ & \tableHeadSize entitlements & &
\\ \midrule \tabSn & $F$ $\Rightarrow$ ep-$F$ $\Rightarrow$ min-$F$-sh
    & -- & -- & -- & -- & -- & \extCref{thm:impl:epistemic} & trivial
\\[\defaultaddspace] \tabSn & g-$F$ $\fimplies$ $F$ + p-$F$
    & -- & -- & -- & -- & -- & \extCref{thm:impl:groupwise} & trivial
\\[\defaultaddspace] \tabSn & ep-$F$ $\fimplies$ $F$
    & -- & -- & -- & $n=2$ & -- & \extCref{thm:impl:epistemic} & trivial
\\[\defaultaddspace] \tabSn & ($F$ or p-$F$) $\fimplies$ g-$F$
    & -- & -- & -- & $n=2$ & -- & \extCref{thm:impl:groupwise} & trivial
\\ \midrule \tabSn & EF $\fimplies$ EFX+EF1
    & -- & -- & -- & -- & -- & \extCref{thm:impl:ef-to-efx+ef1} & trivial
\\[\defaultaddspace] \tabSn & EEF $\fimplies$ EEFX+EEF1
    & -- & -- & -- & -- & -- & \extCref{thm:impl:ef-to-efx+ef1} & trivial
\\[\defaultaddspace] \tabSn & MEFS $\fimplies$ MXS+M1S
    & -- & -- & -- & -- & -- & \extCref{thm:impl:ef-to-efx+ef1} & trivial
\\[\defaultaddspace] \tabSn & EFX $\fimplies$ EF1\textsuperscript{\ref{foot:efx-to-ef1}}
    & additive & -- & -- & -- & -- & \extCref{thm:impl:efx-to-ef1} & trivial
\\[\defaultaddspace] \tabSn & EEFX $\fimplies$ EEF1\textsuperscript{\ref{foot:efx-to-ef1}}
    & additive & -- & -- & -- & -- & \extCref{thm:impl:efx-to-ef1} & trivial
\\[\defaultaddspace] \tabSn & MXS $\fimplies$ M1S\textsuperscript{\ref{foot:efx-to-ef1}}
    & additive & -- & -- & -- & -- & \extCref{thm:impl:efx-to-ef1} & trivial
\\[\defaultaddspace] \tabSn & MXS $\fimplies$ M1S
    & -- & $\le 0$ & -- & -- & -- & \extCref{thm:impl:mxs-to-m1s} & \textbf{new}
\\[\defaultaddspace] \tabSn & MXS $\fimplies$ M1S
    & -- & dblMono\textsuperscript{\ref{foot:dbl-mono-1}}
    & -- & $n=2$ & -- & \extCref{thm:impl:mxs-to-m1s} & \textbf{new}
\\[\defaultaddspace] \tabSn & MXS $\fimplies$ EF1
    & additive & -- & -- & $n=2$ & -- & \extCref{thm:impl:mxs-to-ef1-n2} & \textbf{new}
\\\midrule \tabSn & PROP $\fimplies$ PROPx
    & -- & -- & -- & -- & -- & -- & trivial
\\[\defaultaddspace] \tabSn & PROP $\fimplies$ PROP1
    & -- & -- & -- & -- & -- & -- & trivial
\\[\defaultaddspace] \tabSn & PROPx $\fimplies$ PROPavg
    & -- & -- & -- & -- & -- & \extCref{thm:impl:propx-to-propavg} & trivial
\\[\defaultaddspace] \tabSn & PROPavg $\fimplies$ PROPm
    & -- & -- & -- & -- & -- & -- & trivial
\\[\defaultaddspace] \tabSn & PROPm $\fimplies$ PROPx
    & -- & -- & -- & $n=2$ & -- & -- & trivial
\\[\defaultaddspace] \tabSn & PROPm $\fimplies$ PROPx
    & -- & chores & -- & -- & -- & -- & trivial
\\[\defaultaddspace] \tabSn & PROPm $\fimplies$ PROP1
    & submodular & -- & -- & -- & -- & \extCref{thm:impl:propm-to-prop1} & folklore
\\[\defaultaddspace] \tabSn & PROPm $\fimplies$ PROP1
    & -- & $> 0$, $< 0$ & -- & -- & -- & \extCref{thm:impl:propm-to-prop1} & folklore
\\ \midrule \tabSn & MEFS $\fimplies$ PROP
    & subadditive & -- & -- & -- & -- & \extCref{thm:impl:mefs-to-prop} & known
\\[\defaultaddspace] \tabSn & EF $\fimplies$ GPROP
    & subadditive & -- & -- & -- & -- & \extCref{thm:impl:ef-to-gprop} & known
\\[\defaultaddspace] \tabSn & PROP $\fimplies$ EF
    & superadditive & -- & yes & -- & -- & \extCref{thm:impl:prop-to-ef-superadd-id} & folklore
\\[\defaultaddspace] \tabSn & PPROP $\fimplies$ EF
    & superadditive & -- & -- & -- & -- & \extCref{thm:impl:prop-to-ef-n2} & folklore
\\[\defaultaddspace] \tabSn & PPROP $\fimplies$ GPROP
    & submodular & -- & -- & -- & -- & \extCref{thm:impl:pprop-to-gprop} & \textbf{new}
\\\midrule \tabSn & EEF1 $\fimplies$ PROP1
    & submodular & -- & -- & -- & equal & \extCref{thm:impl:eef1-to-prop1-submod}
        & \textbf{new}\textsuperscript{\ref{foot:eef1-to-prop1}}
\\[\defaultaddspace] \tabSn & EEF1 $\fimplies$ PROP1
    & subadditive & chores & -- & -- & -- & \extCref{thm:impl:eef1-to-prop1-chores} & \textbf{new}
\\[\defaultaddspace] \tabSn & EF1 $\fimplies$ PROP1
    & subadditive & -- & -- & $n=2$ & -- & \extCref{thm:impl:ef1-to-prop1-n2}
    & \textbf{new}\textsuperscript{\ref{foot:ef1-to-prop1-n2}}
\\[\defaultaddspace] \tabSn & EEFX $\fimplies$ PROPx
    & subadditive & chores & -- & -- & -- & \extCref{thm:impl:eefx-to-propx} & known
\\[\defaultaddspace] \tabSn & EFX $\fimplies$ PROPavg
    & submodular & goods & -- & -- & equal & \extCref{thm:impl:efx-to-propavg} & \textbf{new}
\\[\defaultaddspace] \tabSn & EFX $\fimplies$ PROPx
    & subadditive & -- & -- & $n=2$ & -- & \extCref{thm:impl:efx-to-propx-n2} & \textbf{new}
\\[\defaultaddspace] \tabSn & MXS $\fimplies$ PROP1
    & additive & goods & -- & -- & equal & \extCref{thm:impl:mxs-to-prop1} & known
\\\midrule \tabSn & PMMS $\fimplies$ EFX
    & additive & -- & -- & -- & equal & \extCref{thm:impl:mms-to-efx-n2} & folklore
\\[\defaultaddspace] \tabSn & PWMMS $\fimplies$ EFX
    & -- & goods & -- & -- & -- & \extCref{thm:impl:mms-to-efx-n2} & folklore
\\[\defaultaddspace] \tabSn & WMMS $\fimplies$ EEFX
    & -- & goods & -- & -- & -- & \extCref{thm:impl:mms-to-eefx} & known
\\[\defaultaddspace] \tabSn & MMS $\fimplies$ MXS
    & additive & -- & -- & -- & equal & \extCref{thm:impl:mms-to-mxs} & \textbf{new}
\\[\defaultaddspace] \tabSn & MMS $\fimplies$ MXS+M1S
    & -- & goods & -- & -- & equal & \extCref{thm:impl:mms-to-mxs0} & \textbf{new}
\\\midrule \tabSn & PROP $\fimplies$ APS\textsuperscript{\ref{foot:pg}}
    & additive & -- & -- & -- & -- & \extCref{thm:impl:prop-to-aps} & known
\\[\defaultaddspace] \tabSn & PROP $\fimplies$ WMMS\textsuperscript{\ref{foot:pg}}
    & superadditive & -- & -- & -- & -- & \extCref{thm:impl:prop-to-wmms} & folklore
\\[\defaultaddspace] \tabSn & APS $\fimplies$ MMS\textsuperscript{\ref{foot:pg}}
    & -- & -- & -- & -- & equal & \extCref{thm:impl:aps-to-pess} & known
\\[\defaultaddspace] \tabSn & PWMMS $\fimplies$ PAPS
    & additive & -- & -- & -- & -- & \extCref{thm:impl:mms-to-aps-n2} & known
\\ \bottomrule
\end{tabular}

\begin{tightenum}
\item[*] \label{foot:efx-to-ef1}These results hold for additional settings.
    See \extCref{thm:impl:efx-to-ef1} in \extCref{sec:impls-extra:among-ef-efx-ef1} for details.
\item[\textdagger] \label{foot:dbl-mono-1}A function $v: 2^M \to \mathbb{R}$ is \emph{doubly monotone}
    if there is a partition $(G, C)$ of $M$ such that
    $v(g \mid \cdot) \ge 0$ $\forall g \in G$ and $v(c \mid \cdot) \le 0$ $\forall c \in C$.
\item[\textdaggerdbl] \label{foot:pg}In addition to $F_1 \fimplies F_2$,
    we also get p-$F_1$ $\fimplies$ p-$F_2$ and g-$F_1$ $\fimplies$ g-$F_2$.
\item[\S] \label{foot:eef1-to-prop1}\citet{aziz2021fair} proved this for additive valuations.
    Recently, \citet{andersen2026computing} proved this independently for submodular valuations.
\item[\P] \label{foot:ef1-to-prop1-n2}Recently, \citet{andersen2026computing} proved this independently.
\end{tightenum}
\end{table*}
\let\defaultFloatPlacement\oldFp

\let\oldFp\defaultFloatPlacement
\renewcommand{\defaultFloatPlacement}{p}
\begin{table*}[p]
\centering
\caption{Non-implications among fairness notions (additive valuations).}
\label{table:cex-add}
\bigTableSize
\setcounter{tabSerial}{0}
\begin{tabular}{cccccccccc}
\toprule & & \tableHeadSize valuation & \tableHeadSize marginals & \tableHeadSize identical & \tableHeadSize $n$ & \tableHeadSize entitlements & &
\\ \midrule \tabSn & APS+PROPx $\nfimplies$ PROP
    & $m=1$ & $1, -1$ & yes & any & equal & \extCref{cex:single-item} & trivial
\\[\defaultaddspace] \tabSn & APS+PROPx $\nfimplies$ EF1
    & additive & $1$ & yes & $n \ge 3$ & equal & \extCref{cex:share-vs-envy-goods} & folklore
\\[\defaultaddspace] \tabSn & APS+EEFX $\nfimplies$ EF1
    & additive & $-1$ & yes & $n \ge 3$ & equal & \extCref{cex:share-vs-envy-chores} & folklore
\\\midrule \tabSn & EEF $\nfimplies$ EF1
    & additive & bival\textsuperscript{\ref{foot:pos-neg-bival}}
    & no & $n=3$ & equal & \extCref{cex:eef-not-ef1} & \textbf{new}
\\[\defaultaddspace] \tabSn & PROP $\nfimplies$ MEFS
    & additive & $> 0$ & no & $n=3$ & equal & \extCref{cex:prop-not-mefs-goods} & \textbf{new}
\\[\defaultaddspace] \tabSn & PROP $\nfimplies$ MEFS
    & additive & $< 0$ & no & $n=3$ & equal & \extCref{cex:prop-not-mefs-chores} & \textbf{new}
\\[\defaultaddspace] \tabSn & MEFS $\nfimplies$ EEF
    & additive & $> 0$ & no & $n=3$ & equal & \extCref{cex:mefs-not-eef-goods} & \textbf{new}
\\[\defaultaddspace] \tabSn & MEFS $\nfimplies$ EEF1
    & additive & $< 0$ bival & no & $n=3$ & equal & \extCref{cex:mefs-not-eef1-chores} & \textbf{new}
\\\midrule \tabSn & EFX $\nfimplies$ MMS
    & additive & bival\textsuperscript{\ref{foot:pos-neg-bival}}
    & yes & $n=2$ & equal & \extCref{cex:efx-not-mms} & folklore
\\[\defaultaddspace] \tabSn & EF1 $\nfimplies$ MXS or PROPx
    & additive & bival\textsuperscript{\ref{foot:pos-neg-bival}}
    & yes & $n=2$ & equal & \extCref{cex:ef1-not-propx-mxs} & \textbf{new}
\\[\defaultaddspace] \tabSn & PROPx $\nfimplies$ M1S
    & additive & bival\textsuperscript{\ref{foot:pos-neg-bival}}
    & yes & $n=2$ & equal & \extCref{cex:propx-not-m1s} & \textbf{new}
\\[\defaultaddspace] \tabSn & MXS $\nfimplies$ PROPx
    & additive & bival\textsuperscript{\ref{foot:pos-neg-bival}}
    & yes & $n=2$ & equal & \extCref{cex:mxs-not-propx-n2} & known
\\[\defaultaddspace] \tabSn & M1S $\nfimplies$ PROP1
    & additive & bival\textsuperscript{\ref{foot:pos-neg-bival}}
    & yes & $n=2$ & equal & \extCref{cex:m1s-not-prop1} & \textbf{new}
\\\midrule \tabSn & GAPS $\nfimplies$ PROPx
    & additive & $> 0$ bival & yes & $n=3$ & equal & \extCref{cex:gaps-not-propx:additive} & \textbf{new}
\\[\defaultaddspace] \tabSn & GMMS $\nfimplies$ APS
    & additive & $> 0$, $< 0$ & yes & $n=3$ & equal & \extCref{cex:gmms-not-aps} & known
\\[\defaultaddspace] \tabSn & PMMS $\nfimplies$ MMS
    & additive & $> 0$, $< 0$ & yes & $n=3$ & equal & \extCref{cex:pmms-not-mms} & known
\\[\defaultaddspace] \tabSn & APS $\nfimplies$ PROPm
    & additive & $> 0$ & yes & $n=3$ & equal & \extCref{cex:aps-not-propm} & \textbf{new}
\\[\defaultaddspace] \tabSn & APS $\nfimplies$ PROP1
    & additive & $< 0$ bival & yes & $n=3$ & equal & \extCref{cex:aps-not-prop1-chores} & \textbf{new}
\\[\defaultaddspace] \tabSn & GAPS $\nfimplies$ PROPm
    & additive & mixed bival & yes & $n=3$ & equal & \extCref{cex:propm-mixed-manna} & \textbf{new}
\\[\defaultaddspace] \tabSn & PROPm $\nfimplies$ PROPavg
    & additive & $> 0$ bival & yes & $n=3$ & equal & \extCref{cex:propm-not-propavg} & \textbf{new}
\\\midrule \tabSn & PROP1 $\nfimplies$ M1S
    & additive & $-1$, $1$ & yes & $n=2$ & unequal & \extCref{cex:prop1-not-m1s-n2} & \textbf{new}
\\[\defaultaddspace] \tabSn & GAPS $\nfimplies$ PROPx
    & additive & bival\textsuperscript{\ref{foot:pos-neg-bival}}
    & yes & $n=2$ & unequal & \extCref{cex:gaps-not-propx-n2} & \textbf{new}
\\[\defaultaddspace] \tabSn & EF1 $\nfimplies$ EFX
    & additive & $\{-1, 1\}$ & yes & $n=2$ & unequal & \extCref{cex:ef1-not-efx-mixed-ue} & \textbf{new}
\\[\defaultaddspace] \tabSn & WMMS $\nfimplies$ M1S
    & additive & $-1$ & yes & $n=2$ & unequal & \extCref{cex:wmms-plus-m1s-chores} & \textbf{new}
\\[\defaultaddspace] \tabSn & EFX $\nfimplies$ WMMS
    & additive & $-1$ & yes & $n=2$ & unequal & \extCref{cex:wmms-plus-m1s-chores} & \textbf{new}
\\[\defaultaddspace] \tabSn & GWMMS $\nfimplies$ PROP1
    & additive & $1$ & yes & $n=3$ & unequal & \extCref{cex:prop1-plus-m1s-ue} & \textbf{new}
\\[\defaultaddspace] \tabSn & GWMMS $\nfimplies$ PROP1
    & additive & $-1$ & yes & $n=3$ & unequal & \extCref{cex:gwmms-nimpl-prop1-m1s-ue-chores} & \textbf{new}
\\[\defaultaddspace] \tabSn & PROP $\nfimplies$ M1S
    & additive & $\le 0$ bival & no & $n=3$ & unequal & \extCref{cex:prop-not-m1s-chores} & \textbf{new}
\\[\defaultaddspace] \tabSn & PROP $\nfimplies$ MEFS
    & additive & $\{0, 1\}$ & no & $n=3$ & unequal & \extCref{cex:prop-not-mefs:binary-goods} & \textbf{new}
\\\bottomrule
\end{tabular}

\begin{tightenum}
\item[*] \label{foot:pos-neg-bival}Result holds for both
    positive bivalued marginals and negative bivalued marginals.
\end{tightenum}
\end{table*}
\let\defaultFloatPlacement\oldFp

\let\oldFp\defaultFloatPlacement
\renewcommand{\defaultFloatPlacement}{p}
\begin{table*}[p]
\centering
\caption[Tribool implications]{Implications among fairness notions when
marginals belong to the set $\{-1, 0, 1\}$.}
\label{table:impls-tribool}
\bigTableSize
\setcounter{tabSerial}{0}
\begin{tabular}{ccccccccc}
\toprule & & \tableHeadSize valuations & \tableHeadSize marginals & \tableHeadSize identical & \tableHeadSize $n=2$ & \tableHeadSize entitlements & &
\\\midrule \tabSn & EF1 $\fimplies$ EFX\textsuperscript{\ref{foot:epistemic-also}}
    & additive & $\{-1, 0, 1\}$ & -- & -- & equal & \cref{thm:impl:tribool:ef1-to-efx} & trivial
\\[\defaultaddspace] \tabSn & EF1 $\fimplies$ EFX\textsuperscript{\ref{foot:epistemic-also}}
    & additive & $\{0, \pm 1\}$ & -- & -- & -- & \cref{thm:impl:tribool:ef1-to-efx} & trivial
\\[\defaultaddspace] \tabSn & PROP1 $\fimplies$ PROPx & -- & $\{-1, 0, 1\}$
    & -- & -- & -- & \cref{thm:impl:prop1-to-propx-tribool} & trivial
\\\midrule \tabSn & PROP $\fimplies$ EEF & additive & $\{-1, 0, 1\}$
    & -- & -- & equal & \cref{thm:impl:tribool:prop} & \textbf{new}
\\[\defaultaddspace] \tabSn & APS $\fimplies$ PROPx & additive & $\{-1, 0, 1\}$
    & -- & -- & -- & \cref{thm:impl:tribool:prop1,thm:impl:tribool:aps} & \textbf{new}
\\[\defaultaddspace] \tabSn & PROP1 $\fimplies$ APS & additive & $\{-1, 0, 1\}$
    & -- & -- & -- & \cref{thm:impl:tribool:prop1,thm:impl:tribool:aps} & \textbf{new}
\\[\defaultaddspace] \tabSn & M1S $\fimplies$ APS & additive & $\{-1, 0, 1\}$
    & -- & -- & equal & \cref{thm:impl:tribool:aps,thm:impl:tribool:m1s} & \textbf{new}
\\[\defaultaddspace] \tabSn & MMS $\fimplies$ EEFX & additive & $\{-1, 0, 1\}$
    & -- & -- & equal & \cref{thm:impl:tribool:mms-to-eefx} & \textbf{new}
\\\midrule \tabSn & EF1 $\fimplies$ GAPS & additive & $\{-1, 0, 1\}$
    & -- & -- & equal & \cref{thm:impl:tribool:ef1-gaps} & \textbf{new}
\\[\defaultaddspace] \tabSn & EF1 $\fimplies$ PAPS & additive & $\{-1, 0, 1\}$
    & -- & -- & -- & \cref{thm:impl:tribool:ef1-gaps} & \textbf{new}
\\[\defaultaddspace] \tabSn & EF1 $\fimplies$ GAPS & additive & $\{-1, 0\}$
    & -- & -- & -- & \cref{thm:impl:tribool:ef1-gaps} & \textbf{new}
\\[\defaultaddspace] \tabSn & EF1 $\fimplies$ GWMMS & additive & $\{0, 1\}$
    & -- & -- & -- & \cref{thm:impl:binary:ef1-to-gwmms} & \textbf{new}
\\[\defaultaddspace] \tabSn & M1S $\fimplies$ WMMS & additive & $\{0, 1\}$
    & -- & -- & -- & \cref{thm:impl:binary:m1s-to-wmms} & \textbf{new}
\\\midrule \tabSn & MEFS $\fimplies$ EEF & additive & $\{-1, 0, 1\}$
    & -- & -- & -- & \cref{thm:impl:tribool:minfs-to-epistemic} & \textbf{new}
\\[\defaultaddspace] \tabSn & MXS $\fimplies$ EEFX & additive & $\{0, \pm 1\}$
    & -- & -- & -- & \cref{thm:impl:tribool:minfs-to-epistemic} & \textbf{new}
\\[\defaultaddspace] \tabSn & M1S $\fimplies$ EEF1 & additive & $\{0, \pm 1\}$
    & -- & -- & -- & \cref{thm:impl:tribool:minfs-to-epistemic} & \textbf{new}
\\\midrule \tabSn & MMS $\fimplies$ APS\textsuperscript{\ref{foot:pg2}}
    & submodular & $\{0, 1\}$ & -- & -- & equal & \cref{thm:impl:mms-to-aps-matroid} & known
\\[\defaultaddspace] \tabSn & PROPm $\fimplies$ PROPavg
    & -- & $\{-1, 0, 1\}$ & -- & -- & -- & \cref{thm:impl:propm-to-propavg-propx-binary-subadd} & \textbf{new}
\\[\defaultaddspace] \tabSn & PROPm $\fimplies$ PROPx
    & subadditive & $\{-1, 0, 1\}$ & -- & -- & equal
    & \cref{thm:impl:propm-to-propavg-propx-binary-subadd} & \textbf{new}
\\[\defaultaddspace] \tabSn & M1S $\fimplies$ PROPx
    & subadditive & $\{0, 1\}$ & -- & -- & equal & \cref{thm:impl:m1s-to-propx-binary-subadd} & \textbf{new}
\\[\defaultaddspace] \tabSn & M1S $\fimplies$ PROPx
    & subadditive & $\{-1, 0, 1\}$ & -- & $n=2$ & -- & \cref{thm:impl:m1s-to-propx-binary-subadd} & \textbf{new}
\\[\defaultaddspace] \tabSn & M1S $\fimplies$ PROPx
    & subadditive & $\{-1, 0\}$ & -- & -- & -- & \cref{thm:impl:m1s-to-propx-binary-subadd} & \textbf{new}
\\ \bottomrule
\end{tabular}

\begin{tightenum}
\item[*] \label{foot:epistemic-also}EEF1 $\fimplies$ EEFX and M1S $\fimplies$ MXS
    also hold for the same conditions.
\item[\textdagger] \label{foot:pg2}PMMS $\fimplies$ PAPS and GMMS $\fimplies$ GAPS
    also hold for the same conditions.
\end{tightenum}
\end{table*}
\let\defaultFloatPlacement\oldFp

\let\oldFp\defaultFloatPlacement
\renewcommand{\defaultFloatPlacement}{p}
\begin{table*}[p]
\centering
\caption{Non-implications among fairness notions (non-additive valuations).}
\label{table:cex-nonadd}
\bigTableSize
\setcounter{tabSerial}{0}
\begin{tabular}{ccccccccc}
\toprule & & \tableHeadSize valuation & \tableHeadSize marginals & \tableHeadSize identical & \tableHeadSize $n$ & \tableHeadSize entitlements & &
\\\midrule \tabSn & \hfill\makecell[l]{EF+GAPS $\nfimplies$ \\ PROP1 or PROPm}\hfill
    & supermod & $\ge 0$ bival\textsuperscript{\ref{foot:nonneg-bival}}
    & yes & $n \ge 2$ & equal & \cref{cex:ef-not-prop-supmod} & trivial
\\[\defaultaddspace] \tabSn & PPROP $\nfimplies$ PROP1
    & supermod & $\ge 0$ bival\textsuperscript{\ref{foot:nonneg-bival}}
    & yes & $n \ge 3$ & equal & \cref{cex:ef-not-prop-supmod} & trivial
\\\midrule \tabSn & PROP $\nfimplies$ M1S\textsuperscript{\ref{foot:unit-dem}}
    & unit-dem & $\ge 0$ & yes & $n=2$ & equal
    & \cref{cex:prop-not-m1s-unit-demand} & \textbf{new}
\\[\defaultaddspace] \tabSn & PROP1 $\nfimplies$ PROPm\textsuperscript{\ref{foot:unit-dem}}
    & unit-dem & $\ge 0$ & yes & $n=2$ & equal
    & \cref{cex:ud:prop1-not-propm} & \textbf{new}
\\[\defaultaddspace] \tabSn & PROP $\nfimplies$ PPROP\textsuperscript{\ref{foot:unit-dem}}
    & unit-dem & $\ge 0$ & yes & $n=3$ & equal
    & \cref{cex:ud:prop-not-pprop} & \textbf{new}
\\[\defaultaddspace] \tabSn & PROPm $\nfimplies$ PROPavg\textsuperscript{\ref{foot:unit-dem}}
    & unit-dem & $\ge 0$ & yes & $n=3$ & equal
    & \cref{cex:ud:propm-not-propavg} & \textbf{new}
\\[\defaultaddspace] \tabSn & GAPS $\nfimplies$ PROPx
    & unit-dem & $\ge 0$ & yes & $n=3$ & equal & \cref{cex:gaps-not-propx:unit-demand} & \textbf{new}
\\[\defaultaddspace] \tabSn & MMS $\nfimplies$ EF1 or PROPm
    & unit-dem & $\ge 0$ & yes & $n=4$ & equal
    & \cref{cex:ud:mms-not-ef1-propm} & \textbf{new}
\\\midrule \tabSn & EF $\nfimplies$ MMS
    & submod & $\ge 0$ bival\textsuperscript{\ref{foot:nonneg-bival}}
    & yes & $n=2$ & equal & \cref{cex:ef-not-mms-pmrf} & folklore
\\[\defaultaddspace] \tabSn & EEF $\nfimplies$ EF1 or PPROP
    & submod & $\ge 0$ bival\textsuperscript{\ref{foot:nonneg-bival}}
    & yes & $n=3$ & equal & \cref{cex:eef-not-ef1-pprop-pmrf} & \textbf{new}
\\[\defaultaddspace] \tabSn & MEFS $\nfimplies$ EF1
    & submod & $\ge 0$ bival\textsuperscript{\ref{foot:nonneg-bival}}
    & yes & $n=2$ & equal & \cref{cex:mefs-not-ef1-pmrf} & \textbf{new}
\\[\defaultaddspace] \tabSn & PROP $\nfimplies$ M1S
    & submod & $\ge 0$ bival\textsuperscript{\ref{foot:nonneg-bival}}
    & yes & $n=2$ & equal & \cref{cex:prop-not-m1s-uniform-matroid} & \textbf{new}
\\[\defaultaddspace] \tabSn & MMS $\nfimplies$ APS
    & submod & $> 0$ & yes & $n=2$ & equal & \cref{cex:mms-not-aps-n2-submod} & known
\\[\defaultaddspace] \tabSn & GAPS+PPROP $\nfimplies$ PROP1
    & subadd & $\ge 0$ bival\textsuperscript{\ref{foot:nonneg-bival}}
    & yes & $n=3$ & equal & \cref{cex:gaps-not-prop1-subadd} & \textbf{new}
\\\midrule \tabSn & EF1 $\nfimplies$ MXS
    & submod & $\{0, 1\}$ & yes & $n=2$ & equal & \cref{cex:ef1-not-mxs-pmrf} & \textbf{new}
\\[\defaultaddspace] \tabSn & GAPS $\nfimplies$ PROP1 or EF1
    & subadd & $\{0, 1\}$ & yes & $n=2$ & equal & \cref{cex:gaps-not-ef1-prop1-subadd} & \textbf{new}
\\[\defaultaddspace] \tabSn & EF1 $\nfimplies$ PROP1
    & subadd & $\{0, 1\}$ & yes & $n=3$ & equal & \cref{cex:ef1-not-prop1-subadd} & \textbf{new}
\\[\defaultaddspace] \tabSn & \hfill\makecell[l]{PAPS+PPROP \\ $\nfimplies$ PROPm or M1S}\hfill
    & subadd & $\{0, 1\}$ & yes & $n=3$ & equal & \cref{cex:paps-pprop-not-propm-binary-subadd} & \textbf{new}
\\[\defaultaddspace] \tabSn & GMMS $\nfimplies$ APS
    & subadd & $\{0, 1\}$ & yes & $n=3$ & equal & \cref{cex:gmms-not-aps-binary-subadd} & \textbf{new}
\\[\defaultaddspace] \tabSn & PROPavg $\nfimplies$ PROPx
    & supermod & $\{0, 1\}$ & yes & $n \ge 3$ & equal & \cref{cex:ef-not-prop-supmod} & \textbf{new}
\\ \bottomrule
\end{tabular}

\begin{tightenum}
\item[*] \label{foot:unit-dem}Result also holds for
    submodular cancelable valuations with positive marginals.
\item[\textdagger] \label{foot:nonneg-bival}Result holds for both
    $\{0, 1\}$ marginals and positive bivalued marginals.
\end{tightenum}
\end{table*}
\let\defaultFloatPlacement\oldFp

\clearpage

\ifColsTwo\else\newcommand*{\picScale}{0.5}\fi

\begin{figure*}[p]
\centering
\begin{subfigure}{0.39\linewidth}
    \centering
    \includegraphics[scale=\picScale]{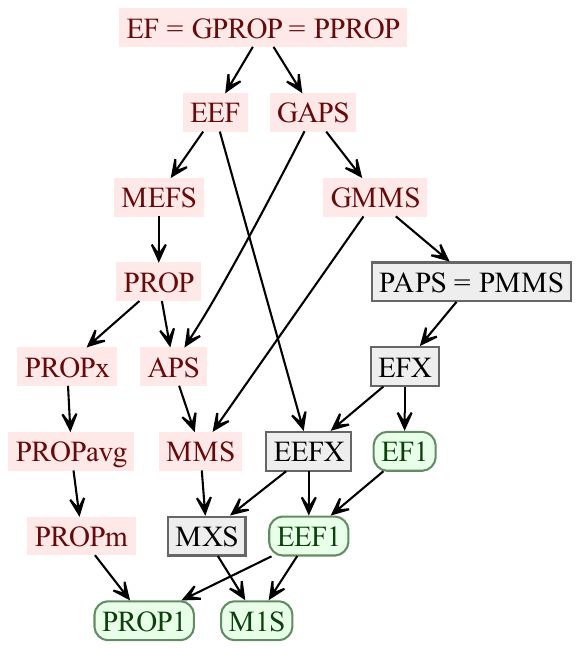}
    \caption{Any number of agents.}
    \label{fig:additive-general-nny}
\end{subfigure}
\hfill
\begin{subfigure}{0.60\linewidth}
    \centering
    \includegraphics[scale=\picScale]{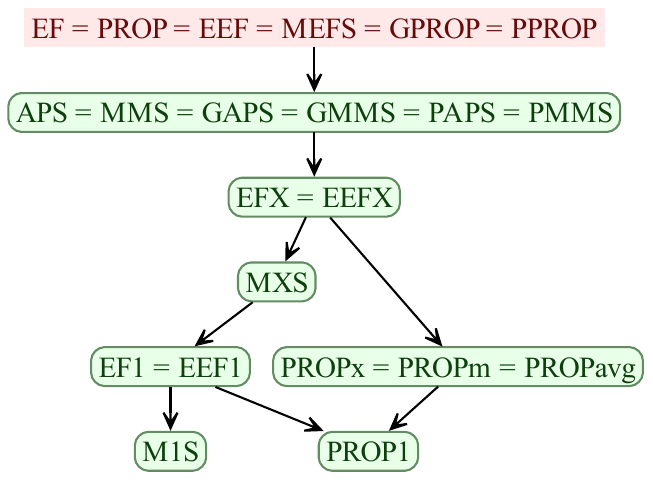}
    \caption{Two agents. We get the same DAG for goods and chores too.}
    \label{fig:additive-general-nyy}
\end{subfigure}
\caption{Additive valuations, mixed manna, equal entitlements.}
\label{fig:additive-general-nay}
\end{figure*}

\begin{figure*}[p]
\centering
\begin{subfigure}{0.49\linewidth}
    \centering
    \includegraphics[scale=\picScale]{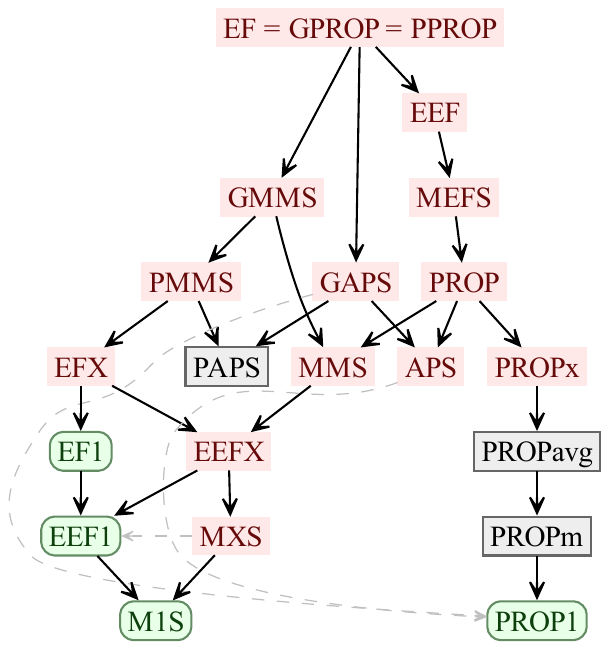}
    \caption{Goods}
    \label{fig:additive-nonneg-nnn}
\end{subfigure}
\hfill
\begin{subfigure}{0.49\linewidth}
    \centering
    \includegraphics[scale=\picScale]{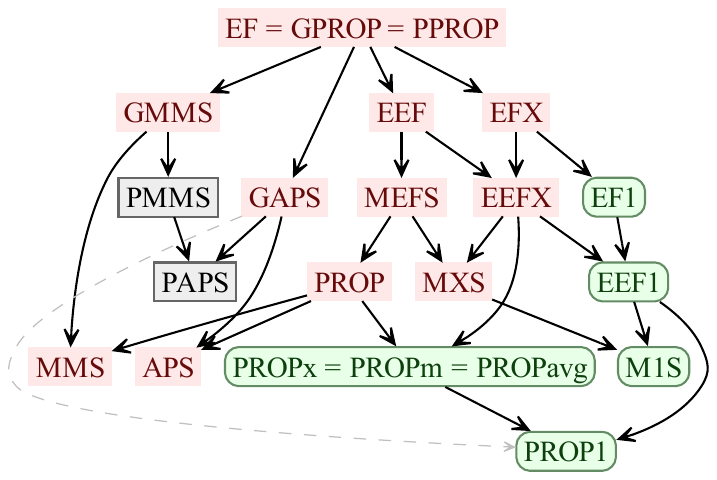}
    \caption{Chores}
    \label{fig:additive-nonpos-nnn}
\end{subfigure}
\caption{Additive valuations, unequal entitlements.}
\label{fig:additive-nnn}
\end{figure*}

\begin{figure*}[p]
\centering
\includegraphics[scale=\picScale]{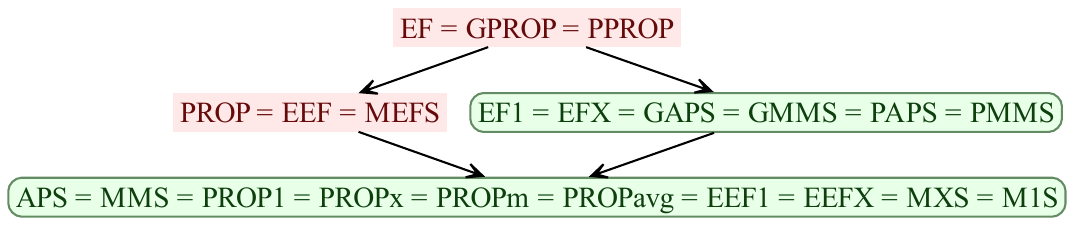}
\caption[Additive valuations, triboolean marginals]{%
Additive valuations, marginals in $\{-1, 0, 1\}$, equal entitlements.
We get the same DAG when marginals are in $\{0, -1\}$ or $\{0, 1\}$.}
\label{fig:additive-tribool-nny}
\end{figure*}

\begin{figure*}
\centering
\begin{subfigure}{0.49\linewidth}
    \centering
    \includegraphics[scale=\picScale]{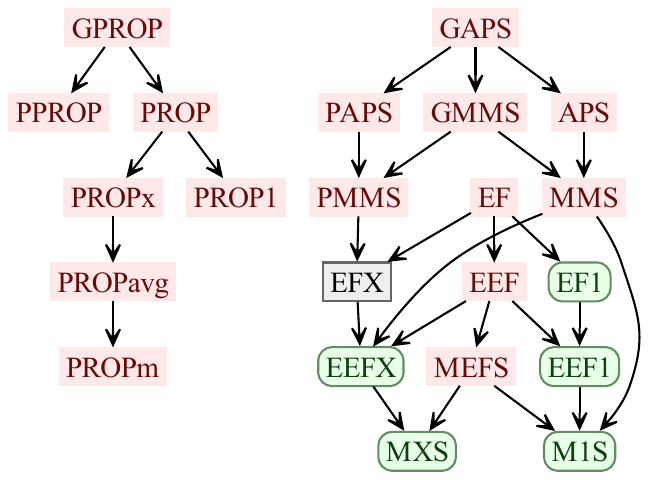}
    \caption{Non-negative marginals.}
    \label{fig:general-nonneg-nny}
\end{subfigure}
\hfill
\begin{subfigure}{0.49\linewidth}
    \centering
    \includegraphics[scale=\picScale]{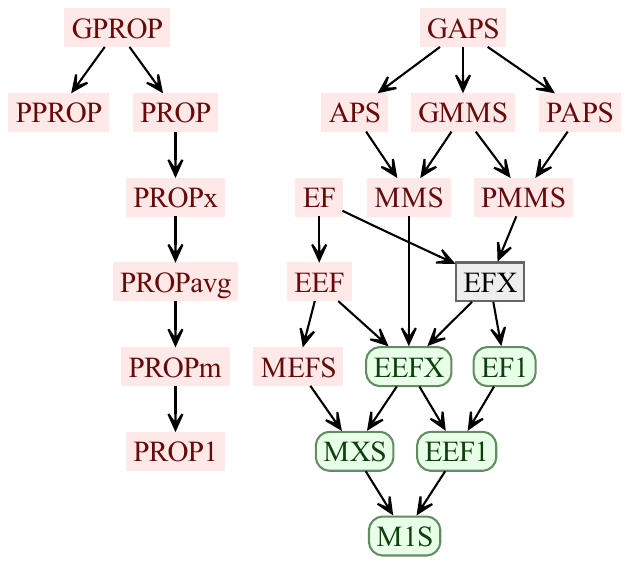}
    \caption{Positive marginals.}
    \label{fig:general-positive-nny}
\end{subfigure}
\caption{General valuations for goods with equal entitlements.}
\label{fig:general-any-nny}
\end{figure*}

\begin{figure*}
\centering
\begin{subfigure}{0.49\linewidth}
    \centering
    \includegraphics[scale=\picScale]{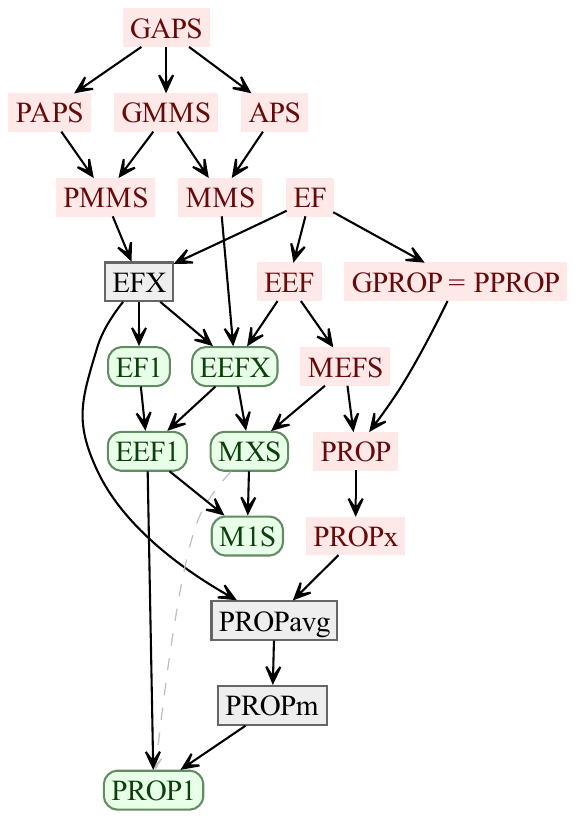}
    \caption{Submodular valuations.}
    \label{fig:submod-nonneg-nny}
\end{subfigure}
\hfill
\begin{subfigure}{0.49\linewidth}
    \centering
    \includegraphics[scale=\picScale]{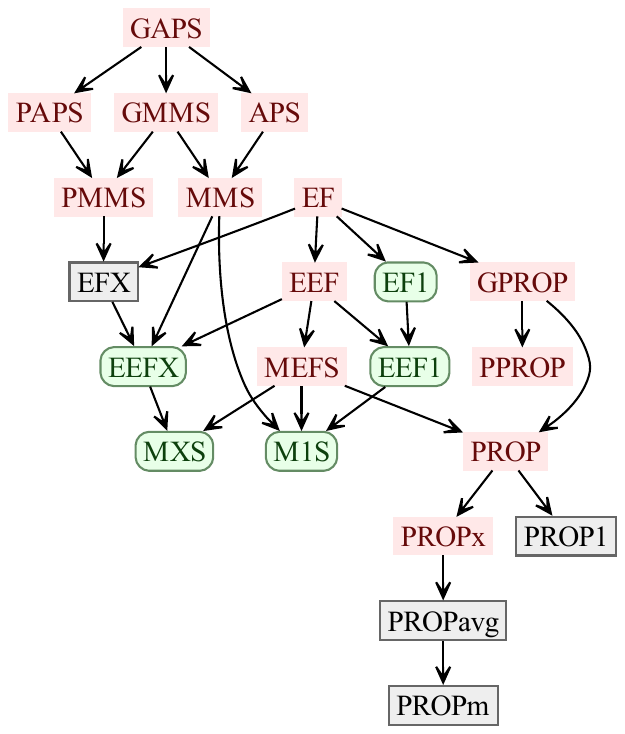}
    \caption{Subadditive valuations.}
    \label{fig:subadd-nonneg-nny}
\end{subfigure}
\caption{Submodular and subadditive valuations over goods with equal entitlements.}
\label{fig:subma-nonneg-nny}
\end{figure*}

\begin{figure*}[p]
\centering
\begin{subfigure}{0.44\linewidth}
    \centering
    \includegraphics[scale=\picScale]{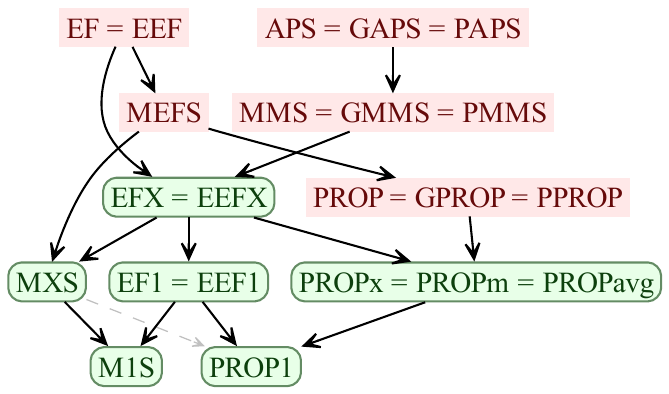}
    \caption{Two agents.}
    \label{fig:submod-nonneg-nyy}
\end{subfigure}
\hfill
\begin{subfigure}{0.55\linewidth}
\centering
\includegraphics[scale=\picScale]{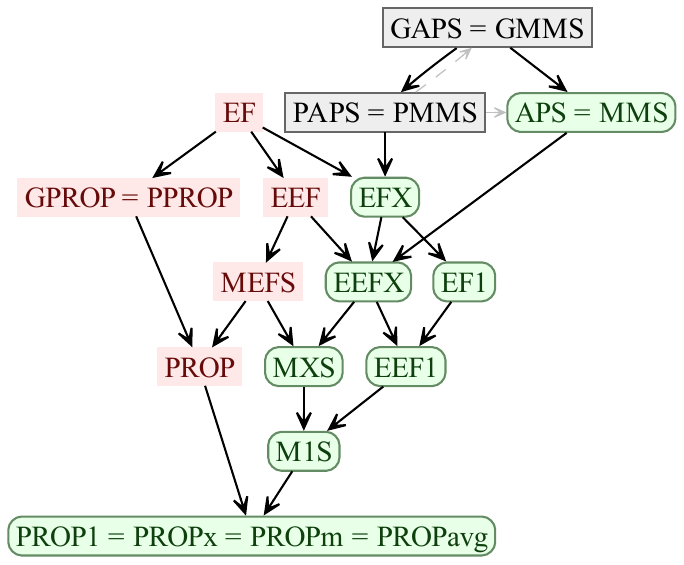}
\caption{Binary marginals (a.k.a.~matroid rank valuations).}
\label{fig:submod-binary-nny}
\end{subfigure}
\caption{Special cases of submodular valuations over goods with equal entitlements.}
\label{fig:submod-special}
\end{figure*}

\clearpage

\phantomsection
\addcontentsline{toc}{section}{References}

\newcommand{\etalchar}[1]{$^{#1}$}

\clearpage

\appendix
\crefalias{section}{appendix}
\crefalias{subsection}{appendix}
\crefalias{subsubsection}{appendix}

\section{Details on Fair Division Settings}
\label{sec:settings-extra}

\subsection{Valuation Function Type}

A function $u: 2^M \to \mathbb{R}$ is
\begin{tightenum}
\item \emph{additive} if for any two disjoint sets $S, T \subseteq M$, we have $u(S \cup T) = u(S) + u(T)$.
    \ifVerbose Equivalently, for every set $S \subseteq M$, we have $u(S) = \sum_{j \in S} u(\{j\})$.\fi
\item \emph{subadditive} if for any two disjoint sets $S, T \subseteq M$, we have $u(S \cup T) \le u(S) + u(T)$.
\item \emph{superadditive} if for any two disjoint sets $S, T \subseteq M$, we have $u(S \cup T) \ge u(S) + u(T)$.
\item \emph{submodular} if for any $S, T \subseteq M$, we have $u(S \cup T) + u(S \cap T) \le u(S) + u(T)$.
\item \emph{supermodular} if for any $S, T \subseteq M$, we have $u(S \cup T) + u(S \cap T) \ge u(S) + u(T)$.
\item \emph{cancelable} \citep{berger2021almost} if
    for any $T \subseteq M$ and $S_1, S_2 \subseteq M \setminus T$,
    we have $u(S_1 \cup T) > u(S_2 \cup T) \implies u(S_1) > u(S_2)$.
\item \emph{unit-demand} if $u(\emptyset) = 0$, and for any $\emptyset \neq S \subseteq M$,
    we have $u(S) \ge 0$ and $u(S) = \max_{j \in S} u(\{j\})$.
\item \emph{XOS} (aka \emph{fractionally subadditive}) \citep{nisan2000bidding,feige2009maximizing}
    if there exist additive functions $a_1$, $\ldots$, $a_k$, such that
    $u(S) = \max_{j=1}^k a_j(S)$ for all $S \subseteq M$.
\end{tightenum}

When $|M|=1$, $u$ belongs to all of these classes.
See \cref{fig:dag-posets:valuation} for relations between valuation function types.

\subsection{Marginal Values}

Marginal values for instance $\fdInst{N}{M}{(v_i)_{i \in N}}{w}$:

\begin{tightenum}
\item \emph{goods}: $v_i(j \mid S) \ge 0$ $\forall$ $S \subseteq M$, $j \in M \setminus S$, $i \in N$.
\item \emph{chores}: $v_i(j \mid S) \le 0$ $\forall$ $S \subseteq M$, $j \in M \setminus S$, $i \in N$.
\item \emph{positive}: $v_i(j \mid S) > 0$ $\forall$ $S \subseteq M$, $j \in M \setminus S$, $i \in N$.
\item \emph{negative}: $v_i(j \mid S) < 0$ $\forall$ $S \subseteq M$, $j \in M \setminus S$, $i \in N$.
\item \emph{bivalued}: $\exists a, b \in \mathbb{R}$ such that
    $v_i(j \mid S) \in \{a, b\}$ $\forall$ $S \subseteq M$, $j \in M \setminus S$, $i \in N$.
\item \emph{binary}: $v_i(j \mid S) \in \{0, 1\}$ $\forall$ $S \subseteq M$, $j \in M \setminus S$, $i \in N$.
\item \emph{negative binary}: $v_i(j \mid S) \in \{0, -1\}$ $\forall$ $S \subseteq M$, $j \in M \setminus S$, $i \in N$.
\item \emph{doubly monotone}: $\forall i \in N$, there is a partition $(G, C)$ of $M$ such that
    $\forall S \subseteq M$, we have $v_i(g \mid S) \ge 0$ $\forall g \in G \setminus S$ and
    $v_i(c \mid S) \le 0$ $\forall c \in C \setminus S$.
\item \emph{doubly strictly monotone}: $\forall i \in N$, there is a partition $(G, C)$ of $M$ such that
    $\forall S \subseteq M$, we have $v_i(g \mid S) > 0$ $\forall g \in G \setminus S$ and
    $v_i(c \mid S) < 0$ $\forall c \in C \setminus S$.
\end{tightenum}

We can break up the class of bivalued instances into positive bivalued, negative bivalued,
binary, negative binary, and mixed bivalued
(mixed means that exactly one of $a$ and $b$ is positive and the other is negative).
See \cref{fig:dag-posets:marginals} for relations between marginal values.

\ifColsOne
\begin{figure*}[htb]
\centering
\begin{subfigure}{0.4\linewidth}
    \centering
    \includegraphics[scale=0.5]{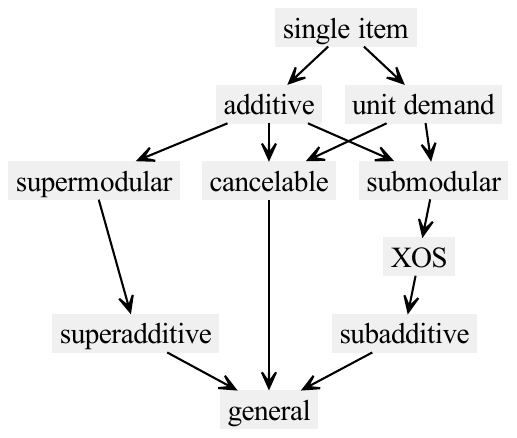}
    \caption{Valuation function type}
    \label{fig:dag-posets:valuation}
\end{subfigure}
\hfill
\begin{subfigure}{0.59\linewidth}
    \centering
    \includegraphics[scale=0.5]{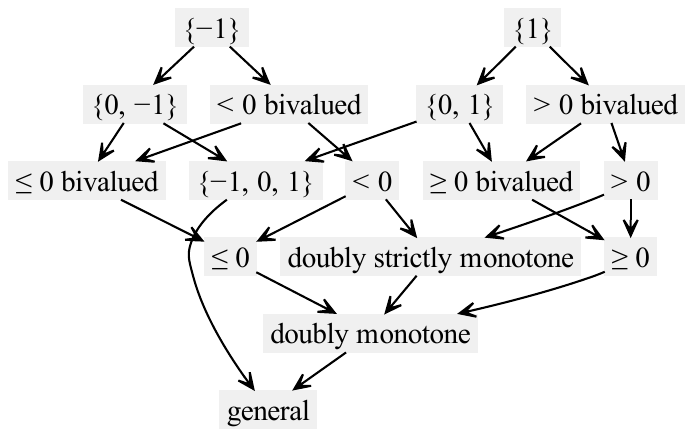}
    \caption{Marginal values}
    \label{fig:dag-posets:marginals}
\end{subfigure}
\caption[Hasse diagrams of valuation function type and marginal values]{%
Valuation function type and marginal values represented as digraphs.
There is a path from $u$ to $v$ iff $u$ is a subclass of $v$.}
\label{fig:dag-posets}
\end{figure*}
\fi

\section{Properties of Submodular Functions}
\label{sec:submod-funcs}

In this section, we describe some properties of submodular functions that would be
useful to understand results in subsequent sections.
We begin by stating multiple equivalent definitions of submodularity.
Recall that $u(X \mid Y) \defeq u(X \cup Y) - u(Y)$ is called
the \emph{marginal} value of $X$ given $Y$.

\begin{definition}[submodularity, \citet{schrijver2003submodular}]
The following are equivalent definitions for a function $u: 2^M \to \mathbb{R}$ being submodular:
\begin{tightenum}
\item for any $S, T \subseteq M$, we have $u(S \cup T) + u(S \cap T) \le u(S) + u(T)$.
\item for any $S \subseteq T \subseteq M$ and $g \in M \setminus T$,
    we have $u(g \mid S) \ge u(g \mid T)$.
\item for any $S \subseteq T \subseteq M$ and $X \subseteq M \setminus T$,
    we have $u(X \mid S) \ge u(X \mid T)$.
\end{tightenum}
\end{definition}

Next, we analyze properties of functions defined using marginal values of submodular functions.

\begin{lemma}[marginal gains are subadditive]
\label[lemma]{thm:submod-mg-subadd}
Let $u: 2^M \to \mathbb{R}$ be a submodular function and let $T \subseteq M$.
For any $X \subseteq M \setminus T$, define $f(X) \defeq u(X \mid T)$.
Then $f$ is subadditive, i.e., for any disjoint sets $X, Y \subseteq M \setminus T$,
we have $f(X \cup Y) \le f(X) + f(Y)$.
\end{lemma}
\begin{proof}
\begin{align*}
f(X \cup Y) &= u(X \cup Y \mid T)
    \wrapIfTwoCols= u(X \mid T) + u(Y \mid X \cup T)
\\ &\le u(X \mid T) + u(Y \mid T)
    \wrapIfTwoCols= f(X) + f(Y).
\qedhere
\end{align*}
\end{proof}

\ifColsTwo\else\let\nlIfTwoCols\relax\fi
\begin{lemma}[marginal losses are superadditive]
\label[lemma]{thm:submod-ml-superadd}
Let $u: 2^M \to \mathbb{R}$ be a submodular function and let $T \subseteq M$.
For any $X \subseteq T$, define $f(X) \defeq u(X \mid T \setminus X)$.
Then $f$ is superadditive, i.e., for any disjoint sets $X, Y \subseteq T$,
we have $f(X \cup Y) \ge f(X) + f(Y)$.
\end{lemma}
\begin{proof}
\begin{align*}
f(X \cup Y) &= u(X \cup Y \mid T \setminus (X \cup Y))
\wrapIfTwoCols= u(X \mid T \setminus (X \cup Y)) + u(Y \mid T \setminus Y)
\\ &\ge u(X \mid T \setminus X) + u(Y \mid T \setminus Y)
\wrapIfTwoCols= f(X) + f(Y).
\qedhere
\end{align*}
\end{proof}

\section{Details on Fairness Notions}
\label{sec:notions-extra}

Although we do not define any new fairness notions in this paper,
we extend some of them to more general settings than they were originally defined for.
For some notions, this extension is not obvious and is based on careful deliberation.
Here we show how we arrived at these extensions and why they make sense.

\subsection{EFX}
\label{sec:notions:efx}

Defining EFX in the fully general setting (non-additive valuations, mixed manna) is tricky.
So, we start with the definition of EFX for additive goods,
and gradually build up to the general definition of EFX from there.
Some of these ideas also appear in \citet{caragiannis2022existence}.

There are actually two competing definitions of EFX for additive goods.
According to the original definition by \citet{caragiannis2019unreasonable},
an allocation $A$ is EFX-fair to agent $i$ if for every other agent $j$,
removing any positively-valued good from $j$'s bundle eliminates $i$'s envy. Formally,
\[ \frac{v_i(A_i)}{w_i} \ge \max_{g \in A_j: v_i(g) > 0} \frac{v_i(A_j \setminus \{g\})}{w_j}. \]

A different definition, often called \EFXZero{},
does not require the good $g$ to have a positive value to $i$ \citep{plaut2020almost}.
See \cref{defn:efx0-goods} for a formal definition,
and \cref{defn:efx0-chores} for the chores analogue.

\begin{definition}[\EFXZero{} for goods]
\label[definition]{defn:efx0-goods}
In a fair division instance $\fdInst{[n]}{[m]}{(v_i)_{i=1}^n}{w}$ over goods,
an allocation $A$ is \EFXZero-fair to agent $i$ if for every other agent $j$, and every $g \in A_j$,
\[ \frac{v_i(A_i)}{w_i} \ge \frac{v_i(A_j \setminus \{g\})}{w_j}. \]
\end{definition}

\begin{definition}[\EFXZero{} for chores]
\label[definition]{defn:efx0-chores}
In a fair division instance $\fdInst{[n]}{[m]}{(v_i)_{i=1}^n}{w}$ over chores,
an allocation $A$ is \EFXZero-fair to agent $i$ if for every other agent $j$, and every $c \in A_i$,
\[ \frac{-v_i(A_i \setminus \{c\})}{w_i} \le \frac{-v_i(A_j)}{w_j}. \]
\end{definition}

\EFXZero{} is known to be incompatible with Pareto optimality (PO) based on the following simple example \citep{plaut2020almost},
whereas it is not known whether EFX is compatible with PO for additive valuations over goods.

\begin{example}
\label[example]{ex:efx-po}
Consider the fair division instance $\fdInst{[2]}{\{g_1,\allowbreak g_2,\allowbreak g_3\}}{(v_i)_{i=1}^2}{\eqEnt}$,
where $v_1$ and $v_2$ are additive, and
\begin{align*}
   v_1(g_1) &= 1, & v_1(g_2) &= 0, & v_1(g_3) &= 10,
\\ v_2(g_1) &= 0, & v_2(g_2) &= 1, & v_2(g_3) &= 10.
\end{align*}
\end{example}

In \cref{ex:efx-po}, for any PO allocation $A$, we have $g_1 \in A_1$ and $g_2 \in A_2$
(otherwise we can transfer $g_1$ to agent 1 or $g_2$ to agent 2 to Pareto-dominate $A$).
The agent who did not get $g_3$ in $A$ is not \EFXZero-satisfied, although she is EFX-satisfied.

While \EFXZero{} is trivial to extend to non-additive valuations, EFX is not.
This is because there exist examples where every good in $j$'s bundle has zero value to agent $i$.
(Additionally, the good's marginal value over $A_j$ and $A_i$ may also be zero.)
Before we try to address this issue, let us instead jump to the setting of additive mixed manna.

One way to define EFX for mixed manna is this:
agent $i$ is EFX-satisfied by allocation $A$ if for every other agent $j$,
either agent $i$ does not envy $j$, or $i$'s envy towards $j$ vanishes after
either removing the least valuable positively-valued item from $j$
or after removing the most valuable negatively-valued item from $i$.
However, we argue that this is not sufficient.

\begin{example}
\label[example]{ex:efx-mixed-manna}
Consider a fair division instance $\Ical$ having 2 agents with equal entitlements,
identical additive valuations, two goods of values $10$ each,
and two chores of values $-9$ each.
Consider an allocation $A$ where all the 4 items are allocated to agent 2.
Then agent 1 would be EFX-satisfied by $A$, even though allocation $B$,
where each agent gets one good and one chore, is intuitively fairer.
\end{example}

For goods, EFX is considered one of the strongest fairness notions,
so we would like it to be a very strong notion for mixed manna too.
The key observation in \cref{ex:efx-mixed-manna} is that in allocation $A$,
we can transfer a set of items (containing one good and one chore) from agent 2 to agent 1
and get a fairer allocation.
This suggests that instead of (hypothetically) removing a single good from $j$
or a single chore from $i$ in the definition of EFX,
we should remove a positively-valued subset of $A_j$
or a negatively-valued subset of $A_i$.

Hence, for additive mixed manna, we say that an allocation $A$ is EFX-fair to agent $i$
if for every other agent $j$, $i$ does not envy $j$, or both of the following hold:
\begin{enumerate}
\ifColsTwo
\else
\item $\displaystyle \frac{v_i(A_i)}{w_i} \ge \frac{\max(\{v_i(A_j \setminus S): S \subseteq A_j
    \textrm{ and } v_i(S) > 0\})}{w_j}$.
\fi
\ifColsTwo
\else
\item $\displaystyle \frac{\min\left(\left\{v_i(A_i \setminus S): S \subseteq A_i
    \textrm{ and } v_i(S) < 0 \right\}\right)}{w_i} \ge \frac{v_i(A_j)}{w_j}$.
\fi
\end{enumerate}

This definition also hints towards how to handle goods with non-additive valuations.
Even if every good in $j$'s bundle has (marginal) value zero to agent $i$,
some subset of $j$'s bundle must have positive (marginal) value.
We replace $v(S) > 0$ by $v_i(S \mid A_i) > 0$ to ensure that
transferring $S$ from $j$ to $i$ leads to an increase in $i$'s valuation,
otherwise we lose compatibility with PO, even with identical valuations, as the following example demonstrates.

\begin{example}
\label[example]{ex:efx-two-colors}
Consider a fair division instance with two equally-entitled agents having identical valuations.
There are 2 red goods and 4 blue goods. The value of any bundle is $\max(k_r, k_b)$,
where $k_r$ and $k_b$ are the numbers of red and blue goods in the bundle, respectively.
\end{example}

For \cref{ex:efx-two-colors}, in any allocation, some agent gets at most 2 blue goods,
and that agent's value for her own bundle can be at most 2.
Also, the maximum value that any bundle can have is 4.
Hence, an allocation is PO iff one agent gets all the blue goods
and the other agent gets all the red goods.
To ensure the compatibility of EFX and PO, we must define EFX such that this allocation is EFX.

Hence, the definition of EFX we converge on is \cref{defn:efx}.
We now identify special cases where this definition
is equivalent to other well-known definitions of EFX.

\begin{lemma}
\label[lemma]{thm:efx-equiv-positive}
In the fair division instance $\Ical \defeq \fdInst{[n]}{[m]}{(v_i)_{i=1}^n}{w}$,
if all marginals are positive for agent $i$
(i.e., $v_i(g \mid R) > 0$ for all $R \subseteq [m]$ and $g \in [m] \setminus R$),
then $A$ is EFX-fair to agent $i$ iff $A$ is \EFXZero-fair to $i$.
\end{lemma}
\begin{proof}[Proof sketch]
$v_i(A_j \setminus S) = v_i(A_j) - v_i(S \mid A_j \setminus S)$.
Thus, in \cref{defn:efx}, we need to find $S \subseteq A_j$ such that
$v_i(S \mid A_i) > 0$ and $v_i(S \mid A_j \setminus S)$ is minimized.
Since marginals are positive, it suffices to only check sets of cardinality 1.
\end{proof}

\begin{lemma}
\label[lemma]{thm:efx-equiv-negative}
In the fair division instance $\Ical \defeq \fdInst{[n]}{[m]}{(v_i)_{i=1}^n}{w}$,
if all marginals are negative for agent $i$
(i.e., $v_i(c \mid R) < 0$ for all $R \subseteq [m]$ and $c \in [m] \setminus R$),
then $A$ is EFX-fair to agent $i$ iff $A$ is \EFXZero-fair to agent $i$.
\end{lemma}
\begin{proof}[Proof sketch]
$v_i(A_i \setminus S) = v_i(A_i) - v_i(S \mid A_i \setminus S)$.
Thus, in \cref{defn:efx}, we need to find $S \subseteq A_i$ such that
$-v_i(S \mid A_i \setminus S) > 0$ and $-v_i(S \mid A_i \setminus S)$ is minimized.
Since marginals are negative, it suffices to only check sets of cardinality 1.
\end{proof}

\begin{lemma}
\label[lemma]{thm:efx-equiv-submod-goods}
In the fair division instance $\Ical \defeq \fdInst{[n]}{[m]}{(v_i)_{i=1}^n}{w}$,
if $v_i$ is submodular and all marginals are non-negative for agent $i$
(i.e., $v_i(g \mid R) \ge 0$ for all $R \subseteq [m]$ and $g \in [m] \setminus R$),
then $A$ is EFX-fair to agent $i$ iff for all $j \in [n] \setminus \{i\}$, we have
\[ \frac{v_i(A_i)}{w_i} \ge \frac{\max
    \ifColsTwo
\else
    (\{v_i(A_j \setminus \{g\}): g \in A_j \textrm{ and } v_i(g \mid A_i) > 0\})
    \fi
    }{w_j}. \]
\end{lemma}
\begin{proof}[Proof sketch]
In \cref{defn:efx}, we need to find $S \subseteq A_j$ such that
$v_i(S \mid A_i) > 0$ and $v_i(S \mid A_j \setminus S)$ is minimized.
By the subadditivity of $v_i$'s marginal gains (\cref{thm:submod-mg-subadd}), we get
$v_i(S \mid A_i) \le \sum_{g \in S} v_i(g \mid A_i)$.
Since $v_i(S \mid A_i) > 0$, there exists $\ghat \in S$ such that $v_i(\ghat \mid A_i) > 0$.

Moreover, $v_i(S \mid A_j \setminus S) \ge v_i(\ghat \mid A_j \setminus \{\ghat\})$
since $v_i(S \mid A_j \setminus S) = v_i(S \setminus \{\ghat\} \mid A_j \setminus S) + v_i(\ghat \mid A_j \setminus \{\ghat\})$.
Thus, we can replace $S$ by $\{\ghat\}$,
so we can assume \wLoG{} that $|S| = 1$.
\end{proof}

\begin{lemma}
\label[lemma]{thm:efx-equiv-submod-chores}
In the fair division instance $\Ical \defeq \fdInst{[n]}{[m]}{(v_i)_{i=1}^n}{w}$,
if $v_i$ is submodular and all marginals are non-positive for agent $i$
(i.e., $v_i(c \mid R) \le 0$ for all $R \subseteq [m]$ and $c \in [m] \setminus R$),
then $A$ is EFX-fair to agent $i$ iff for all $j \in [n] \setminus \{i\}$, we have
\[ \frac{\max
    \ifColsTwo
\else
    (\{-v_i(A_i \setminus \{c\}): c \in A_i \textrm{ and } v_i(c \mid A_i \setminus \{c\}) < 0\})
    \fi
    }{w_i} \le \frac{-v_i(A_j)}{w_j}. \]
\end{lemma}
\begin{proof}[Proof sketch]
Let $d_i \defeq -v_i$. In \cref{defn:efx}, we need to find $S \subseteq A_i$ such that
$d_i(S \mid A_i \setminus S) > 0$ and $d_i(S \mid A_i \setminus S)$ is minimized.
By the superadditivity of $v_i$'s marginal losses (\cref{thm:submod-ml-superadd}), we get
$d_i(S \mid A_i \setminus S) \le \sum_{c \in S} d_i(c \mid A_i \setminus \{c\})$.
Since $d_i(S \mid A_i \setminus S) > 0$, there exists $\chat \in A_i$
such that $d_i(\chat \mid A_i \setminus \{\chat\}) > 0$.

Moreover, $d_i(S \mid A_i \setminus S) \ge d_i(\chat \mid A_i \setminus \{\chat\})$ since
$d_i(S \mid A_i \setminus S)
= d_i(S \setminus \{\chat\} \mid A_i \setminus S) + d_i(\chat \mid A_i \setminus \{\chat\})$.
Thus, we can replace $S$ by $\{\chat\}$,
so we can assume \wLoG{} that $|S| = 1$.
\end{proof}

\subsection{MMS}
\label{sec:notions:mms}

\begin{definition}[WMMS, \citet{farhadi2019fair}]
\label[definition]{defn:wmms}
Let $\Ical \defeq (N, M, (v_i)_{i \in N}, w)$ be a fair division instance.
Let $\Acal$ be the set of all allocations for $\Ical$.
Then agent $i$'s \emph{weighted maximin share} (WMMS) is
\[ \WMMS_i := w_i \max_{X \in \Acal} \min_{j \in N} \frac{v_i(X_j)}{w_j}. \]
An allocation $A$ is WMMS-fair to agent $i$ if $v_i(A_i) \ge \WMMS_i$.
An allocation $X$ that maximizes $w_i \min_{j \in N} \frac{v_i(X_j)}{w_j}$
is called agent $i$'s \emph{WMMS partition}.
\end{definition}

For equal entitlements, we get $\WMMS_i = \MMS_i$.

\ifVerbose
\let\pessIfVerbose\pessShare
\fi

\ifVerbose
\begin{definition}[pessShare, \citet{babaioff2023fair}]
\label[definition]{defn:pessShare}
Let $1 \le \ell \le d$. Let $\Pi_d(M)$ be the set of all $d$-partitions of $M$.
Then agent $i$'s $\ell$-out-of-$d$ share is defined as
\[ \loodM_i := \!\!\!\!\!\!\max_{\substack{X \in \Pi_d(M):\\v_i(X_j) \le v_i(X_{j+1}) \forall j \in [d-1]}}
    \sum_{j=1}^{\ell} v_i(X_j). \]
Then agent $i$'s pessimistic share is defined as
\[ \pessShare_i := \sup_{1 \le \ell \le d:\;\ell / d \le w_i} \loodM_i. \]
Allocation $A$ is pessShare-fair to agent $i$ if $v_i(A_i) \ge \pessShare_i$.
\end{definition}

\begin{lemma}
\label[lemma]{thm:pess-is-mms}
For any fair division instance with equal entitlements,
the pessShare of any agent is the same as her maximin share.
\end{lemma}
\begin{proof}
Let $\Ical \defeq \fdInst{[n]}{[m]}{(v_i)_{i=1}^n}{w}$ be a fair division instance,
where $w_i = 1/n$ for all $i \in [n]$.
Any agent $i$'s $1$-out-of-$n$ share is the same as $\MMS_i$, so $\pessShare_i \ge \MMS_i$.

Let $\ell/d \le n$ and let $X \in \Pi_d(M)$ such that $v_i(X_j) \le v_i(X_{j+1})$ for all $j \in [d-1]$.
Now let $Y_1$ be the union of the first $\ell$ bundles of $X$,
let $Y_2$ be the union of the next $\ell$ bundles of $X$, and so on.
Add any remaining bundles of $X$ to $Y_n$.
Formally, $Y \in \Pi_n(M)$ where $Y_k \defeq \bigcup_{j=(k-1)\ell+1}^{k\ell} X_j$ for $k \in [n-1]$
and $Y_n \defeq M \setminus \bigcup_{j=1}^{n-1} Y_j$.
Then $Y_n$ contains at least $\ell$ bundles of $X$, since $d \ge \ell n$.
Hence, $v_i(Y_1) \le v_i(Y_2) \le \ldots \le v_i(Y_n)$,
and $v_i(Y_1)$ is agent $i$'s $\ell$-out-of-$d$ share.

By definition of MMS, $v_i(Y_1) \le \MMS_i$.
Hence, for any $\ell$ and $d$ such that $\ell/d \le n$,
agent $i$'s $\ell$-out-of-$d$ share is at most her MMS.
Hence, $\pessShare \le \MMS_i$.
\end{proof}
\fi

\subsection{APS}
\label{sec:notions:aps}

For any $x \in \mathbb{R}^m$ and $S \subseteq [m]$, define $x(S) \defeq \sum_{j \in S} x_j$.

\begin{definition}[APS, \citet{babaioff2023fair}]
\label[definition]{defn:aps}
For a fair division instance $\Ical \defeq \fdInst{[n]}{[m]}{(v_i)_{i=1}^n}{w}$,
agent $i$'s AnyPrice Share (APS) is defined as
\[ \APS_i \defeq \min_{p \in \mathbb{R}^m}\;\max_{S \subseteq [m]: p(S) \le w_ip([m])} v_i(S). \]
Here $p$ is called the \emph{price vector}.
A vector $p^* \in \mathbb{R}^m$ is called an \emph{optimal} price vector for agent $i$ if
\[ p^* \in \argmin_{p \in \mathbb{R}^m}\;\max_{S \subseteq [m]: p(S) \le w_ip([m])} v_i(S). \]
\end{definition}

\Cref{defn:aps} is slightly different from the original definition given in \citet{babaioff2023fair}.
However, they assume that all items are goods, and for that special case,
we can prove that their definition is equivalent to \cref{defn:aps}.
We begin by proving that, \wLoG, we can assume that goods have a non-negative price
and chores have a non-positive price.

\begin{lemma}
\label[lemma]{thm:aps-optimal-price}
Let $\Ical \defeq \fdInst{[n]}{[m]}{(v_i)_{i=1}^n}{w}$ be a fair division instance.
For agent $i$, let $G$ be the set of goods and $C$ be the set of chores,
i.e., $G \defeq \{g \in [m]: v_i(g \mid R) \ge 0 \; \forall R \subseteq [m] \setminus \{g\}\}$
and $C \defeq \{c \in [m]: v_i(c \mid R) \le 0 \; \forall R \subseteq [m] \setminus \{c\}\}$.
Then for some optimal price vector $\phat \in \mathbb{R}^m$, we have
$\phat_g \ge 0$ for all $g \in G$ and $\phat_c \le 0$ for all $c \in C$.
\end{lemma}
\begin{proof}
Let $p^* \in \mathbb{R}^m$ be an optimal price vector.
Let $\Ghat \defeq \{g \in G: p^*_g < 0\}$ and $\Chat \defeq \{c \in C: p^*_c > 0\}$.
The high-level idea is that if we change the price of $\Ghat \cup \Chat$ to 0, we get potentially better prices.
Let $\phat_j \defeq 0$ if $j \in \Ghat \cup \Chat$ and let $\phat_j \defeq p^*_j$ otherwise. Let
\[ \Shat \in \argmax_{S \subseteq [m]: \phat(S) \le w_i\phat([m])} v_i(S) .\]
Since $\phat(\Shat \cup \Ghat \setminus \Chat) = \phat(\Shat)$
and $v_i(\Shat \cup \Ghat \setminus \Chat) \ge v_i(\Shat)$,
we assume without loss of generality that
$\Ghat \subseteq \Shat$ and $\Chat \cap \Shat = \emptyset$.
Moreover, $p^*(\Shat) - w_ip^*([m])
= (\phat(\Shat) - w_i\phat([m])) - w_ip^*(\Chat) - (1-w_i)(-p^*(\Ghat))
\le 0$.
Hence,
\[ \max_{\substack{S \subseteq [m]:\\p^*(S) \le w_ip^*([m])}} \!\!\!\!v_i(S)
    \,\ge\, v_i(\Shat)
    \,= \!\!\!\!\max_{\substack{S \subseteq [m]:\\\phat(S) \le w_i\phat([m])}} \!\!\!\!v_i(S), \]
so $\phat$ is also an optimal price vector.
\end{proof}

For any positive integer $m$, define
$\Delta_m \defeq \{x \in \mathbb{R}^m_{\ge 0}: \sum_{j=1}^m x_j = 1\}$.
When all items are goods, by \cref{thm:aps-optimal-price}, we can assume \wLoG{} that
$p$ is non-negative and $p([m]) = 1$. Hence,
\[ \APS_i = \min_{p \in \Delta_m}\;\max_{S \subseteq [m]: p(S) \le w_i} v_i(S). \]
When all items are chores, by \cref{thm:aps-optimal-price}, we can assume that
$p$ is non-positive and $p([m]) = -1$. Hence,
\[ -\APS_i = \max_{q \in \Delta_m}\;\min_{S \subseteq [m]: q(S) \ge w_i} (-v_i(S)). \]

\citet{babaioff2023fair} give another equivalent definition of APS,
called the \emph{dual} definition.

\begin{definition}[APS (dual)]
\label[definition]{defn:aps-dual}
Let $\Ical \defeq \fdInst{[n]}{[m]}{(v_i)_{i=1}^n}{w}$ be a fair division instance.
For an agent $i$ and any $z \in \mathbb{R}$, let $\Scal_z \defeq \{S \subseteq [m]: v_i(S) \ge z\}$.
Then agent $i$'s AnyPrice share, denoted by $\APS_i$, is the largest value $z$ such that
\[ \exists x \in \mathbb{R}_{\ge 0}^{\Scal_z}, \sum_{S \in \Scal_z} \!x_S = 1
    \textrm{ and } \bigg(\forall j \in [m], \!\!\!\!\sum_{S \in \Scal_z: j \in S} \!\!\!\!x_S = w_i\bigg). \]
\end{definition}

\Cref{defn:aps-dual} can be interpreted as a linear programming relaxation of MMS.
Formally, when entitlements are equal, adding the integrality constraints
$nx_S \in \mathbb{Z}_{\ge 0}$ for all $S \in \Scal_z$ gives us an alternate definition of MMS.

\citet{babaioff2023fair} show that the primal and dual definitions of APS are equivalent.
They prove this only for goods, but their proof can be easily adapted to the case of mixed manna.

\ifVerbose
\begin{lemma}
\label[lemma]{thm:aps-primal-dual-equiv}
\Cref{defn:aps,defn:aps-dual} are equivalent.
\end{lemma}
\begin{proof}
Let $\pAPS_i$ and $\dAPS_i$ be agent $i$'s AnyPrice shares given by
the primal and dual definitions, respectively.
We will show that for any $z \in \mathbb{R}$, $\pAPS_i \ge z$ iff $\dAPS_i \ge z$.
This would prove that $\pAPS_i = \dAPS_i$.

$\dAPS_i \ge z$ iff the following LP has a feasible solution:
\[ \min_{x \in \mathbb{R}_{\ge 0}^{\Scal_z}} 0
\textrm{ where} \!\sum_{S \in \Scal_z} \!x_S = 1
    \textrm{ and } \bigg(\forall j \in [m], \!\!\!\!\sum_{S \in \Scal_z: j \in S} \!\!\!\!x_S = w_i\bigg). \]
Its dual is
\[ \max_{p \in \mathbb{R}^m, r \in \mathbb{R}} r - w_ip([m])
\textrm{ where } p(S) \ge r \textrm{ for all } S \in \Scal_z. \]
The dual LP is feasible since $(0, 0)$ is a solution.
Furthermore, if $(p, r)$ is feasible for the dual LP,
then $(\alpha p, \alpha r)$ is also feasible, for any $\alpha \ge 0$.
Hence, by strong duality of LPs, the primal LP is feasible iff
all feasible solutions to the dual have objective value at most 0.

For a given $p$, the optimal $r$ is $\min_{S \in \Scal_z} p(S)$.
Hence, the dual LP is bounded iff for all $p \in \mathbb{R}^m$,
\[ \min_{S \in \Scal_z} p(S) \le w_ip([m]). \]
Furthermore,
\ifColsTwo\providecommand{\tempWrap}{\\ &\qquad\qquad}\else\let\tempWrap\relax\fi
\begin{align*}
& \forall p \in \mathbb{R}^m, \min_{S \in \Scal_z} p(S) \le w_ip([m])
\\ &\iff \forall p \in \mathbb{R}^m, \exists S \subseteq [m] \textrm{ such that }
    \tempWrap p(S) \le w_ip([m]) \textrm{ and } v_i(S) \ge z
\\ &\iff \left(\min_{p \in \mathbb{R}^m} \,\max_{S \subseteq [m]: p(S) \le w_ip([m])} v_i(S)\right) \ge z
\\ &\iff \pAPS_i \ge z.
\end{align*}
Hence, $\dAPS_i \ge z \iff \pAPS_i \ge z$.
\end{proof}
\fi

\subsection{PROP1}
\label{sec:notions:prop1}

We use the definition of PROP1 from \citet{feige2025low}.
This definition (c.f.~\cref{defn:prop1}) uses strict inequalities, i.e.,
$v_i(A_i \cup \{g\}) > w_iv_i([m])$ instead of $v_i(A_i \cup \{g\}) \ge w_iv_i([m])$,
and $v_i(A_i \setminus \{c\}) > w_iv_i([m])$ instead of $v_i(A_i \setminus \{c\}) \ge w_iv_i([m])$.
This is different from the convention followed by most fair division papers.

We define it this way to make it a slightly stronger notion.
In particular, when the items are identical, the stricter notion aligns with the intuitive idea
that each agent must get at least $\floor{m/n}$ goods or at most $\ceil{m/n}$ chores.
However, outside of the identical items case, this nuance does not affect our results much.
Almost all implications, counterexamples, and (in)feasibility results
hold for both the stronger and weaker definitions,
and we explicitly mention when they don't.
We alter PROPx, PROPm, and PROPavg to also use strict inequalities.

\subsection{PROPx}
\label{sec:notions:propx}

We identify special cases where our definition of PROPx (\cref{defn:propx}) is equivalent to
well-known definitions of PROPx \citep{aziz2020polynomial,li2022almost}.

\begin{lemma}
\label[lemma]{thm:propx-equiv-positive}
In the fair division instance $\Ical \defeq \fdInst{[n]}{[m]}{(v_i)_{i=1}^n}{w}$,
if all marginals are positive for agent $i$
(i.e., $v_i(g \mid R) > 0$ for all $R \subseteq [m]$ and $g \in [m] \setminus R$),
then $A$ is PROPx-fair to agent $i$ iff $v_i(A_i) \ge w_iv_i([m])$ or
$v_i(A_i \cup \{g\}) > w_iv_i([m])$ for some $g \in [m] \setminus A_i$.
\end{lemma}
\begin{proof}[Proof sketch]
In \cref{defn:propx}, we need to find $S \subseteq [m] \setminus A_i$ such that
$v_i(S \mid A_i) > 0$ and $v_i(S \mid A_i)$ is minimized.
Since marginals are positive, we require $|S| = 1$.
\end{proof}

\begin{lemma}
\label[lemma]{thm:propx-equiv-negative}
In the fair division instance $\Ical \defeq \fdInst{[n]}{[m]}{(v_i)_{i=1}^n}{w}$,
if all marginals are negative for agent $i$
(i.e., $v_i(c \mid R) < 0$ for all $R \subseteq [m]$ and $c \in [m] \setminus R$),
then $A$ is PROPx-fair to agent $i$ iff $-v_i(A_i) \le w_i\cdot(-v_i([m]))$ or
$-v_i(A_i \setminus \{c\}) < w_i \cdot (-v_i([m]))$ for some $c \in A_i$.
\end{lemma}
\begin{proof}[Proof sketch]
In \cref{defn:propx}, we need to find $S \subseteq A_i$ such that
$-v_i(S \mid A_i \setminus S) > 0$ and $-v_i(S \mid A_i \setminus S)$ is minimized.
Since marginals are negative, we require $|S| = 1$.
\end{proof}

\begin{lemma}
\label[lemma]{thm:propx-equiv-submod-goods}
In the fair division instance $\Ical \defeq \fdInst{[n]}{[m]}{(v_i)_{i=1}^n}{w}$,
if $v_i$ is submodular and all marginals are non-negative for agent $i$
(i.e., $v_i(g \mid R) \ge 0$ for all $R \subseteq [m]$ and $g \in [m] \setminus R$),
then $A$ is PROPx-fair to agent $i$ iff $v_i(A_i) \ge w_iv_i([m])$ or
$v_i(A_i \cup \{g\}) > w_iv_i([m])$ for some $g \in [m] \setminus A_i$ such that $v_i(g \mid A_i) > 0$.
\end{lemma}
\begin{proof}[Proof sketch]
In \cref{defn:propx}, we need to find $S \subseteq [m] \setminus A_i$ such that
$v_i(S \mid A_i) > 0$ and $v_i(S \mid A_i)$ is minimized.
By the subadditivity of $v_i$'s marginal gains (\cref{thm:submod-mg-subadd}),
$v_i(S \mid A_i) > 0$ tells us that $\exists \ghat \in S$ such that $v_i(\ghat \mid A_i) > 0$.
Since $v_i(S \mid A_i) \ge v_i(\ghat \mid A_i)$, we can replace $S$ by $\{\ghat\}$,
so we can assume \wLoG{} that $|S| = 1$.
\end{proof}

\begin{lemma}
\label[lemma]{thm:propx-equiv-submod-chores}
In the fair division instance $\Ical \defeq \fdInst{[n]}{[m]}{(v_i)_{i=1}^n}{w}$,
if $v_i$ is submodular and all marginals are non-positive for agent $i$
(i.e., $v_i(c \mid R) \le 0$ for all $R \subseteq [m]$ and $c \in [m] \setminus R$),
then $A$ is PROPx-fair to agent $i$ iff $-v_i(A_i) \le w_i \cdot (-v_i([m]))$ or
$-v_i(A_i \setminus \{c\}) < w_i \cdot (-v_i([m]))$ for some $c \in A_i$
such that $v_i(c \mid A_i \setminus \{c\}) < 0$.
\end{lemma}
\begin{proof}[Proof sketch]
In \cref{defn:propx}, we need to find $S \subseteq A_i$ such that
$-v_i(S \mid A_i \setminus S) > 0$ and $-v_i(S \mid A_i \setminus S)$ is minimized.
By the superadditivity of $v_i$'s marginal losses (\cref{thm:submod-ml-superadd}),
$-v_i(S \mid A_i \setminus S) > 0$ tells us that $\exists \chat \in S$ such that
$-v_i(\chat \mid A_i \setminus \{\chat\}) > 0$.
Since $-v_i(S \mid A_i \setminus S) \ge -v_i(\chat \mid A_i \setminus \{\chat\})$,
we can replace $S$ by $\{\chat\}$, so we can assume \wLoG{} that $|S| = 1$.
\end{proof}

\subsection{PROPm and PROPavg}
\label{sec:notions:propm-propavg}

\begin{definition}[PROPm and PROPavg]
\label[definition]{defn:propm}
\label[definition]{defn:propavg}
Let $A$ be an allocation for the fair division instance $\fdInst{[n]}{[m]}{(v_i)_{i=1}^n}{w}$.
Fix an agent $i$. For every $j \in [n] \setminus \{i\}$,
let $K_j \defeq \{v_i(S \mid A_i): S \subseteq A_j \text{ and } v_i(S \mid A_i) > 0\}$.
Let $\tau_j \defeq 0$ if $K_j = \emptyset$ and $\tau_j \defeq \min(K_j)$ otherwise.
Let $T \defeq \{\tau_j: j \in [n] \setminus \{i\} \textrm{ and } \tau_j > 0\}$.
Define conditions C1, C2, C3, and C4 as follows:
\begin{tightenum}
\item[(C1)]\label{item:propm:prop}$v_i(A_i) \ge w_i \cdot v_i([m])$.
\item[(C2)]\label{item:propm:chores}
    $v_i(A_i \setminus S) > w_i \cdot v_i([m])$ for every $S \subseteq A_i$
    such that $v_i(S \mid A_i \setminus S) < 0$.
\item[(C3)]\label{item:propm:goods}
    $T = \emptyset$, or $v_i(A_i) + \max(T) > w_i \cdot v_i([m])$.
\item[(C4)]\label{item:propavg}
    $T = \emptyset$, or $v_i(A_i) + \frac{1}{|T|}\sum_{\tau_j \in T}\tau_j > w_i \cdot v_i([m])$.
\end{tightenum}
$A$ is PROPm-fair to $i$ if either C1 holds or both C2 and C3 hold.
$A$ is PROPavg-fair to $i$ if either C1 holds or both C2 and C4 hold.
\end{definition}

PROPm was originally defined by \citet{baklanov2021achieving},
and PROPavg was originally defined in \citet{kobayashi2025proportional},
both for the additive goods setting.
Note that PROPm, PROPavg, and PROPx are all equivalent to each other for chores.

The definitions of PROPm and PROPavg simplify under submodular valuations over goods and over chores.
The proof is similar to that of PROPx (\cref{thm:propx-equiv-submod-goods,thm:propx-equiv-submod-chores})
and is based on \cref{thm:submod-mg-subadd,thm:submod-ml-superadd}, so we omit the details here.

The original definitions of PROPm and PROPavg have minor errors, which our definitions fix.
We now explain these errors and motivate why we decided to define them the way have done.

\subsubsection{Defining PROPm}
\label{sec:notions:propm}

PROPm was first defined in \citet{baklanov2021achieving}
for equal entitlements and goods with additive valuations.
Moreover, they claimed that PROPx implies PROPm and PROPm implies PROP1.

According to \citet{baklanov2021achieving}, when dividing a set $[m]$ of goods among $n$ agents,
$A$ is PROPm-fair to agent $i$ if $v_i(A_i) + \max_{j \neq i} m_i(A_j) \ge v_i([m])/n$,
where $m_i(S) \defeq \min_{g \in S} v_i(g)$.
However, \citet{baklanov2021achieving} do not explicitly state what $m_i(\emptyset)$ is.
The well-known convention of $\min(\emptyset) = \infty$ leads to the strange phenomenon
where every agent is PROPm-satisfied if two agents receive no goods
(whereas PROP1 is not guaranteed).
One way to fix this is to only consider agents with non-empty bundles,
i.e., we say that $i$ is PROPm-satisfied by $A$ if
$v_i(A_i) + \max_{j \in [n] \setminus \{i\}: A_j \neq \emptyset} m_i(A_j) \ge v_i([m])/n$.
This does not give a satisfactory definition when $A_j = \emptyset$
for all $j \in N \setminus \{i\}$ (assuming $\max(\emptyset) = -\infty$),
so we define $A$ to be PROPm-fair to $i$ for that edge case.

The above idea, along with extending the definition to mixed manna
in the same way as PROPx (\cref{defn:propx}),
gives us our definition of PROPm (\cref{defn:propm}).
Moreover, the following two example instances
(having 3 equally-entitled agents with identical additive valuations)
guided our definition of PROPm.

\begin{enumerate}
\item Consider three goods of values 100, 10, and 1.
    Intuitively, each agent should get 1 good each, and that should be considered fair.
\item Consider 5 items of values $-100$, $-100$, $-100$, $10$, $1$.
    Intuitively, allocation ($\{-100, 10, 1\}$, $\{-100\}$, $\{-100\}$) should not be fair,
    and allocation ($\{-100, 10\}$, $\{-100, 1\}$, $\{-100\}$) should be fair.
    In both allocations, removing a chore makes an agent PROP-satisfied, so just
    satisfying this condition is not enough. We also need to look at the goods.
\end{enumerate}

For mixed manna, \citet{livanos2022almost} define a notion called PropMX,
but that definition is weak: when all items are goods, every allocation is trivially PropMX.

\citet{baklanov2021propm} showed that for equal entitlements and goods with additive valuations,
a PROPm allocation always exists and can be computed in polynomial time.
It can be verified that their result also works for our definition of PROPm (\cref{defn:propm}).

\subsubsection{Defining PROPavg}
\label{sec:notions:propavg}

PROPavg was first defined by \citet{kobayashi2025proportional} for additive goods.
Moreover, they claimed that PROPx implies PROPavg and PROPavg implies PROPm.
However, their former claim is incorrect for their definition.

Their definition of PROPavg is equivalent to ours for equal entitlements and additive goods,
except that in condition C4, they use $\Sum(T)/(n-1)$ instead of $\Sum(T)/|T|$.
Consider a fair division instance with three equally-entitled agents and two identical goods.
Let $A$ be an allocation where agent 1 gets all goods and the other agents get nothing.
This allocation is PROPx and PROPm, but is not PROPavg by \citet{kobayashi2025proportional}'s definition.

We believe that finding a feasible fairness notion in between PROPx and PROPm is well-motivated,
so we slightly change \citet{kobayashi2025proportional}'s definition to ensure that PROPx implies PROPavg.
Also, for additive goods, when no agent gets an empty bundle in an allocation,
then $|T| = n-1$, so both definitions are the same.

\section{Proofs of Implications}
\label{sec:impls-extra}

\subsection{Among Derived Notions}
\label{sec:impls-extra:among-derived}

\begin{remark}
\label[remark]{thm:impl:epistemic}
For any fairness notion $F$, if an allocation is $F$-fair to an agent $i$,
then it is also epistemic-$F$-fair to agent $i$.
If an allocation is epistemic-$F$-fair to an agent $i$,
then it is also minimum-$F$-share-fair to agent $i$.
If there are only two agents, then an allocation is epistemic-$F$-fair to an agent $i$
iff it is $F$-fair to agent $i$.
\end{remark}

\begin{remark}
\label[remark]{thm:impl:groupwise}
For any fairness notion $F$, if an allocation is groupwise-$F$-fair to an agent $i$,
then it is also pairwise-$F$-fair to agent $i$ and $F$-fair to agent $i$.
When there are only two agents, all three of these notions are equivalent.
\end{remark}

\begin{lemma}
\label[lemma]{thm:impl:ext-to-epistemic}
Let $\Omega$ be a set containing pairs of the form $(\Ical, A)$,
where $\Ical$ is a fair division instance and $A$ is an allocation for $\Ical$.
For any two fairness notions $F_1$ and $F_2$, let $\phi(F_1, F_2)$ be the proposition
``\,$\forall (\Ical, A) \in \Omega$, for every agent $i$ in $\Ical$,
$A$ is $F_2$-fair to $i$ whenever $A$ is $F_1$-fair to $i$". Then
$\phi(F_1, F_2)$ implies $\phi(\textrm{epistemic-}F_1, \textrm{epistemic-}F_2)$
and $\phi(\textrm{min-}F_1\textrm{-share}, \textrm{min-}F_2\textrm{-share})$.
\end{lemma}
\begin{proof}[Proof sketch]
If $\phi(F_1, F_2)$, then an epistemic-$F_1$-certificate is also an epistemic-$F_2$-certificate,
and a min-$F_1$-share-certificate is also a min-$F_2$-share-certificate.
\end{proof}

\subsection{Among EF, EFX, EF1}
\label{sec:impls-extra:among-ef-efx-ef1}

Here we look at implications among EF, EFX, EF1, and their epistemic variants.

\begin{remark}[EF $\fimplies$ EFX+EF1]
\label[remark]{thm:impl:ef-to-efx+ef1}
If an allocation is EF-fair to agent $i$, then it is also EFX-fair and EF1-fair to $i$.
\end{remark}

Because of how we define EFX (\cref{defn:efx}), EFX does not always imply EF1
(see \cref{cex:gaps-not-ef1-prop1-subadd} for a counterexample with subadditive valuations).
However, EFX $\fimplies$ EF1 for many common settings, as the following lemma shows.

\begin{lemma}[EFX $\fimplies$ EF1]
\label[lemma]{thm:impl:efx-to-ef1}
For the fair division instance $\fdInst{[n]}{[m]}{(v_i)_{i=1}^n}{w}$,
let allocation $A$ be EFX-fair to agent $i$.
Then agent $i$ is EF1-satisfied in these scenarios:
\begin{tightenum}
\item $v_i$ is additive.
\item $v_i$ is doubly strictly monotone ($\exists$ a partition $(G, C)$ of $[m]$ s.t.
    $v_i(g|\cdot) > 0$ $\forall g \in G$, and $v_i(c|\cdot) < 0$ $\forall c \in C$).
\item Agents have equal entitlements, $v_i$ is submodular, and all items are goods for agent $i$.
\item $v_i$ is submodular and all items are chores for agent $i$.
\end{tightenum}
\end{lemma}
\begin{proof}
Suppose $i$ is EFX-satisfied but EF1-envies agent $j$.
Since $i$ EF1-envies $j$, for all $t \in A_j$, we have
\[ \frac{v_i(A_i)}{w_i} < \frac{v_i(A_j \setminus \{t\})}{w_j}. \]
If $v_i(t \mid A_i) > 0$ for some $t \in A_j$, then $i$ would EFX-envy $j$,
so we get $v_i(t \mid A_i) \le 0$ for all $t \in A_j$.

Since $i$ EF1-envies $j$, for all $t \in A_i$, we have
\[ \frac{v_i(A_i \setminus \{t\})}{w_i} < \frac{v_i(A_j)}{w_j}. \]
If $v_i(t \mid A_i \setminus \{t\}) < 0$ for some $t \in A_i$, then $i$ would EFX-envy $j$,
so we get $v_i(t \mid A_i \setminus \{t\}) \ge 0$ for all $t \in A_i$.

If $v_i$ is additive, we get $v_i(A_i) \ge 0 \ge v_i(A_j)$, which is a contradiction.
If $v_i$ is doubly strictly monotone, then
all items in $A_j$ are chores and all items in $A_i$ are goods.
Hence, $v_i(A_i) \ge 0 \ge v_i(A_j)$, which is a contradiction.

Suppose agents have equal entitlements, all items are goods for $i$, and $v_i$ is submodular.
By the subadditivity of marginal gains (\cref{thm:submod-mg-subadd}), we get
\[ v_i(A_j \mid A_i) \le \sum_{g \in A_j} v_i(g \mid A_i) \le 0. \]
Hence, $v_i(A_j) \le v_i(A_i \cup A_j) = v_i(A_i) + v_i(A_j \mid A_i) \le v_i(A_i)$,
which is a contradiction.

Suppose all items are chores for $i$ and $v_i$ is submodular.
By the superadditivity of marginal losses (\cref{thm:submod-ml-superadd}),
\[ v_i(A_i) \ge \sum_{c \in A_i} v_i(c \mid A_i \setminus \{c\}) \ge 0. \]
Hence, $v_i(A_i) \ge 0 \ge v_i(A_j)$, which is a contradiction.

A contradiction implies that it is impossible for agent $i$ to be
EFX-satisfied but not EF1-satisfied.
\end{proof}

\begin{lemma}[MXS $\fimplies$ EF1 for $n=2$]
\label[lemma]{thm:impl:mxs-to-ef1-n2}
Let $\fdInst{[2]}{[m]}{(v_1, v_2)}{w}$ be a fair division instance with indivisible items.
If $v_1$ is additive and agent 1 is MXS-satisfied by allocation $A$,
then agent 1 is also EF1-satisfied by $A$.
\end{lemma}
\begin{proof}
Suppose $A$ is MXS-fair to agent 1 but not EF1-fair to her.
Then agent 1 envies agent 2 in $A$, so $v_1(A_1) < v_1(A_2)$.
Let $B$ be agent 1's MXS-certificate for $A$. Then $v_1(B_1) \le v_1(A_1)$.
Moreover, $v_1(A_2) = v_1([m]) - v_1(A_1) \le v_1([m]) - v_1(B_1) = v_1(B_2)$.
Hence, we get $v_1(B_1) \le v_1(A_1) < v_1(A_2) \le v_1(B_2)$.

Let $G \defeq \{g \in [m]: v_1(g) > 0\}$ and $C \defeq \{c \in [m]: v_1(c) < 0\}$.
Let $\max(\emptyset) \defeq -\infty$ and $\min(\emptyset) \defeq \infty$.

Since agent 1 is EFX-satisfied by $B$ and not EF1-satisfied by $A$,
for every $\ghat \in A_2$, we get
\begin{align*}
& \frac{v_1(A_2) - v_1(\ghat)}{w_2} > \frac{v_1(A_1)}{w_1}
\ge \frac{v_1(B_1)}{w_1}
\wrapIfTwoCols \ge \frac{1}{w_2}\left(v_1(B_2) - \min_{g \in B_2 \cap G} v_1(g)\right)
\\ &\ge \frac{1}{w_2}\left(v_1(A_2) - \min_{g \in B_2 \cap G} v_1(g)\right).
\end{align*}
Hence, for every $\ghat \in A_2$, we get $v_1(\ghat) < \min_{g \in B_2 \cap G} v_1(g)$.
Hence, $A_2 \cap G$ and $B_2 \cap G$ are disjoint, so $A_2 \cap G \subseteq B_1 \cap G$.
Let $d_i \defeq -v_i$ for all $i$. Then for every $\chat \in A_1$, we get
\begin{align*}
& \frac{d_1(A_1) - d_1(\chat)}{w_1} > \frac{d_1(A_2)}{w_2}
\ge \frac{d_1(B_2)}{w_2}
\wrapIfTwoCols \ge \frac{1}{w_1}\left(d_1(B_1) - \min_{c \in B_1 \cap C} d_1(c)\right)
\\ &\ge \frac{1}{w_1}\left(d_1(A_1) - \min_{c \in B_1 \cap C} d_1(c)\right).
\end{align*}
Hence, for every $\chat \in A_1$, we have $d_1(\chat) < \min_{c \in B_1 \cap C} d_1(c)$.
Hence, $A_1 \cap C$ and $B_1 \cap C$ are disjoint, so $B_1 \cap C \subseteq A_2 \cap C$.
Hence,
\begin{align*}
v_1(A_2) &= v_1(A_2 \cap G) - d_1(A_2 \cap C)
\wrapIfTwoCols \le v_1(B_1 \cap G) - d_1(B_1 \cap C) = v_1(B_1),
\end{align*}
which is a contradiction.
Hence, it's not possible for $A$ to be MXS-fair to agent 1 but not EF1-fair to her.
\end{proof}

We will now show that every MXS allocation is also M1S
when marginal values are \emph{doubly monotone}.
We first state a simple observation on how EFX and EF1 are related.
Then we use that as the basis for an iterative algorithm
for computing an agent's M1S certificate using their MXS certificate.

\begin{observation}
\label[observation]{thm:special-efx-is-ef1}
Let $\fdInst{[n]}{[m]}{(v_i)_{i=1}^n}{w}$ be a fair division instance
where $v_i$ is $(G, C)$-doubly-monotone for some agent $i$ (i.e., $[m] = G \cup C$,
$v_i(g \mid \cdot) \ge 0$ $\forall g \in G$, and $v_i(c \mid \cdot) \le 0$ $\forall c \in C$).
If all of the following hold for an allocation $A$, then $A$ is EF1-fair to agent $i$.
\begin{tightenum}
\item $A$ is EFX-fair to agent $i$.
\item $v_i(g \mid A_i) > 0$ for all $g \in G \setminus A_i$.
\item $v_i(c \mid A_i \setminus \{c\}) < 0$ for all $c \in A_i \cap C$.
\end{tightenum}
\end{observation}

\begin{algorithm}[htb]
\caption{$\improve_i(Z)$: Here $Z$ is an allocation for the fair division instance
$([n], G \cup C, (v_j)_{j=1}^n, w)$.}
\label{algo:mxs-to-m1s}
\begin{algorithmic}[1]
\While{\texttt{true}}
    \If{$\exists j \in [n] \setminus \{i\}$ and $\exists g \in G \cap Z_j$ s.t. $v_i(g \mid Z_i) = 0$}
        \State Set $Z_i = Z_i \cup \{g\}$ and $Z_j = Z_j \setminus \{g\}$.
    \ElsIf{$\exists c \in Z_i \cap C$ s.t. $v_i(c \mid Z_i \setminus \{c\}) = 0$}
        \State Set $Z_i = Z_i \setminus \{c\}$.
        \State Pick any $j \in [n] \setminus \{i\}$. Set $Z_j = Z_j \cup \{c\}$.
    \Else\label{alg-line:mxs-to-m1s:else}
        \State \label{alg-line:mxs-to-m1s:return}\Return $Z$.
    \EndIf
\EndWhile
\end{algorithmic}
\end{algorithm}

\begin{lemma}
\label[lemma]{thm:impl:mxs-to-m1s}
Let $\fdInst{[n]}{[m]}{(v_i)_{i=1}^n}{w}$ be a fair division instance
where $v_i$ is $(G, C)$-doubly-monotone for some agent $i$ (i.e., $[m] = G \cup C$,
$v_i(g \mid \cdot) \ge 0$ $\forall g \in G$, and $v_i(c \mid \cdot) \le 0$ $\forall c \in C$).
If an allocation $A$ is MXS-fair to agent $i$,
then it is also M1S-fair to agent $i$ if $n=2$ or $G = \emptyset$.
\end{lemma}
\begin{proof}
\WLoG{}, let $i=1$. Let $B$ be agent 1's MXS-certificate for $A$.
Then $B$ is EFX-fair to agent 1 and $v_1(B_1) \le v_1(A_1)$.
Define $\Bhat \defeq \improve_1(B)$, where $\improve_1$ is defined in \cref{algo:mxs-to-m1s}.

In any iteration of the while loop in \cref{algo:mxs-to-m1s},
let $X$ be the allocation at the beginning of the iteration.
Either $X$ is returned at line \ref{alg-line:mxs-to-m1s:return},
or the allocation gets modified to $Y$, and (as we will prove soon) the following invariants hold:
\begin{tightenum}
\item If $X$ is EFX-fair to agent 1, then $Y$ is also EFX-fair to agent 1.
\item $|Y_1 \cap G| - |Y_1 \cap C| > |X_1 \cap G| - |X_1 \cap C|$.
\item $v_1(Y_1) = v_1(X_1)$.
\end{tightenum}
The first invariant ensures that $\Bhat$ is also EFX-fair to agent 1.
The second invariant ensures that $\improve_1$ terminates.
Then using \cref{thm:special-efx-is-ef1}, we get that $\Bhat$ is EF1-fair to agent 1.
The third invariant ensures that $v_1(\Bhat_1) = v_1(B_1) \le v_1(A_1)$,
which makes $\Bhat$ an M1S-certificate for allocation $A$ and agent 1.
The second and third invariants are trivial to show.
We now prove the first invariant if $n = 2$ or $G = \emptyset$.

Suppose $Y_1 = X_1 \cup \{g\}$ and $Y_j = X_j \setminus \{g\}$,
where $g \in G \cap X_j$ such that $v_1(g \mid X_1) = 0$.
Since $G \neq \emptyset$, we have $n = j = 2$.
Pick any $S \subseteq Y_2$ such that $v_1(S \mid Y_1) > 0$. Then
\begin{align*}
& v_1(S \cup \{g\} \mid X_1)
\wrapIfTwoCols = v_1(S \cup \{g\} \cup X_1) - v_1(X_1)
\\ &= (v_1(S \cup Y_1) - v_1(Y_1)) + (v_1(Y_1) - v_1(X_1))
\wrapIfTwoCols = v_1(S \mid Y_1) > 0.
\end{align*}
Since $X$ is EFX-fair to agent 1, we get
\[ \frac{v_1(Y_1)}{w_1} = \frac{v_1(X_1)}{w_1}
    \ge \frac{v_1(X_2 \setminus (S \cup \{g\}))}{w_2}
    = \frac{v_1(Y_2 \setminus S)}{w_j}, \]
so 1 doesn't EFX-envy agent $j$. Hence, $Y$ is EFX-fair to 1.

Suppose $Y_1 = X_1 \setminus \{c\}$ and $Y_j = X_j \cup \{c\}$,
where $c \in X_1 \cap C$ such that $v_1(c \mid X_1 \setminus \{c\}) = 0$.
Pick any $S \subseteq Y_1$ such that $v_1(S \mid Y_1 \setminus S) < 0$. Then
\begin{align*}
& v_1(S \cup \{c\} \mid X_1 \setminus (S \cup \{c\}))
\wrapIfTwoCols = v_1(X_1) - v_1(X_1 \setminus (S \cup \{c\}))
\\ &= (v_1(X_1) - v_1(Y_1)) + (v_1(Y_1) - v_1(Y_1 \setminus S))
\wrapIfTwoCols = v_1(S \mid Y_1 \setminus S) < 0.
\end{align*}
Since $X$ is EFX-fair to agent 1, we get that
for any other agent $k \in [n] \setminus \{1\}$,
\[ \frac{v_1(Y_1 \setminus S)}{w_1} = \frac{v_1(X_1 \setminus (S \cup \{c\}))}{w_1}
    \ge \frac{v_1(X_k)}{w_k} \ge \frac{v_1(Y_k)}{w_k}. \]
Hence, agent 1 is EFX-satisfied with $Y$.

Hence, invariant 1 holds, so $\Bhat$ is agent 1's M1S-certificate for allocation $A$.
\end{proof}

\subsection{Among PROP-Based Notions}
\label{sec:impls-extra:among-prop-based}

\begin{lemma}[PROPx $\fimplies$ PROPavg]
\label[lemma]{thm:impl:propx-to-propavg}
In a fair division instance $\fdInst{[n]}{[m]}{(v_i)_{i=1}^n}{w}$,
if an allocation is PROPx-fair to agent $i$, then it is also PROPavg-fair to agent $i$.
\end{lemma}
\begin{proof}
Assume (for the sake of contradiction) that there is an allocation $A$ where
agent $i$ is PROPx-satisfied but not PROPavg-satisfied.
Since $i$ is not PROPavg-satisfied, we get $v_i(A_i) \le w_iv_i([m])$.
Since $i$ is PROPx-satisfied, we get

\begin{itemize}
\item $v_i(A_i \setminus S) > w_iv_i([m])$ for all $S \subseteq A_i$ such that
    $v_i(S \mid A_i \setminus S) < 0$.
\item $v_i(A_i \cup S) > w_iv_i([m])$ for all $S \subseteq [m] \setminus A_i$
    such that $v_i(S \mid A_i) > 0$.
\end{itemize}

Since $i$ is not PROPavg-satisfied, we get that $T \neq \emptyset$
and $v_i(A_i) + \Sum(T)/|T| \le w_iv_i([m])$ ($T$ is defined in \cref{defn:propavg}).
For each $\tau_j \in T$, we have $0 < \tau_j = v_i(S_j \mid A_i)$ for some $S_j \subseteq A_j$.
Since $A$ is PROPx-fair to $i$, we have $v_i(A_i) + \tau_j > w_iv_i([m])$ for all $\tau_j \in T$.
On averaging these inequalities, we get $v_i(A_i) + \Sum(T)/|T| > w_iv_i([m])$,
which implies that $A$ is PROPavg-satisfied, which is a contradiction.
Hence, if $i$ is PROPx-satisfied by $A$, then she is also PROPavg-satisfied by $A$.
\end{proof}

\begin{lemma}[PROPm $\fimplies$ PROP1]
\label[lemma]{thm:impl:propm-to-prop1}
For a fair division instance $\fdInst{[n]}{[m]}{(v_i)_{i=1}^n}{w}$,
if an allocation $A$ is PROPm-fair to agent $i$, then it is also PROP1-fair to agent $i$
if at least one of these conditions holds:
\begin{tightenum}
\item $v_i$ is submodular.
\item $v_i$ is doubly strictly monotone ($\exists$ a partition $(G, C)$ of $[m]$ s.t.
    $v_i(g|\cdot) > 0$ $\forall g \in G$, and $v_i(c|\cdot) < 0$ $\forall c \in C$).
\end{tightenum}
\end{lemma}
\begin{proof}
Suppose allocation $A$ is PROPm-fair to $i$ but not PROP1-fair to $i$. Then
\begin{tightenum}
\item\label{item:impl:propm-to-prop1:unprop}$v_i(A_i) < w_iv_i([m])$ (not PROP1).
\item\label{item:impl:propm-to-prop1:unprop1-chores}$v_i(A_i \setminus \{c\}) \le w_iv_i([m])$
    for all $c \in A_i$ (not PROP1).
\item\label{item:impl:propm-to-prop1:unprop1-goods}$v_i(A_i \cup \{g\}) \le w_iv_i([m])$
    for all $g \not\in A_i$ (not PROP1).
\item\label{item:impl:propm-to-prop1:propm-chores}$v_i(A_i \setminus \{c\}) > w_iv_i([m])$
    for all $c \in A_i$ such that $v_i(c \mid A_i \setminus \{c\}) < 0$ (by PROPm fairness).
\item\label{item:impl:propm-to-prop1:propm-goods}$T = \emptyset$ or $v_i(A_i) + \max(T) > w_iv_i([m])$
    (by PROPm fairness; c.f.~\cref{defn:propm} for the definition of $T$).
\end{tightenum}

By \ref{item:impl:propm-to-prop1:unprop1-chores} and \ref{item:impl:propm-to-prop1:propm-chores},
we get $v_i(c \mid A_i \setminus \{c\}) \ge 0$ for all $c \in A_i$.
If $v_i$ is doubly strictly monotone, then $A_i$ only contains goods, so $v_i(A_i) \ge 0$.
If $v_i$ is submodular, then by superadditivity of $v_i$'s marginal loss (\cref{thm:submod-ml-superadd}), we get
\[ v_i(A_i) \ge \sum_{c \in A_i} v_i(c \mid A_i \setminus \{c\}) \ge 0. \]

Suppose $T = \emptyset$. Then $\tau_j = 0$ for all $j \in [n] \setminus \{i\}$.
Hence, for all $j \in [n] \setminus \{i\}$, we have $v_i(A_j \mid A_i) \le 0$.
If $v_i$ is doubly strictly monotone, then $[m] \setminus A_i$ contains only chores,
so $v_i([m] \setminus A_i \mid A_i) \le 0$. If $v_i$ is submodular, then
$v_i(\cdot \mid A_i)$ is subadditive by \cref{thm:submod-mg-subadd}, so
\[ v_i([m] \setminus A_i \mid A_i)
    \le \sum_{j \in [n] \setminus \{i\}} v_i(A_j \mid A_i) \le 0. \]
Hence, $v_i(A_i) \ge v_i([m])$.
If $v_i([m]) \le 0$, then $v_i(A_i) \ge 0 \ge w_iv_i([m])$,
and if $v_i([m]) \ge 0$, then $v_i(A_i) \ge v_i([m]) \ge w_iv_i([m])$.
This contradicts \ref{item:impl:propm-to-prop1:unprop}, so $T \neq \emptyset$.

Let $\max(T) = \tau_{\jhat} > 0$. By definition of $\tau_{\jhat}$, we get
\begin{align*}
0 < \tau_{\jhat} &= \min(\{v_i(S \mid A_i) \mid S \subseteq A_{\jhat} \textrm{ and } v_i(S \mid A_i) > 0\})
\\ &\le \min(\{v_i(g \mid A_i) \mid g \in A_{\jhat} \textrm{ and } v_i(g \mid A_i) > 0\}).
\end{align*}

\textbf{Case 1}: $v_i(g \mid A_i) \le 0$ for all $g \in A_{\jhat}$.
\\ If $v_i$ is doubly strictly monotone, then $A_{\jhat}$ only has chores,
so $v_i(S \mid A_i) \le 0$ for all $S \subseteq A_{\jhat}$.
This contradicts $\tau_{\jhat} > 0$.
Now let $v_i$ be submodular.
Since $v_i(\cdot \mid A_i)$ is subadditive by \cref{thm:submod-mg-subadd},
for any $S \subseteq A_{\jhat}$, we get
$v_i(S \mid A_i) \le \sum_{c \in S} v_i(c \mid A_i) \le 0$.
This contradicts $\tau_{\jhat} > 0$.

\textbf{Case 2}: $v_i(\ghat \mid A_i) > 0$ for some $\ghat \in A_{\jhat}$.
\\ Then $\max(T) = \tau_{\jhat} \le v_i(\ghat \mid A_i)$.
By \ref{item:impl:propm-to-prop1:propm-goods}, we get
$w_iv_i([m]) < v_i(A_i) + \max(T) \le v_i(A_i \cup \{\ghat\})$.
But this contradicts \ref{item:impl:propm-to-prop1:unprop1-goods}.

Hence, it cannot happen that $i$ is PROPm-satisfied by $A$ but not PROP1-satisfied.
\end{proof}

\subsection{EF vs PROP}
\label{sec:impls-extra:ef-vs-prop}

Here we look at implications between EF (and its epistemic variants) and PROP.

\begin{lemma}[MEFS $\fimplies$ PROP, \citet{bouveret2016characterizing}]
\label[lemma]{thm:impl:mefs-to-prop}
For a fair division instance $\fdInst{[n]}{[m]}{(v_i)_{i=1}^n}{w}$,
if $v_i$ is subadditive and an allocation $A$ is MEFS-fair to $i$, then $A$ is also PROP-fair to $i$.
\end{lemma}
\begin{proof}
Let $B$ be agent $i$'s MEFS-certificate for $A$.
Then for all $j \in [n]$, we have $v_i(B_i)/w_i \ge v_i(B_j)/w_j$.
Sum these inequalities over all $j \in [n]$, weighting each by $w_j$,
to get $v_i(B_i)/w_i \ge \sum_{j=1}^n v_i(B_j)$.
Since $v_i$ is subadditive, we get $v_i([m]) \le \sum_{j=1}^n v_i(B_j)$.
Hence,
\[ \frac{v_i(A_i)}{w_i} \ge \frac{v_i(B_i)}{w_i} \ge \sum_{j=1}^n v_i(B_j) \ge v_i([m]).
\qedhere \]
\end{proof}

\begin{lemma}[EF $\fimplies$ GPROP, \citet{bouveret2016characterizing}]
\label[lemma]{thm:impl:ef-to-gprop}
For a fair division instance $\fdInst{[n]}{[m]}{(v_i)_{i=1}^n}{w}$,
if $v_i$ is subadditive and agent $i$ is envy-free in $A$,
then $A$ is groupwise-PROP-fair to $i$.
\end{lemma}
\begin{proof}
Let $S$ be a subset of agents containing $i$.
Let $\Ahat$ be the allocation obtained by restricting $A$ to $S$ (c.f.~\cref{defn:restricting}).
Then $i$ is also envy-free in $\Ahat$.
$\Ahat$ is also MEFS-fair to $i$, since $\Ahat$ is its own MEFS-certificate.
By \cref{thm:impl:mefs-to-prop}, agent $i$ is PROP-satisfied by $\Ahat$.
Since this is true for all $S$ containing $i$,
we get that $A$ is groupwise-PROP-fair to agent $i$.
\end{proof}

\begin{lemma}[PROP $\fimplies$ EF for idval]
\label[lemma]{thm:impl:prop-to-ef-superadd-id}
In a fair division instance $\fdInst{[n]}{[m]}{(v_i)_{i=1}^n}{w}$ with identical superadditive valuations,
a PROP allocation is also an EF allocation.
\end{lemma}
\begin{proof}
Let $v$ be the common valuation function. Let $A$ be a PROP allocation.
Then $v(A_i) \ge w_iv([m])$ for each agent $i$.
Suppose $v(A_k) > w_kv([m])$ for some agent $k$.
Sum these inequalities to get $\sum_{i=1}^n v(A_i) > v([m])$.
This contradicts superadditivity of $v$, so $v(A_i) = w_iv([m])$ for each agent $i$.
Hence, $v(A_i)/w_i = v([m])$ for all $i$, so $A$ is EF.
\end{proof}

\begin{lemma}[PROP $\fimplies$ EF for $n=2$]
\label[lemma]{thm:impl:prop-to-ef-n2}
In the fair division instance $\fdInst{[2]}{[m]}{(v_1, v_2)}{w}$, for some agent $i$,
if $v_i$ is superadditive and agent $i$ is PROP-satisfied by allocation $A$,
then she is also envy-free in $A$.
\end{lemma}
\begin{proof}
Assume $i=1$ \wLoG. Then
\begin{align*}
& \frac{v_1(A_1)}{w_1} \ge v_1([m]) \ge v_1(A_1) + v_1(A_2)
\wrapIfTwoCols \implies v_1(A_2) \le v_1(A_1)\left(\frac{1}{w_1} - 1\right) = w_2 \frac{v_1(A_1)}{w_1}.
\end{align*}
Hence, agent 1 does not envy agent 2.
\end{proof}

\begin{lemma}[PPROP $\fimplies$ GPROP]
\label[lemma]{thm:impl:pprop-to-gprop}
In instance $\fdInst{[n]}{[m]}{(v_i)_{i=1}^n}{w}$,
if $v_i$ is submodular and allocation $A$ is PPROP-fair to agent $i$,
then $A$ is also GPROP-fair to $i$.
\end{lemma}
\begin{proof}
Since $A$ is PPROP-fair to $i$, for all $j \in [n] \setminus \{i\}$,
\begin{align*}
& v_i(A_i) \ge \frac{w_i}{w_i + w_j}v_i(A_i \cup A_j)
\wrapIfTwoCols \implies (w_i+w_j)v_i(A_i) \ge w_i(v_i(A_j \mid A_i) + v_i(A_i))
\\ &\implies w_j v_i(A_i) \ge w_i v_i(A_j \mid A_i).
\end{align*}
Sum these inequalities over all $j \in [n] \setminus \{i\}$ to get
\[ (1-w_i) v_i(A_i) \ge w_i \sum_{j \neq i} v_i(A_j \mid A_i). \]
By \cref{thm:submod-mg-subadd}, $v_i(\cdot \mid A_i)$ is subadditive. Hence,
\[ (1-w_i) v_i(A_i) \ge w_i \sum_{j \neq i} v_i(A_j \mid A_i) \ge w_i v_i([m] \setminus A_i \mid A_i)
\implies v_i(A_i) \ge w_i v_i([m]). \]
So, $A$ is PROP to $i$.
This also works if we restrict the allocation to a subset of agents,
so $A$ is also GPROP to $i$.
\end{proof}

\subsection{EFX, EF1 vs PROPx, PROPm, PROP1}
\label{sec:impls-extra:efx-ef1-vs-propx-propm-prop1}

\begin{lemma}[MXS $\fimplies$ PROP1, Theorem 3 of \citet{caragiannis2023new}]
\label[lemma]{thm:impl:mxs-to-prop1}
In a fair division instance with equal entitlements,
if $v_i$ is additive for some agent $i$, $v_i(g) \ge 0$ for every item $g$,
and an allocation $A$ is MXS-fair to agent $i$, then $A$ is also PROP1-fair to agent $i$.
\end{lemma}

\begin{lemma}[EEF1 $\fimplies$ PROP1 for chores]
\label[lemma]{thm:impl:eef1-to-prop1-chores}
For a fair division instance $\fdInst{[n]}{[m]}{(v_i)_{i=1}^n}{w}$
where all items are chores to agent $i$ and $v_i$ is subaditive,
if an allocation $X$ is EEF1-fair to agent $i$,
then $X$ is also PROP1-fair to $i$.
\end{lemma}
\begin{proof}
Let $A$ be agent $i$'s EEF1-certificate for $X$.
Suppose $A$ is not PROP1-fair to $i$. We get $v_i(A_i) < w_iv_i([m])$
and $v_i(A_i \setminus \{c\}) \le w_iv_i([m])$ for all $c \in A_i$.

Since $v_i$ is subadditive, there exists $j \in [n] \setminus \{i\}$
such that $v_i(A_j) > w_jv_i([m])$
(otherwise $v_i([m]) \le \sum_{j=1}^n v_i(A_j) < \sum_{j=1}^n w_jv_j([m]) = v_i([m])$). Hence,
\[ \frac{v_i(A_i)}{w_i} < v_i([m]) < \frac{v_i(A_j)}{w_j}. \]
Thus, $i$ envies $j$. But $A$ is EF1-fair to $i$,
so $\exists c \in A_i$ such that
\[ \frac{v_i(A_i \setminus \{c\})}{w_i} \ge \frac{v_i(A_j)}{w_j}. \]
Since $i$ is PROP1-unsatisfied and $v_i(A_j) > w_jv_i([m])$,
\[ \frac{v_i(A_j)}{w_j} \le \frac{v_i(A_i \setminus \{c\})}{w_i} \le v_i([m]) < \frac{v_i(A_j)}{w_j}, \]
which is a contradiction.
Hence, $A$ is PROP1-fair to $i$, and so, $X$ is PROP1-fair to agent $i$.
\end{proof}

\begin{lemma}[EEF1 $\fimplies$ PROP1]
\label[lemma]{thm:impl:eef1-to-prop1-submod}
Consider a fair division instance $\fdInst{[n]}{[m]}{(v_i)_{i=1}^n}{w}$
where $v_i$ is submodular for some agent $i$
and $w_i \le w_j$ for all $j \in [n]$.
If an allocation $X$ is EEF1-fair to agent $i$,
then $X$ is also PROP1-fair to $i$.
\end{lemma}
\begin{proof}
Let $A$ be agent $i$'s EEF1-certificate for $X$.
Then $A$ is EF1-fair to $i$.
We will show that $A$ is PROP1-fair to $i$
(which would prove that $X$ is also PROP1-fair to $i$).
This is trivial if $A$ is PROP-fair to $i$,
so assume $v_i(A_i) < w_iv_i([m])$.

By \cref{thm:submod-mg-subadd}, we get that
$v_i(\cdot \mid A_i)$ is subadditive. Hence,
\[ v_i([m]) = v_i(A_i) + v_i([m] \setminus A_i \mid A_i)
    \le v_i(A_i) + \sum_{j \neq i} v_i(A_j \mid A_i). \]
Hence, for some $j \in [n] \setminus \{i\}$,
we get $v_i(A_j \mid A_i) > w_jv_i([m])$. Hence,
\[ \frac{v_i(A_i)}{w_i} < v_i([m]) < \frac{v_i(A_j \mid A_i)}{w_j} \le \frac{v_i(A_j)}{w_j}. \]
So, $i$ envies $j$ in $A$. Since $A$ is EF1-fair to $i$, we get two cases:

\textbf{Case 1}: $\exists c \in A_i$ such that
    $\displaystyle \frac{v_i(A_i \setminus \{c\})}{w_i} \ge \frac{v_i(A_j)}{w_j}$.
\\ Then $v_i(A_i \setminus \{c\}) > w_iv_i([m])$, so $A$ is PROP1-fair.

\textbf{Case 2}: $\exists g \in A_j$ such that
    $\displaystyle \frac{v_i(A_i)}{w_i} \ge \frac{v_i(A_j \setminus \{g\})}{w_j}$.
\begin{align*}
& \frac{v_i(A_i \cup \{g\})}{w_i}
    = \frac{v_i(A_i)}{w_i} + \frac{v_i(g \mid A_i)}{w_i}
\\ &\ge \frac{v_i(A_j \setminus \{g\})}{w_j} + \frac{v_i(g \mid A_i)}{w_j}
    \tag{$w_i \le w_j$}
\\ &\ge \frac{v_i(A_j \setminus \{g\} \mid A_i) + v_i(g \mid A_i)}{w_j}
    \tag{$v_i$ is submodular}
\\ &\ge \frac{v_i(A_j \mid A_i)}{w_j}
    \tag{$v_i(\cdot \mid A_i)$ is subadditive}
\\ &> v_i([m]).
\end{align*}
In both cases, we get that $A$ is PROP1-fair to $i$.
Hence, $X$ is also PROP1-fair to $i$.
\end{proof}

\begin{lemma}[EF1 $\fimplies$ PROP1 for $n=2$]
\label[lemma]{thm:impl:ef1-to-prop1-n2}
Consider a fair division instance $\fdInst{[2]}{[m]}{(v_1, v_2)}{w}$ where $v_1$ is subadditive.
If an allocation $A$ is EF1-fair to agent 1, then it is also PROP1-fair to agent 1.
\end{lemma}
\begin{proof}
If $A$ is PROP-fair to agent 1, then we are done, so assume otherwise.
By subadditivity of $v_1$, we get
$v_1([m]) \le v_1(A_1) + v_1(A_2)$.
Since $v_1(A_1) < w_1v_1([m])$, we get $v_1(A_2) > w_2v_1([m])$.
Hence, $v_1(A_1)/w_1 < v_1([m]) < v_1(A_2)/w_2$, so agent 1 envies agent 2.
Since $A$ is EF1-fair to agent 1, there are two cases:

\textbf{Case 1}: $\exists c \in A_1$ such that
    $\displaystyle \frac{v_1(A_1 \setminus \{c\})}{w_1} \ge \frac{v_1(A_2)}{w_2}$.
\\ Then $v_1(A_1 \setminus \{c\}) > w_1v_1([m])$, so $A$ is PROP1-fair to 1.

\textbf{Case 2}: $\exists \ghat \in A_2$ such that
    $\displaystyle \frac{v_1(A_1)}{w_1} \ge \frac{v_1(A_2 \setminus \{\ghat\})}{w_2}$.
\\ By subadditivity of $v_1$ and PROP-unfairness of $A$, we get
\providecommand{\tempRhs}{v_1(A_1 \cup \{\ghat\}) + v_1(A_2 \setminus \{\ghat\})
    \wrapIfTwoCols\le v_1(A_1 \cup \{\ghat\}) + \frac{w_2}{w_1}v_1(A_1)
    \wrapIfTwoCols< v_1(A_1 \cup \{\ghat\}) + w_2v_1([m])}
\ifColsTwo
\else
\[ v_1([m]) \le \tempRhs \]
\fi
\[ \implies w_1 v_1([m]) < v_1(A_1 \cup \{\ghat\}), \]
so $A$ is PROP1-fair to agent 1.
\end{proof}

\begin{lemma}[EEFX $\fimplies$ PROPx, Lemma 2.1 of \citet{li2022almost}]
\label[lemma]{thm:impl:eefx-to-propx}
Consider a fair division instance $\fdInst{[n]}{[m]}{(v_i)_{i=1}^n}{w}$ where
the items are chores to agent $i$ and $v_i$ is subadditive.
If an allocation $A$ is EEFX-fair to agent $i$,
then it is also PROPx-fair to agent $i$.
\end{lemma}
\begin{proof}
Let $B$ be agent $i$'s EEFX-certificate for $A$.

Since $i$ is EFX-satisfied with $B$, for all $j \in [n] \setminus \{i\}$,
\[ \min_{S \in \Scal} \frac{v_i(B_i \setminus S)}{w_i} \ge \frac{v_i(B_j)}{w_j}, \]
where $\Scal \defeq \{S \subseteq B_i: v_i(S \mid B_i \setminus S) < 0\}$.
Add these inequalities for all $j$, weighting each by $w_j$, to get
\[ \min_{S \in \Scal} \frac{v_i(B_i \setminus S)}{w_i} > \sum_{j=1}^n v_i(B_j) \ge v_i([m]). \]
Thus, $B$ is PROPx-fair to $i$, so $A$ is also PROPx-fair to $i$.
\end{proof}

\begin{lemma}[EFX $\fimplies$ PROPavg]
\label[lemma]{thm:impl:efx-to-propavg}
Let $\fdInst{[n]}{[m]}{(v_i)_{i=1}^n}{\eqEnt}$ be a fair division instance
where $v_i$ is submodular for some agent $i$ and $v_i(A_i) \ge 0$.
If an allocation $A$ is EFX-fair to agent $i$,
then it is also PROPavg-fair to agent $i$.
\end{lemma}
\begin{proof}
If $A$ is PROP-fair to agent $i$, we are done, so assume otherwise.
By \cref{thm:submod-mg-subadd}, $v_i(\cdot \mid A_i)$ is subadditive.
Hence, $v_i([m]) = v_i(A_i) + v_i([m] \mid A_i) \le v_i(A_i) + \sum_{j \neq i} v_i(A_j \mid A_i)$.
Then for some agent $\jhat \in [n] \setminus \{i\}$, we have
\[ v_i(A_i) < \frac{v_i([m])}{n} < v_i(A_{\jhat} \mid A_i) \le v_i(A_{\jhat}). \]

Since $i$ doesn't EFX-envy $\jhat$, for all $S \subseteq A_i$ such that
$v_i(S \mid A_i \setminus S) < 0$, we get
\[ v_i(A_i \setminus S) \ge v_i(A_{\jhat}) > \frac{v_i([m])}{n}. \]
Thus, condition C2 of PROPavg holds (c.f.~\cref{defn:propavg}).

For any $j \in [n] \setminus \{i\}$, recall that
$\tau_j \defeq 0$ if $v_i(S \mid A_i) \le 0$ $\forall S \subseteq A_j$, and
$\tau_j \defeq \min(\{v_i(S \mid A_i): S \subseteq A_j \text{ and } v_i(S \mid A_i) > 0\})$ otherwise.
If $\tau_j = 0$, then $v_i(A_j \mid A_i) \le 0 \le v_i(A_i) + \tau_j$.
Now let $\tau_j > 0$. Then $\tau_j = v_i(S \mid A_i) > 0$ for some $S \subseteq A_j$.
Since $A$ is EFX-fair to $i$,
\begin{align*}
& v_i(A_i) + \tau_j
    = v_i(A_i) + v_i(S \mid A_i)
\\ &\ge v_i(A_j \setminus S) + v_i(S \mid A_i)
\\ &\ge v_i(A_j \setminus S \mid A_i) + v_i(S \mid A_i)
    \tag{$v_i$ is submodular}
\\ &\ge v_i(A_j \mid A_i).
    \tag{$v_i(\cdot \mid A_i)$ is subadditive}
\end{align*}
Let $T \defeq \{\tau_j: j \in [n] \setminus \{i\} \text{ and } \tau_j > 0\}$.
On adding the above inequalities for all $j \in [n] \setminus \{i\}$,
and using the subadditivity of $v_i(\cdot \mid A_i)$, we get
\begin{align*}
& (n-1)v_i(A_i) + \sum_{j \neq i}\tau_j
    \wrapIfTwoCols\ge \sum_{j \neq i} v_i(A_j \mid A_i)
    \ge v_i([m] \setminus A_i \mid A_i)
\\ &\implies nv_i(A_i) + \Sum(T) \ge v_i([m]).
\end{align*}
If $T = \emptyset$, condition C4 of PROPavg trivially holds.
Else, $v_i(A_i) + \Sum(T)/(n-1) > v_i(A_i) + \Sum(T)/n \ge v_i([m])/n$.
\end{proof}

\begin{lemma}[EFX $\fimplies$ PROPx for $n=2$]
\label[lemma]{thm:impl:efx-to-propx-n2}
Let $\fdInst{[2]}{[m]}{(v_1, v_2)}{w}$ be a fair division instance
where $v_i$ is subadditive for some agent $i$.
If an allocation $A$ is EFX-fair to agent $i$,
then it is also PROPx-fair to agent $i$.
\end{lemma}
\begin{proof}
Let the two agents be $i$ and $j$.
If $A$ is PROP-fair to agent $i$, we are done, so assume $v_i(A_i) < w_iv_i([m])$.
By subadditivity of $v_i$, we get $v_i([m]) \le v_i(A_i) + v_i(A_j)$. Hence,
\[ \frac{v_i(A_i)}{w_i} < v_i([m]) < \frac{v_i(A_j)}{w_j}. \]
Hence, $i$ envies $j$.
Let $S \subseteq A_i$ such that $v_i(S \mid A_i \setminus S) < 0$.
Since $A$ is EFX-fair to $i$, we get
\[ \frac{v_i(A_i \setminus S)}{w_i} \ge \frac{v_i(A_j)}{w_j} > v_i([m]). \]

Let $S \subseteq A_j$ such that $v_i(S \mid A_i) > 0$. Then
\begin{align*}
& v_i([m]) \le v_i(A_i \cup S) + v_i(A_j \setminus S)
    \tag{$v_i$ is subadditive}
\\ &\le v_i(A_i \cup S) + \frac{w_j}{w_i}v_i(A_i)
    \tag{$A$ is EFX-fair to $i$}
\\ &< v_i(A_i \cup S) + \frac{w_j}{w_i}v_i(A_i \cup S)
    = \frac{v_i(A_i \cup S)}{w_i}.
\end{align*}

Hence, $A$ is PROPx-fair to $i$.
\end{proof}

\subsection{MMS vs EFX}
\label{sec:impls-extra:mms-vs-efx}

We prove some results connecting MMS, EFX, and related notions,
using techniques from \citet{plaut2020almost,caragiannis2023new}.

\begin{lemma}
\label[lemma]{thm:mms-and-all-envy}
For a fair division instance $\fdInst{[n]}{[m]}{(v_i)_{i=1}^n}{w}$,
suppose an allocation $A$ is WMMS-fair to agent $i$
and $i$ envies every other agent.
Then $A$ is also EFX-fair to agent $i$ if at least one of these conditions hold:
\begin{tightenum}
\item The items are goods to agent $i$.
\item $v_i$ is additive and $w_i \le w_j$ for all $j \in [n] \setminus \{i\}$.
\end{tightenum}
\end{lemma}
\begin{proof}
Suppose agent $i$ is not EFX-satisfied by $A$, i.e., she EFX-envies some agent $j$.
Then $\exists S \subseteq A_i \cup A_j$ where either

\begin{tightenum}
\item $S \subseteq A_j$, $v_i(S \mid A_i) > 0$, and
    $\displaystyle \frac{v_i(A_i)}{w_i} < \frac{v_i(A_j \setminus S)}{w_j}$.
\item $S \subseteq A_i$, $v_i(S \mid A_i \setminus S) < 0$, and
    $\displaystyle \frac{v_i(A_i \setminus S)}{w_i} < \frac{v_i(A_j)}{w_j}$.
\end{tightenum}

If all items are goods, case 2 doesn't occur.

\textbf{Case 1}: $S \subseteq A_j$

Let $B$ be the allocation obtained by transferring $S$ from $A_j$ to $A_i$.
Formally, let $B_i \defeq A_i \cup S$, $B_j \defeq A_j \setminus S$,
and $B_k \defeq A_k$ for all $k \in [n] \setminus \{i, j\}$. Then
\ifColsTwo\else\providecommand{\tempWrap}{&}\fi
\begin{align*}
\frac{v_i(B_i)}{w_i} &= \frac{v_i(A_i) + v_i(S \mid A_i)}{w_i} > \frac{v_i(A_i)}{w_i},
\tempWrap \frac{v_i(B_j)}{w_j} &= \frac{v_i(A_j \setminus S)}{w_j} > \frac{v_i(A_i)}{w_i},
\end{align*}
and for any $k \in [n] \setminus \{i, j\}$, we get
\[ \frac{v_i(B_k)}{w_k} = \frac{v_i(A_k)}{w_k} > \frac{v_i(A_i)}{w_i}. \]
Hence,
\[ \min_{k=1}^n \frac{v_i(B_k)}{w_k} > \frac{v_i(A_i)}{w_i} \ge \frac{\WMMS_i}{w_i}, \]
which is a contradiction.

\textbf{Case 2}: $S \subseteq A_i$

Let $v_i$ be additive and $w_i \le w_j$.
Let $B$ be the allocation obtained by transferring $S$ from $A_i$ to $A_j$. Formally,
let $B_i \defeq A_i \setminus S$, $B_j \defeq A_j \cup S$,
and $B_k \defeq A_k$ for all $k \in [n] \setminus \{i, j\}$. Then
\[ \frac{v_i(B_i)}{w_i} = \frac{v_i(A_i) - v_i(S)}{w_i} > \frac{v_i(A_i)}{w_i}, \]
\[ \frac{v_i(B_j)}{w_j} = \frac{v_i(A_j) + v_i(S)}{w_j} > \frac{v_i(A_i \setminus S)}{w_i} - \frac{(-v_i(S))}{w_j} \ge \frac{v_i(A_i)}{w_i}, \]
and for any $k \in [n] \setminus \{i, j\}$, we get
\[ \frac{v_i(B_k)}{w_k} = \frac{v_i(A_k)}{w_k} > \frac{v_i(A_i)}{w_i}. \]
Hence,
\[ \min_{k=1}^n \frac{v_i(B_k)}{w_k} > \frac{v_i(A_i)}{w_i} \ge \frac{\WMMS_i}{w_i}, \]
which is a contradiction.
\end{proof}

\begin{lemma}[MMS $\fimplies$ EFX for $n=2$]
\label[lemma]{thm:impl:mms-to-efx-n2}
For a fair division instance $\fdInst{[2]}{[m]}{(v_i)_{i=1}^2}{w}$,
suppose an allocation $A$ is WMMS-fair to agent 1.
Then $A$ is also EFX-fair to agent 1 if at least one of these conditions hold:
\begin{tightenum}
\item The items are goods to agent 1.
\item $v_1$ is additive and $w_1 \le w_2$.
\end{tightenum}
\end{lemma}
\begin{proof}
If agent 1 doesn't envy agent 2, she is EFX-satisfied.
Otherwise, she is EFX-satisfied because of \cref{thm:mms-and-all-envy}.
\end{proof}

Theorem 2 in \citet{caragiannis2023new} proves that MMS $\fimplies$ EEFX
for additive valuations over goods and equal entitlements.
The proof can be easily adapted to non-additive valuations over goods and unequal entitlements.
\ifVerbose
For the sake of completeness, we give a proof below.
\fi

\begin{lemma}[MMS $\fimplies$ EEFX, \citet{caragiannis2023new}]
\label[lemma]{thm:impl:mms-to-eefx}
For a fair division instance $\fdInst{[n]}{[m]}{(v_i)_{i=1}^n}{w}$,
if all items are goods to agent $i$ and allocation $A$ is WMMS-fair to her,
then $A$ is also EEFX-fair to $i$.
\end{lemma}
\ifVerbose
\begin{proof}
In any allocation $X$, let
$E_X$ be the set of agents envied by $i$,
$S_X$ be the set of agents EFX-envied by $i$,
$W_X$ be the total number of items among the agents in $S_X$.

$E_X \defeq \{t \in [n] \setminus \{i\}: i \textrm{ envies } t \textrm{ in } X\}$,
$S_X \defeq \{t \in [n] \setminus \{i\}: i \textrm{ EFX-envies } t \textrm{ in } X\}$,
$W_X \defeq \sum_{t \in S_X} |X_t|$, and $\phi(X) \defeq (-|E_X|, W_X)$.

First, we show that for any allocation $X$ where $|E_X| \le n-2$ and $S_X \neq \emptyset$,
there exists a \emph{better} allocation $Y$, i.e, $Y_i = X_i$ and $\phi(Y) < \phi(X)$
(tuples are compared lexicographically).

Let $j \in S_X$ and $k \in [n] \setminus \{i\} \setminus E_X$.
Since $i$ EFX-envies $j$, $\exists S \subseteq X_j$ such that $v_i(S \mid X_i) > 0$ and
\[ \frac{v_i(X_i)}{w_i} < \frac{v_i(X_j \setminus S)}{w_j}. \]
Let $Y_k \defeq X_k \cup S$, $Y_j \defeq X_j \setminus S$,
and $Y_t \defeq X_t$ for all $t \in [n] \setminus \{j, k\}$.

Then $i$ envies $j$ in $Y$. Hence, $E_X \subseteq E_Y$.
If $k \in E_Y$, then $|E_Y| > |E_X|$, so $\phi(Y) < \phi(X)$.
If $k \not\in E_Y$, then $W_Y < W_X$, so $\phi(Y) < \phi(X)$.
Hence, $Y$ is better than $X$.

Set $X = A$. As long as $|E_X| < n-1$ and $S_X \neq \emptyset$,
keep modifying $X$ as per Lemma 1.
This process will eventually end, since $\phi(X)$ keeps reducing,
and there are a finite number of different values $\phi(X)$ can take.
Let $B$ be the final allocation thus obtained.
Then $B_i = A_i$, and $|E_B| = n-1$ or $S_B = \emptyset$.

By \cref{thm:mms-and-all-envy},
$|E_B| = n-1$ implies $S_B = \emptyset$.
Hence, $B$ is agent $i$'s EEFX-certificate for $A$.
\end{proof}
\fi

\begin{definition}[leximin partition]
\label[definition]{defn:leximin}
Let $R$ be a totally-ordered set.
For any sequence $X = (x_i)_{i=1}^n$ from $R$, define $\sorted(X)$ to be a permutation of $X$
where entries occur in non-decreasing order.
For any two sequences $X = (x_i)_{i=1}^n$ and $Y = (y_i)_{i=1}^n$ from $R$, we say that $X \le Y$ if
$\exists i \in [n]$ such that $x_i \le y_i$ and $x_j = y_j$ for all $j \in [i-1]$.
(Note that this relation $\le$ over sequences is a total ordering.)

For any set $S \subseteq M$, let $\Pi_n(S)$ denote the set of all $n$-partitions of $S$.
We say that $P \in \Pi_n([m])$ is a leximin $n$-partition of a function $f: 2^{[m]} \to R$ if
\[ P \in \argmax_{X \in \Pi_n([m])} \sorted\left((f(X_j))_{j=1}^n\right). \]
\end{definition}

\begin{lemma}[MMS $\fimplies$ MXS]
\label[lemma]{thm:impl:mms-to-mxs}
Let $\Ical \defeq \fdInst{[n]}{[m]}{(v_i)_{i=1}^n}{\eqEnt}$ be a fair division instance.
Let allocation $A$ be MMS-fair to agent $i$.
Then $A$ is also MXS-fair to $i$ if either all items are goods to agent $i$
or $v_i$ is additive.
\end{lemma}
\begin{proof}
We will show that agent $i$'s leximin $n$-partition is her MXS-certificate for $A$.
Without loss of generality, let $i = 1$.
Let $P$ be a leximin $n$-partition of $v_1$
such that $v_1(P_1) \le \ldots \le v_1(P_n)$.
Then $v_1(P_1) = \MMS_1 \le v_1(A_1)$.

For any agent $j \ge 2$, the allocation $(P_1, P_j)$ is leximin for
the instance $\Icalhat \defeq \fdInst{\{i, j\}}{P_1 \cup P_j}{(v_1, v_j)}{\eqEnt})$.
Hence, agent 1 is MMS-satisfied by $(P_1, P_j)$ in $\Icalhat$.
Since either all items are goods to agent 1 or $v_1$ is additive,
by \cref{thm:impl:mms-to-efx-n2}, agent 1 is EFX-satisfied by $(P_1, P_j)$ in $\Icalhat$.
Thus, $P$ is EFX-fair to agent 1 in $\Ical$,
so $P$ is agent 1's MXS-certificate for $A$.
\end{proof}

\begin{lemma}[MMS $\fimplies$ \MXSZero]
\label[lemma]{thm:impl:mms-to-mxs0}
Let $\Ical \defeq \fdInst{[n]}{[m]}{(v_i)_{i=1}^n}{\eqEnt}$ be a fair division instance over goods.
If an allocation $A$ is MMS-fair to agent $i$, then it is also \MXSZero-fair to $i$.
(\EFXZero{} is defined in \cref{defn:efx0-goods}.
\MXSZero{} is the minimum-\EFXZero{}-share.
\EFXZero{} implies EF1, so \MXSZero{} implies M1S.)
\end{lemma}
\begin{proof}
Without loss of generality, let $i = 1$.
Let $\vhat(X) \defeq (v_1(X), |X|)$.
Then $\vhat$'s co-domain is totally-ordered, assuming pairs are compared lexicographically.
Let $P$ be a leximin $n$-partition of $\vhat$ (c.f.~\cref{defn:leximin})
such that $\vhat(P_1) \le \ldots \le \vhat(P_n)$.
($P$ is also called the leximin++ partition of $v_1$ \citep{plaut2020almost}).
Then $v_1(P_1) = \min_{j=1}^n v_1(P_j)$.
Hence, $v_1(P_1) = \MMS_1 \le v_1(A_1)$.

For any $j \ge 2$, we get that $(P_1, P_j)$ is also a leximin 2-partition of $\vhat$
(when $\vhat$'s domain is restricted to $P_1 \cup P_j$).
We will now show that agent 1 doesn't \EFXZero{}-envy agent $j$, i.e.,
$v_1(P_1) \ge \max_{g \in P_j} v_1(P_j \setminus \{g\})$.
(If $P_j = \emptyset$, this is trivially true.)
Suppose $v_1(P_1) < v_1(P_j \setminus \{g\})$ for some $g \in P_j$.
Let $Q_1 \defeq P_1 \cup \{g\}$ and $Q_j \defeq P_j \setminus \{g\}$.
Then $\vhat(Q_1) > \vhat(P_1)$ and $\vhat(Q_j) > \vhat(P_1)$.
Hence, $(P_1, P_j)$ is not a leximin 2-partition of $\vhat$, which is a contradiction.
Hence, agent 1 doesn't \EFXZero{}-envy agent $j$.

Thus, $P$ is \EFXZero-fair to agent 1 in $\Ical$,
so $P$ is agent 1's \MXSZero{}-certificate for $A$.
\end{proof}

\subsection{Among PROP, APS, MMS}

\begin{lemma}[PROP $\fimplies$ APS, Proposition 4 of \citet{babaioff2023fair}]
\label[lemma]{thm:impl:prop-to-aps}
For any fair division instance $\fdInst{[n]}{[m]}{(v_i)_{i=1}^n}{w}$,
$\APS_i \le w_iv_i([m])$ for agent $i$ if $v_i$ is additive.
\end{lemma}
\ifVerbose
\begin{proof}
Set the price $p(g)$ of each item $g$ to $v_i(g)$. Then
\[ \APS_i \le \!\!\!\!\!\!\max_{\substack{S \subseteq [m]:\\p(S) \le w_ip([m])}} \!\!\!\!\!\!v_i(S)
= \!\!\!\!\!\!\max_{\substack{S \subseteq [m]:\\v_i(S) \le w_iv_i([m])}} \!\!\!\!\!\!v_i(S) \le w_iv_i([m]).
\qedhere \]
\end{proof}
\fi

\begin{lemma}[PROP $\fimplies$ WMMS]
\label[lemma]{thm:impl:prop-to-wmms}
In any instance $\fdInst{[n]}{[m]}{(v_i)_{i=1}^n}{w}$,
if $v_i$ is superadditive, then $\WMMS_i \le w_iv_i([m])$.
\end{lemma}
\begin{proof}
Let $P$ be agent $i$'s WMMS partition. Then
\begin{align*}
v_i([m]) &\ge \sum_{j=1}^n v_i(P_j) = \sum_{j=1}^n w_j\left(\frac{v_i(P_j)}{w_j}\right)
    \wrapIfTwoCols\ge \sum_{j=1}^n w_j\frac{\WMMS_i}{w_i} = \frac{\WMMS_i}{w_i}.
\qedhere
\end{align*}
\end{proof}

\begin{lemma}
\label[lemma]{thm:prefix-sum-bound}
Let $a_1 \le a_2 \le \ldots \le a_n$ be $n$ real numbers.
Let $s_k \defeq \sum_{i=1}^k a_i$ for any $0 \le k \le n$.
Then $s_k \le (k/n)s_n$.
\end{lemma}
\begin{proof}
\[ s_n = s_k + \!\!\sum_{i=k+1}^n \!\!a_i \ge s_k + (n-k)a_k \ge s_k + (n-k)\frac{s_k}{k}
= \frac{n}{k}s_k. \qedhere \]
\end{proof}

\ifVerbose
\begin{lemma}[PROP $\fimplies$ pessShare]
\label[lemma]{thm:impl:prop-to-pessShare}
For any fair division instance $\fdInst{[n]}{[m]}{(v_i)_{i=1}^n}{w}$,
if $v_i$ is superadditive, then $\pessShare_i \le w_iv_i([m])$.
\end{lemma}
\begin{proof}
Let $P$ be agent $i$'s $\ell$-out-of-$d$-partition.
Assume \wLoG{} that $v_i(P_1) \le v_i(P_2) \le \ldots \le v_i(P_d)$.
Then by \cref{thm:prefix-sum-bound} and superadditivity of $v_i$, we get
\[ \loodM_i = \sum_{j=1}^{\ell} v_i(P_j)
    \le \frac{\ell}{d}\sum_{j=1}^d v_i(P_j) \le \frac{\ell}{d}v_i([m]). \]
Hence,
\begin{align*}
\pessShare_i &\defeq \sup_{1 \le \ell \le d: \ell/d \le w_i} \loodM_i
    \wrapIfTwoCols\le \sup_{1 \le \ell \le d: \ell/d \le w_i} (\ell/d)v_i([m])
    \wrapIfTwoCols\le w_iv_i([m]).
\qedhere
\end{align*}
\end{proof}
\fi

\begin{lemma}[APS $\fimplies$ $\pessIfVerbose$, \citet{babaioff2023fair}]
\label[lemma]{thm:impl:aps-to-pess}
In any fair division instance, $\APS_i \ge \pessIfVerbose_i$ for each agent $i$.
\end{lemma}
\begin{proof}
\ifVerbose
Proposition 2 in \citet{babaioff2023fair} proves this for goods,
but their proof works for mixed manna too.
\fi
\end{proof}

\begin{lemma}
\label[lemma]{thm:wmms-vs-knapsack}
For a fair division instance $\fdInst{[n]}{[m]}{(v_i)_{i=1}^n}{w}$ and any agent $i$, define
\[ \beta_i \defeq \max_{j=1}^n\; \max_{S \subseteq [m]:\,v_i(S) \le w_jv_i([m])}\; \frac{v_i(S)}{w_j}. \]
If $v_i$ is superadditive, then $\WMMS_i \le w_i\beta_i$.
If $v_i$ is additive and $n = 2$, then $\WMMS_i = w_i\beta_i$.
\end{lemma}
\begin{proof}
$\beta_i = \max_{j=1}^n v_i(S_j)/w_j$, where
\[ S_j \defeq \argmax_{S \subseteq [m]:\,v_i(S) \le w_jv_i([m])}\; \frac{v_i(S)}{w_j}. \]
Let $\Pi_n([m])$ be the set of all $n$-partitions of $[m]$. Let
\begin{align*}
P &\defeq \argmax_{P \in \Pi_n([m])} \min_{j=1}^n \frac{v_i(P_j)}{w_j},
& t &\defeq \argmin_{j=1}^n \frac{v_i(P_j)}{w_j}.
\end{align*}
Then $\WMMS_i/w_i \defeq v_i(P_t)/w_t$.

By \cref{thm:impl:prop-to-wmms}, we get $v_i(P_t)/w_t \le v_i([m])$. Hence,
\[ \frac{\WMMS_i}{w_i} = \frac{v_i(P_t)}{w_t} \le \frac{v_i(S_t)}{w_t}
    \le \max_{j=1}^n \frac{v_i(S_j)}{w_j} = \beta_i. \]

Now let $v_i$ be additive and $n = 2$. For any $j \in [2]$,
let $Q^{(j)}$ be an allocation where $Q^{(j)}_j \defeq S_j$
and $Q^{(j)}_k \defeq [m] \setminus S_j$.
Then $v_i(Q^{(j)}_j) = v_i(S_j) \le w_jv_i([m])$ and
$v_i(Q^{(j)}_k) = v_i([m]) - v_i(S_j)
    \ge v_i([m]) - w_jv_i([m]) = w_kv_i([m])$.
Hence,
\[ \frac{v_i(Q^{(j)}_j)}{w_j} \le v_i([m]) \le \frac{v_i(Q^{(j)}_k)}{w_k}. \]
Therefore,
\begin{align*}
& \frac{\WMMS_i}{w_i} = \min\left(\frac{v_i(P_1)}{w_1}, \frac{v_i(P_2)}{w_2}\right)
\\ &\ge \max_{j=1}^n \min\left(\frac{v_i(Q^{(j)}_j)}{w_j}, \frac{v_i(Q^{(j)}_k)}{w_k}\right)
    \tag{by definition of $P$}
\\ &= \max_{j=1}^n \frac{v_i(S_j)}{w_j} = \beta_i.
\end{align*}
Hence, $\WMMS_i/w_i = \beta_i$.
\end{proof}

\begin{lemma}[WMMS $\fimplies$ APS for $n=2$, \citet{babaioff2023fair}]
\label[lemma]{thm:impl:mms-to-aps-n2}
For a fair division instance $\fdInst{[2]}{[m]}{(v_i)_{i=1}^2}{w}$,
if $v_i$ is additive for some $i$, then $\APS_i \le \WMMS_i$.
Moreover, when entitlements are equal, we get $\APS_i = \MMS_i$.
\end{lemma}
\begin{proof}
By \cref{thm:wmms-vs-knapsack}, we get
\begin{align*}
& \frac{\WMMS_i}{w_i} = \beta_i \defeq \max_{j=1}^2 \frac{v_i(S_j)}{w_j},
    \wrapIfTwoCols\quad\textrm{where}\quad S_j \defeq \argmax_{\substack{S \subseteq [m]:\\ v_i(S) \le w_jv_i([m])}} v_i(S).
\end{align*}
Setting $p = v_i$ in the definition of APS gives us
\[ \APS_i \le \!\!\!\!\max_{\substack{S \subseteq [m]:\\ p(S) \le w_ip([m])}} \!\!\!\!v_i(S)
    = v_i(S_i) \le w_i\beta_i = \WMMS_i. \]
For equal entitlements, $\APS_i \ge \MMS_i$ by \cref{thm:impl:aps-to-pess}.
\end{proof}

\subsection{Binary and Negative-Binary Valuations}
\label{sec:impls-extra:tribool}

In this section, we study implications between fairness notions when agents' marginals are \emph{triboolean},
i.e., for all $R \subseteq [m]$ and $t \in [m] \setminus R$, we have $v_i(t \mid R) \in \{-1, 0, 1\}$.

We begin by proving EF1 $\fimplies$ EFX (\cref{thm:impl:tribool:ef1-to-efx})
and PROP1 $\fimplies$ PROPx (\cref{thm:impl:prop1-to-propx-tribool}) for additive valuations.

\begin{lemma}[EF1 $\fimplies$ EFX]
\label[lemma]{thm:impl:tribool:ef1-to-efx}
Consider a fair division instance $\fdInst{[n]}{[m]}{(v_i)_{i=1}^n}{w}$ where
$v_i$ is additive for some agent $i$. If an allocation $A$ is EF1-fair to $i$,
then it is also EFX-fair to $i$ if at least one of these conditions is satisfied:
\begin{tightenum}
\item agents have equal entitlements and $v_i$'s marginals are triboolean
    (i.e., $v_i(t \mid \cdot) \in \{-1, 0, 1\}$ $\forall t \in [m]$).
\item $v_i$'s marginals are binary
    (i.e., $v_i(t \mid \cdot) \in \{0, 1\}$ $\forall t \in [m]$).
\item $v_i$'s marginals are negative binary
    (i.e., $v_i(t \mid \cdot) \in \{0, -1\}$ $\forall t \in [m]$).
\end{tightenum}
\end{lemma}
\begin{proof}
Suppose agents have equal entitlements and $v_i$'s marginals are triboolean.
Since $i$ is EF1-satisfied, for any other agent $j$, either $v_i(A_i) \ge v_i(A_j)$,
or $v_i(A_i) \ge v_i(A_j \setminus \{g\})$ for some $g \in A_j$,
or $v_i(A_i \setminus \{c\}) \ge v_i(A_j)$ for some $c \in A_i$.
Hence, we get $v_i(A_i) \ge v_i(A_j) - 1$.
For any $S \subseteq A_j$ such that $v_i(S \mid A_i) > 0$,
we have $v_i(A_j \setminus S) = v_i(A_j) - v_i(S) \le v_i(A_j) - 1 \le v_i(A_i)$.
For any $S \subseteq A_i$ such that $v_i(S \mid A_i \setminus S) < 0$,
we have $v_i(A_i \setminus S) = v_i(A_i) - v_i(S) \ge v_i(A_i) + 1 \ge v_i(A_j)$.
Hence, $i$ doesn't EFX-envy $j$.
Since this holds for every $j$, $A$ is EFX-fair to agent $i$.

Suppose $v_i$'s marginals are binary.
Since $i$ is EF1-satisfied, for any other agent $j$, either $v_i(A_i)/w_i \ge v_i(A_j)/w_j$,
or for some $g \in A_j$, we have $v_i(A_i)/w_i \ge v_i(A_j \setminus \{g\})/w_j$,
which implies $v_i(A_i)/w_i \ge (v_i(A_j) - 1)/w_j$.
For any $S \subseteq A_j$ such that $v_i(S \mid A_i) > 0$, we have
\[ \frac{v_i(A_j \setminus S)}{w_j} = \frac{v_i(A_j) - v_i(S)}{w_j}
\le \frac{v_i(A_j) - 1}{w_j} \le \frac{v_i(A_i)}{w_i}. \]
Hence, $i$ doesn't EFX-envy agent $j$.
Since this holds for every $j$, $A$ is EFX-fair to agent $i$.

Suppose $v_i$'s marginals are negative binary.
Since $i$ is EF1-satisfied, for any other agent $j$, either $v_i(A_i)/w_i \ge v_i(A_j)/w_j$,
or for some $c \in A_i$, we have $v_i(A_i \setminus \{c\})/w_i \ge v_i(A_j)/w_j$,
which implies $(v_i(A_i) + 1)/w_i \ge v_i(A_j)/w_j$.
For any $S \subseteq A_i$ such that $v_i(S \mid A_i \setminus S) < 0$, we have
\[ \frac{v_i(A_i \setminus S)}{w_i} = \frac{v_i(A_i) - v_i(S)}{w_i}
\ge \frac{v_i(A_i) + 1}{w_i} \ge \frac{v_i(A_j)}{w_j}. \]
Hence, $i$ doesn't EFX-envy agent $j$.
Since this holds for every $j$, $A$ is EFX-fair to agent $i$.
\end{proof}

\begin{lemma}[PROP1 $\fimplies$ PROPx]
\label[lemma]{thm:impl:prop1-to-propx-tribool}
Consider a fair division instance $\fdInst{[n]}{[m]}{(v_i)_{i=1}^n}{w}$ where
marginals are triboolean for some agent $i$ (i.e., $v_i(t \mid \cdot) \in \{-1, 0, 1\}$ $\forall t \in [m]$).
Then if an allocation $A$ is PROP1-fair to agent $i$, then it is also PROPx-fair to $i$.
\end{lemma}
\begin{proof}
If $v_i(A_i) \ge w_iv_i([m])$, then $A$ is PROPx.
Otherwise, we get $v_i(A_i) + 1 > w_iv_i([m])$ since $A$ is PROP1
and marginals are triboolean.
Removing any positive-disutility subset from $A_i$
or adding any positive-utility subset to $A_i$ will increase its value by at least 1.
Hence, $A$ is PROPx.
\end{proof}

We now characterize different fairness notions.
If agent $i$ having entitlement $w_i$ has an additive valuation function $v_i$
with triboolean marginals, then for any allocation $A$,
\begin{enumerate}
\item for equal entitlements, $A$ is PROP-fair to $i$ $\iff$ $A$ is EEF-fair to $i$
    $\iff$ $v_i(A_i) \ge \ceil{v_i([m])/n}$ (\cref{thm:impl:tribool:prop}).
\item $A$ is PROP1-fair to $i$ $\iff$ $A$ is PROPx-fair to $i$ $\iff$ $A$ is APS-fair to $i$
    $\iff$ $v_i(A_i) \ge \floor{w_i \cdot v_i([m])}$ (\cref{thm:impl:tribool:prop1,thm:impl:tribool:aps}).
\item for equal entitlements, $A$ is MMS-fair to $i$ $\iff$ $A$ is M1S-fair to $i$ $\iff$ $A$ is EEFX-fair to $i$
    $\iff$ $v_i(A_i) \ge \floor{v_i([m])/n}$ (\cref{thm:impl:tribool:m1s,thm:impl:tribool:mms-to-eefx}).
\end{enumerate}

\begin{lemma}
\label[lemma]{thm:tribool-rr}
For an additive function $f: 2^{[m]} \to \{-1, 0, 1\}$,
there exists an $n$-partition $P$ such that $|f(P_i) - f(P_j)| \le 1$
for all $i, j \in [n]$ and
and $\floor{f([m])/n} \le f(P_i) \le \ceil{f([m])/n}$ for all $i \in [n]$.
\end{lemma}
\begin{proof}
Partition $[m]$ into
goods $M_+ \defeq \{g \in [m]: v_i(g) > 0\}$,
chores $M_- \defeq \{c \in [m]: v_i(c) < 0\}$,
and neutral items $M_0 \defeq \{t \in [m]: v_i(t) = 0\}$.
Fuse items $M_0$, $\min(|M_+|, |M_-|)$ goods, and $\min(|M_+|, |M_-|)$ chores
into a single item $h$.
Then we are left with only goods and a neutral item,
or only chores and a neutral item.
Using round-robin, one can allocate items such that
any two bundles differ by at most one item.
\end{proof}

\begin{lemma}[PROP $\fimplies$ EEF]
\label[lemma]{thm:impl:tribool:prop}
Consider a fair division instance $\fdInst{[n]}{[m]}{(v_i)_{i=1}^n}{\eqEnt}$
where $v_i$ is additive and $v_i(t) \in \{-1, 0, 1\}$ for all $t \in [m]$ for some agent $i$.
If an allocation $A$ is PROP-fair to $i$, then it is also EEF-fair to $i$.
\end{lemma}
\begin{proof}
Since $A$ is PROP-fair to $i$, and $v_i(S) \in \mathbb{Z}$ for all $S \subseteq [m]$,
we get $v_i(A_i) \ge \ceil{v_i([m])/n}$.

Construct allocation $B$ where $B_i = A_i$, and items $[m] \setminus A_i$
are allocated among agents $[n] \setminus \{i\}$ using \cref{thm:tribool-rr} with $f = v_i$.
We show that $B$ is agent $i$'s EEF-certificate for $A$.

In $B$, for each agent $j \in [n] \setminus \{i\}$, we have
\begin{align*}
v_i(B_j) &\le \bigceil{\frac{v_i([m] \setminus A_i)}{n-1}}
    \le \bigceil{\frac{v_i([m]) - v_i([m])/n}{n-1}}
    \wrapIfTwoCols\le \bigceil{\frac{v_i([m])}{n}} \le v_i(A_i).
\end{align*}
Hence, $i$ doesn't envy anyone in $B$.
Hence, $B$ is agent $i$'s EEF-certificate for $A$.
\end{proof}

\begin{lemma}
\label[lemma]{thm:impl:tribool:prop1}
Consider a fair division instance $\fdInst{[n]}{[m]}{(v_i)_{i=1}^n}{w}$ where
$v_i$ is additive and $v_i(t) \in \{-1, 0, 1\}$ for all $t \in [m]$ for some agent $i$.
Then the following statements are equivalent:
\begin{tightenum}
\item Allocation $A$ is PROP1-fair to $i$.
\item Allocation $A$ is PROPx-fair to $i$.
\item $v_i(A_i) \ge \floor{w_iv_i([m])}$.
\end{tightenum}
\end{lemma}
\begin{proof}
Partition $[m]$ into
goods $M_+ \defeq \{g \in [m]: v_i(g) > 0\}$,
chores $M_- \defeq \{c \in [m]: v_i(c) < 0\}$,
and neutral items $M_0 \defeq \{t \in [m]: v_i(t) = 0\}$.

\textbf{Case 1}: $A_i$ has all goods and no chores.
\\ Then $v_i(A_i) \ge \max(0, v_i([m]))$.
If $v_i([m]) \ge 0$, then $v_i(A_i) \ge v_i([m]) \ge w_iv_i([m])$,
else $v_i(A_i) \ge 0 \ge w_iv_i([m])$.
Hence, $v_i(A_i) \ge w_iv_i([m])$ and $A$ is PROPx+PROP1.

\textbf{Case 2}: $A_i$ has a chore or some good is outside $A_i$.
\\ Then adding a good to $A_i$ or removing a chore from $A_i$
makes its value more than $w_iv_i([m])$ iff $v_i(A_i) \ge \floor{w_iv_i([m])}$.
So, $A$ is PROP1 iff $A$ is PROPx iff $v_i(A_i) \ge \floor{w_iv_i([m])}$.
\end{proof}

\begin{lemma}
\label[lemma]{thm:impl:tribool:aps}
Consider a fair division instance $\fdInst{[n]}{[m]}{(v_i)_{i=1}^n}{w}$ where
$v_i$ is additive and $v_i(t) \in \{-1, 0, 1\}$ for all $t \in [m]$ for some agent $i$.
Then $\APS_i = \floor{w_iv_i([m])}$.
\end{lemma}
\begin{proof}
Partition $[m]$ into goods $M_+ \defeq \{g \in [m]: v_i(g) > 0\}$,
chores $M_- \defeq \{c \in [m]: v_i(c) < 0\}$, and
neutral items $M_0 \defeq \{t \in [m]: v_i(t) = 0\}$.

First, set $p(t) = v_i(t)$ for each $t \in [m]$ to get $\APS_i \le w_iv_i([m])$.
Since the APS is the value of some bundle, and bundle values are integers,
we get $\APS_i \le \floor{w_iv_i([m])}$.

Pick an arbitrary price vector $p \in \Delta_m$.
We will construct a set $S$ such that $p(S) \le w_ip([m])$ and $v_i(S) \ge \floor{w_iv_i([m])}$.
Fuse items $M_0$, $\min(|M_+|, |M_-|)$ goods, and $\min(|M_+|, |M_-|)$ chores
into a single item $h$.
Let $M'_+$ and $M'_-$ be the remaining goods and chores, respectively.
Then $M'_+ = \emptyset$ or $M'_- = \emptyset$.
Using techniques from \cref{thm:aps-optimal-price},
we can assume \wLoG{} that $p_g \ge 0$ for all $g \in M'_+$,
$p_c \le 0$ for all $c \in M'_-$, and $p_h = 0$.

\textbf{Case 1}: $M'_- = \emptyset$.
\\ Let $m_i \defeq |M'_+| = v_i([m])$.
Let $S$ be the cheapest $\floor{w_im_i}$ items in $M'_+$.
Then using \cref{thm:prefix-sum-bound}, we get
\[ p(S) \le \frac{\floor{w_im_i}}{m_i}p(M_+) \le w_ip([m]). \]
Then $S$ is affordable, so $\APS_i \ge \floor{w_im_i}$.

\textbf{Case 2}: $M'_+ = \emptyset$.
\\ Let $m_i \defeq |M'_-| = -v_i([m])$.
Let $S$ be the cheapest $\ceil{w_im_i}$ items in $M'_- \cup \{h\}$.
Then using \cref{thm:prefix-sum-bound}, we get
\[ -p(S) \ge \frac{\ceil{w_im_i}}{m_i}(-p([m])) \ge w_i(-p([m])). \]
Then $S$ is affordable and $-v_i(S) \le \ceil{w_im_i} = -\floor{w_iv_i([m])}$.
Hence, $\APS_i \ge \floor{w_iv_i([m])}$.
\end{proof}

\begin{lemma}
\label[lemma]{thm:impl:tribool:m1s}
Consider an instance $\fdInst{[n]}{[m]}{(v_i)_{i=1}^n}{\eqEnt}$
where $v_i$ is additive and $v_i(t) \in \{-1, 0, 1\}$ for all $t \in [m]$ for some agent $i$.
Then $\MMS_i = \mathrm{M1S}_i = \floor{v_i([m])/n}$.
\end{lemma}
\begin{proof}
Allocate items $[m]$ among agents $[n]$ using \cref{thm:tribool-rr} with $f = v_i$.
Then any two bundles differ by a value of at most one.
Hence, $\MMS_i = \floor{m'/n}$ and $\mathrm{M1S}_i \le \floor{m'/n}$,
where $m' \defeq v_i([m])$.

Let $X$ be an allocation where agent $i$ is EF1-satisfied.
Then any two bundles can differ by a value of at most one.
Hence, the smallest value $v_i(X_i)$ can have is $\floor{m'/n}$.
Hence, $\mathrm{M1S}_i \ge \floor{m'/n}$.
\end{proof}

\begin{remark}
\label[remark]{thm:ceil-floor}
For any $m \in \mathbb{Z}$ and $n \in \mathbb{Z}_{\ge 1}$, we have
\[ \bigfloor{\frac{m}{n}} = \bigceil{\frac{m+1}{n}} - 1. \]
\end{remark}

\begin{lemma}
\label[lemma]{thm:impl:tribool:mms-to-eefx}
Consider a fair division instance $\fdInst{[n]}{[m]}{(v_i)_{i=1}^n}{\eqEnt}$
where $v_i$ is additive and $v_i(t) \in \{-1, 0, 1\}$ for all $t \in [m]$ for some agent $i$.
Then an allocation $A$ is EEFX-fair to $i$ iff $v_i(A_i) \ge \floor{v_i([m])/n}$.
\end{lemma}
\begin{proof}
Suppose $v_i(A_i) \ge \floor{v_i([m])/n}$.
Construct an allocation $B$ where $B_i = A_i$, and items $[m] \setminus A_i$
are allocated among agents $[n] \setminus \{i\}$ using \cref{thm:tribool-rr} with $f = v_i$.
We will show that $B$ is agent $i$'s EEFX-certificate for $A$.

Let $k \defeq v_i([m])$.
Suppose $v_i(A_i) \ge \floor{k/n} = \ceil{(k+1)/n} - 1$ (c.f.~\cref{thm:ceil-floor}).
Then for any other agent $j \in [n] \setminus \{i\}$, we get
\begin{align*}
v_i(B_j) &\le \bigceil{\frac{v_i([m] \setminus A_i)}{n-1}}
    \le \bigceil{\frac{k - (k+1)/n + 1}{n-1}}
    \wrapIfTwoCols= \bigceil{\frac{k+1}{n}} \le v_i(B_i) + 1.
\end{align*}
If $B_j$ contains no goods and $B_i$ contains no chores,
then $v_i(B_i) \ge 0 \ge v_i(B_j)$, so $i$ doesn't envy $j$ in $B$.
Otherwise, removing a good from $B_j$ or removing a chore from $B_i$
eliminates $i$'s envy towards $j$.
Hence, $B$ is agent $i$'s EEFX-certificate for $A$, so $A$ is EEFX-fair to $i$.

Now suppose $v_i(A_i) < \floor{v_i([m])/n}$. Let $B$ be any allocation where $B_i = A_i$.
Then $\exists j \in [n] \setminus \{i\}$ such that $v_i(B_j) > v_i([m])/n$.
If $v_i([m])/n$ is an integer, then $v_i(B_j) > v_i([m])/n > v_i(B_i)$, so $v_i(B_j) \ge v_i(B_i) + 2$.
Otherwise, $v_i(B_j) \ge \ceil{v_i([m])/n} > \floor{v_i([m])/n} > v_i(B_i)$, so $v_i(B_j) \ge v_i(B_i) + 2$.
Thus, $B_j$ has a good or $B_i$ has a chore.
Even after removing a good from $B_j$ or removing a chore from $B_i$,
agent $i$ will still envy agent $j$, so $B$ is not EFX-fair to agent $i$.
Thus, $A$ is not EEFX-fair to $i$.
\end{proof}

We now prove more implications for additive valuations.

\begin{lemma}[EF1 $\fimplies$ GAPS]
\label[lemma]{thm:impl:tribool:ef1-gaps}
Consider a fair division instance $\fdInst{[n]}{[m]}{(v_i)_{i=1}^n}{w}$ where
$v_i$ is additive and $v_i(t) \in \{-1, 0, 1\}$ for all $t \in [m]$ for some agent $i$.
If an allocation $A$ is EF1-fair to $i$, then it is also groupwise-APS-fair to $i$
if at least one of these conditions hold:
\begin{tightenum}
\item $n=2$
\item $w_i \le w_j$ for all $j \in [n] \setminus \{i\}$.
\item $v_i(c) \in \{0, -1\}$ for all $c \in [m]$.
\end{tightenum}
\end{lemma}
\begin{proof}
On restricting $A$ to any subset $S$ of agents, we get an allocation $B$ that is EF1-fair to $i$.
$B$ is also PROP1-fair to $i$ by \cref{thm:impl:eef1-to-prop1-chores,thm:impl:eef1-to-prop1-submod,thm:impl:ef1-to-prop1-n2},
and APS-fair to $i$ by \cref{thm:impl:tribool:prop1,thm:impl:tribool:aps}.
Hence, $A$ is groupwise-APS-fair to $i$.
\end{proof}

\begin{lemma}[EF1 $\fimplies$ GWMMS, M1S $\fimplies$ WMMS]
\label[lemma]{thm:impl:binary:ef1-to-gwmms}
\label[lemma]{thm:impl:binary:m1s-to-wmms}
Consider a fair division instance $\fdInst{[n]}{[m]}{(v_i)_{i=1}^n}{w}$ where for some agent $i$,
$v_i$ is additive and $v_i(g) \in \{0, 1\}$ for all $g \in [m]$.
If an allocation $A$ is EF1-fair to $i$, then it is also GWMMS-fair to $i$.
If an allocation $B$ is M1S-fair to $i$, then it is also WMMS-fair to $i$.
\end{lemma}
\begin{proof}
Suppose $A$ is not WMMS-fair to $i$. By the definition of WMMS, we get
\begin{align*}
& \frac{\WMMS_i}{w_i} = \min_{k=1}^n \frac{v_i(X^*_k)}{w_k},
\wrapIfTwoCols\quad\text{where}\quad X^* \in \argmax_X \min_{k=1}^n \frac{v_i(X_k)}{w_k}.
\end{align*}
Then $v_i(X^*_i) \ge \WMMS_i > v_i(A_i)$.
Now pick any $j \in [n] \setminus \{i\}$. Since $A$ is EF1-fair to $i$, we get
\[ \frac{v_i(X^*_j)}{w_j} \ge \frac{\WMMS_i}{w_i} > \frac{v_i(A_i)}{w_i} \ge \frac{v_i(A_j)-1}{w_j}. \]
Bundle values are integers, so $v_i(X^*_j) \ge v_i(A_j)$. Thus,
\[ v_i([m]) = \sum_{k=1}^n v_i(X^*_k) > \sum_{k=1}^n v_i(A_k) = v_i([m]), \]
which is a contradiction. Thus, $A$ is WMMS-fair to $i$.

Restricting an EF1 allocation to a subset of agents
gives us another EF1 allocation, so $A$ is also groupwise WMMS.

Since allocation $B$ is M1S-fair to $i$, there is an allocation $C$
that is EF1-fair to $i$ and has $v_i(C_i) \le v_i(B_i)$.
Thus, $C$ is also WMMS-fair to $i$, so $B$ is WMMS-fair to $i$.
\end{proof}

\begin{lemma}[monotonicity of envy-freeness]
\label[lemma]{thm:tribool:ef-monotonicity}
Let $Q \subseteq \{-1, 0, 1\}$.
Consider a fair division instance $\fdInst{[n]}{[m]}{(v_i)_{i=1}^n}{w}$ where for some agent $i$,
$v_i$ is additive and $v_i(t) \in Q$ for all $t \in [m]$.
Let $F \in \{\mathrm{EF},\allowbreak \mathrm{EF1},\allowbreak \mathrm{EFX}\}$.
Let allocation $A$ be $F$-fair to agent $i$.
Let $B$ be an allocation such that $B_i = A_i$ and $v_i(B_j) \le v_i(A_j)$ for some agent $j$.
Then $i$ does not $F$-envy $j$ in $B$ if $F$ = EF, $Q = \{0, 1\}$, or $Q = \{-1, 0\}$.
\end{lemma}
\begin{proof}
This holds trivially if $F$ = EF, or if $i$ does not envy $j$ in $A$ or $B$.
So now assume $F$ is EF1 or EFX and $i$ envies $j$ in $A$ and $B$.

First, let us consider $F$ = EFX.
The chores condition of EFX continues to be satisfied in $B$. We need to verify the goods condition.
The goods condition holds trivially if $B_j$ has no goods, so we are done for EFX if $Q = \{-1, 0\}$.
Now assume $B_j$ has at least one good. We need to prove that $v_i(B_i)/w_i \ge (v_i(B_j)-1)/w_j$.
If $Q = \{0, 1\}$, then $A_j$ always contains a good since $v_i(A_j)/w_j > v_i(A_i)/w_i \ge 0$.
Since $i$ does not $F$-envy $j$ in $A$, we get
\[ \frac{v_i(B_i)}{w_i} = \frac{v_i(A_i)}{w_i} \ge \frac{v_i(A_j)-1}{w_j} \ge \frac{v_i(B_j)-1}{w_j}. \]
Hence, $i$ does not EFX-envy $j$ if $Q = \{-1, 0\}$ or $Q = \{0, 1\}$.

Now let us consider $F$ = EF1. We have two cases, depending on the value of $Q$.
\begin{tightenum}
\item If $Q = \{-1, 0\}$, then $A_i$ contains a chore since $i$ envies $j$ in $A$.
    Since $i$ does not EF1-envy $j$ in $A$, we get $(v_i(A_i) + 1)/w_i \ge v_i(A_j)/w_j \ge v_i(B_j)/w_j$.
\item If $Q = \{0, 1\}$, then $A_j$ contains a good and $B_j$ contains a good since $i$ envies $j$ in $A$ and $B$.
    Since $i$ does not EF1-envy $j$ in $A$, we get $v_i(A_i)/w_i \ge (v_i(A_j) - 1)/w_j \ge (v_i(B_j) - 1)/w_j$.
\end{tightenum}
In both cases, $i$ does not EF1-envy $j$ in $B$.
\end{proof}

\begin{lemma}
\label[lemma]{thm:tribool:equalizing-shift}
Consider a fair division instance $\fdInst{[n]}{[m]}{(v_i)_{i=1}^n}{w}$ where for some agent $i$,
$v_i$ is additive and $v_i(t) \in \{-1, 0, 1\}$ for all $t \in [m]$.
Let $F \in \{\mathrm{EF},\allowbreak \mathrm{EF1},\allowbreak \mathrm{EFX}\}$.
Let allocation $A$ be $F$-fair to agent $i$.

Let $j, k \in [n] \setminus \{i\}$ be agents such that $v_i(A_j) > 0$ and $v_i(A_k) < 0$.
Obtain allocation $B$ by moving a good from $A_j$ to $A_k$, i.e.,
for some $g \in A_j$ such that $v_i(g) > 0$, let $B_j \defeq A_j \setminus \{g\}$,
$B_k \defeq A_k \cup \{g\}$, and $B_{\ell} \defeq A_{\ell}$ for all $\ell \in [n] \setminus \{j, k\}$.
Then $B$ is also $F$-fair to $i$.
\end{lemma}
\begin{proof}
Since the bundles of agents $[n] \setminus \{j, k\}$ didn't change, agent $i$ doesn't $F$-envy them.
Since $A$ is $F$-fair to $i$, and $v_i(A_j) > 0$, we get $v_i(A_i) \ge 0$.
Then $i$ doesn't envy $k$ because $v_i(B_k)/w_k \le 0 \le v_i(A_i)/w_i = v_i(B_i)/w_i$.
Since $v_i(B_j) = v_i(A_j) - 1$, one can verify that $i$ doesn't $F$-envy agent $j$.
Thus, $B$ is $F$-fair to agent $i$.
\end{proof}

\begin{lemma}[min-$F$-share $\fimplies$ epistemic-$F$]
\label[lemma]{thm:impl:tribool:minfs-to-epistemic}
\label[lemma]{thm:impl:tribool:minfs-to-epistemic:ef}
\label[lemma]{thm:impl:tribool:minfs-to-epistemic:ef1-goods}
\label[lemma]{thm:impl:tribool:minfs-to-epistemic:ef1-chores}
\label[lemma]{thm:impl:tribool:minfs-to-epistemic:efx-goods}
\label[lemma]{thm:impl:tribool:minfs-to-epistemic:efx-chores}
Let $Q \subseteq \{-1, 0, 1\}$.
Consider a fair division instance $\fdInst{[n]}{[m]}{(v_i)_{i=1}^n}{w}$ where for some agent $i$,
$v_i$ is additive and $v_i(t) \in Q$ for all $t \in [m]$.
Let $F \in \{\mathrm{EF},\allowbreak \mathrm{EF1},\allowbreak \mathrm{EFX}\}$.
If allocation $A$ is min-$F$-share-fair to agent $i$, then it is also epistemic-$F$-fair to $i$
if $F$ = EF, $Q = \{-1, 0\}$, or $Q = \{0, 1\}$.
\end{lemma}
\begin{proof}
Since $A$ is min-$F$-share-fair to $i$, there exists an allocation $B$ such that
$v_i(B_i) \le v_i(A_i)$ and $B$ is $F$-fair to $i$.
By repeatedly applying the transformation of \cref{thm:tribool:equalizing-shift},
we can assume \wLoG{} that either $v_i(B_j) \ge 0$ for all $j$ or $v_i(B_j) \le 0$ for all $j$.

Since $v_i([m] \setminus A_i) \le v_i([m] \setminus B_i)$, we can redistribute $[m] \setminus A_i$
to get an allocation $C$ such that $C_i = A_i$ and $v_i(C_i) \le v_i(B_i)$.
This is because we can fuse items in $[m] \setminus A_i$ such that we either get rid of all goods or all chores.
By \cref{thm:tribool:ef-monotonicity}, $C$ is $F$-fair to agent $i$.
Thus, $C$ is agent $i$'s epistemic-$F$-certificate for $A$.
\end{proof}

We now look at non-additive valuations.

\begin{lemma}[MMS = APS, \citet{kulkarni2024approximating}]
\label[lemma]{thm:impl:mms-to-aps-matroid}
For a fair division instance $\fdInst{[n]}{[m]}{(v_i)_{i=1}^n}{\eqEnt}$,
if $v_i$ is submodular for some agent $i$, and $v_i$ has binary marginals
(i.e., $v_i(g \mid X) \in \{0, 1\}$ for some $X \subseteq [m]$ and $g \in [m] \setminus X$),
then $\APS_i = \MMS_i$.
\end{lemma}

\begin{lemma}
\label[lemma]{thm:binary-inc-subset}
Let $f: 2^{[m]} \to \mathbb{R}$ such that $f$ has tribool marginals, i.e.,
for all $X \subseteq [m]$ and $j \in [m] \setminus X$, we have $f(j \mid X) \in \{-1, 0, 1\}$.
For any disjoint sets $S, T \subseteq [m]$ such that $f(T \mid S) > 0$,
there exists $T' \subseteq T$ such that $f(T' \mid S) = 1$.
\end{lemma}
\begin{proof}
Let $T = \{g_1, \ldots, g_k\}$.
For any $0 \le j \le k$, define $T_j \defeq \{g_1, \ldots, g_j\}$
and $\alpha_j \defeq f(T_j \mid S)$.
Then $\alpha_0 = f(\emptyset \mid S) = 0$, and $\alpha_k = f(T \mid S) > 0$.
Let $t$ be the smallest value in $[k]$ such that $\alpha_t > 0$. Then
$\alpha_t = f(T_t \mid S) = \alpha_{t-1} + f(g_t \mid S \cup T_{t-1})$.
Since $\alpha_{t-1} \le 0$, we get $\alpha_t \le f(g_t \mid S \cup T_{t-1}) \le 1$.
Since marginals are in $\{-1, 0, 1\}$, we get $\alpha_t = 1$, so $f(T_t \mid S) = 1$.
\end{proof}

\begin{lemma}[PROPm vs PROPavg vs PROPx for binary subadd]
\label[lemma]{thm:impl:propm-to-propavg-propx-binary-subadd}
Let $\fdInst{[n]}{[m]}{(v_i)_{i=1}^n}{w}$ be a fair division instance
where $v_i$ has tribool marginals for some agent $i$, i.e.,
for all $X \subseteq [m]$ and $j \in [m] \setminus X$,
we have $v_i(j \mid X) \in \{-1, 0, 1\}$.
Suppose agent $i$ is PROPm-satisfied by allocation $A$.
Then $A$ is also PROPavg-fair to $i$.
Moreover, if $v_i$ is subadditive and $w_i = 1/n$,
then $A$ is also PROPx-fair to $i$.
\end{lemma}
\begin{proof}
If $A$ is PROP-fair to $i$, then we are done, so assume otherwise.
For $j \neq i$, recall the definition of $\tau_j$ from \cref{defn:propm}.
We have two cases:

\textbf{Case 1}: $\tau_j = 0$ for all $j \neq i$.
\\ Then $A$ is PROPavg by definition. Let $v_i$ be subadditive.
Then for all $j \in [n] \setminus \{i\}$, we have
\begin{align*}
\tau_j = 0 &\iff v_i(A_j \mid A_i) = 0
    \wrapIfTwoCols\implies v_i(A_j) \le v_i(A_i \cup A_j) = v_i(A_i).
\end{align*}
Hence,
\[ v_i([m]) \le \sum_{j=1}^n v_i(A_i) \le nv_i(A_i). \]
Therefore, if $w_i = 1/n$, then $A$ is PROP-fair to $i$.

\textbf{Case 2}: $\tau_j > 0$ for some $j \neq i$.
\\ By \cref{thm:binary-inc-subset}, we get $\tau_j \in \{0, 1\}$ for all $j$.
Since $A$ is PROPm-fair to $i$, we get $v_i(A_i) + 1 > w_iv_i([m])$.
The average of $\{\tau_j: j \neq i, \tau_j > 0\}$ is 1, so $A$ is also PROPavg-fair to $i$.
If $v_i([m] \setminus A_i \mid A_i) = 0$, then $v_i(A_i) = v_i([m]) > w_iv_i([m])$, so $A$ is PROP-fair to $i$.
Otherwise, for every $S \subseteq [m] \setminus A_i$ such that $v_i(S \mid A_i) > 0$,
we get $v_i(A_i \cup S) \ge v_i(A_i) + 1 > w_iv_i([m])$, so $A$ is PROPx-fair to $i$.
\end{proof}

\begin{lemma}[M1S $\fimplies$ PROPx for binary subadditive]
\label[lemma]{thm:impl:m1s-to-propx-binary-subadd}
Let $\fdInst{[n]}{[m]}{(v_i)_{i=1}^n}{w}$ be a fair division instance where
$v_i$ is subadditive and has $\{-1, 0, 1\}$ marginals,
i.e., $v_i(g \mid X) \in \{-1, 0, 1\}$ for all $X \subseteq [m]$ and $g \in [m] \setminus X$.
Let allocation $A$ be M1S-fair to agent $i$.
Then $A$ is also PROPx-fair to agent $i$ if at least one of the following hold:
\begin{tightenum}
\item $n=2$.
\item $w_i < 1/(n-1)$ and $v_i$ has $\{0, 1\}$ marginals.
\item $v_i$ has $\{-1, 0\}$ marginals.
\end{tightenum}
\end{lemma}
\begin{proof}
Let $B$ be agent $i$'s M1S-certificate for $A$. Let $k \defeq v_i(B_i)$.
Then $v_i(A_i) \ge k$ and $B$ is EF1-fair to $i$.
If $v_i(A_i) \ge w_iv_i([m])$, then $A$ is PROP-fair to $i$, and we are done.
Now let $k \le v_i(A_i) < w_iv_i([m])$.
Define sets $J_0$, $J_G$, and $J_C$:
\ifColsTwo\else\let\tempWrap\relax\fi
\begin{align*}
J_0 &\defeq \bigg\{j \in [n] \setminus \{i\}: \frac{v_i(B_i)}{w_i} \ge \frac{v_i(B_j)}{w_j} \bigg\},
\\ J_G &\defeq \bigg\{j \in [n] \setminus \{i\}: \frac{v_i(B_j \setminus \{g\})}{w_j}
    \tempWrap\le \frac{v_i(B_i)}{w_i} < \frac{v_i(B_j)}{w_j} \text{ for some } g \in B_j \bigg\},
\\ J_C &\defeq \bigg\{j \in [n] \setminus \{i\}: \frac{v_i(B_i)}{w_i} < \frac{v_i(B_j)}{w_j}
    \tempWrap\le \frac{v_i(B_i \setminus \{c\})}{w_i} \text{ for some } c \in B_i \bigg\}.
\end{align*}
Since $B$ is EF1-fair to $i$, we have $J_0 \cup J_G \cup J_C = [n] \setminus \{i\}$.
Also, $J_0$ and $J_G \cup J_C$ are disjoint.

If all items are goods, then $J_C = \emptyset$.
If all items are chores, then $J_G = \emptyset$.
If $n = 2$, then $J_G = \emptyset$ or $J_C = \emptyset$.
Hence, if at least one of the three conditions of the lemma hold,
then either $J_G = \emptyset$, or $J_C = \emptyset$ and $w_i < 1/(n-1)$.
We first show that, under these conditions, $w_iv_i([m]) < k + 1$.

\textbf{Case 1}: $J_C = \emptyset$ and $w_i < 1/(n-1)$.
\\ Pick $j \in [n] \setminus \{i\}$.
If $j \in J_0$, then $v_i(B_j) \le kw_j/w_i$.
Otherwise, $\exists g \in B_j$ such that $v_i(B_j \setminus \{g\}) \le kw_j/w_i$.
So, $v_i(B_j) = v_i(B_j \setminus \{g\}) + v_i(g \mid B_j \setminus \{g\}) \le k(w_j/w_i) + 1$.
Hence,
\[ v_i([m]) \le \sum_{j=1}^n v_i(B_j) \le (n-1) + k/w_i. \]
So, $w_iv_i([m]) \le k + w_i(n-1) < k + 1$ (since $w_i < 1/(n-1)$).

\textbf{Case 2}: $J_G = \emptyset$
\\ Pick $j \in [n] \setminus \{i\}$.
If $j \in J_0$, then $v_i(B_j) \le kw_j/w_i$.
Otherwise, $\exists c \in B_i$ such that $v_i(B_i \setminus \{c\})/w_i \ge v_i(B_j)/w_j$.
$v_i(B_i \setminus \{c\}) = v_i(B_i) - v_i(c \mid B_i \setminus \{c\}) \le k + 1$.
So, $v_i(B_j) \le (k+1)w_j/w_i$. Hence,
\[ v_i([m]) \le \sum_{j=1}^n v_i(B_j) \le \frac{k+1}{w_i} - 1. \]
Hence, $w_iv_i([m]) \le k+1 - w_i < k + 1$.

Now we will show that $A$ is PROPx-fair to $i$.
Note that $k \le v_i(A_i) < w_iv_i([m]) < k + 1$, so $v_i(A_i) = k$.
If $S \subseteq [m] \setminus A_i$ such that $v_i(S \mid A_i) > 0$,
then $v_i(A_i \cup S) \ge k + 1 > w_iv_i([m])$.
If $S \subseteq A_i$ such that $v_i(S \mid A_i \setminus S) < 0$,
then $v_i(A_i \setminus S) = v_i(A_i) - v_i(S \mid A_i \setminus S) \ge k + 1 > w_iv_i([m])$.
Hence, $A$ is PROPx-fair to agent $i$.
\end{proof}

\ifVerbose
\subsection{Unit-Demand Valuations}
\label{sec:impls-extra:unit-demand}

A function $v: 2^{[m]} \to \mathbb{R}_{\ge 0}$ is called \emph{unit-demand} if $v(\emptyset) = 0$,
and for any non-empty set $X \subseteq M$, we have $v(X) = \max_{g \in X} v(\{g\})$.
Such functions capture scenarios where agents want just one good, e.g., when the goods are houses.
Unit-demand functions are known to be submodular and cancelable.

\begin{table}[htb]
\centering
\caption[Unit-demand implications]{Implications among fairness notions when
valuations are unit-demand, marginals are non-negative, and entitlements are equal.}
\label{table:impls-unit-demand}
\small
\begin{tabular}{cHHHHcH}
\toprule implication & \footnotesize valuations & \footnotesize marginals & \footnotesize $n=2$ & \footnotesize entitlements & proof &
\\\midrule MEFS $\fimplies$ EF
    & unit-demand & $\ge 0$ & -- & equal
    & \cref{thm:ud:ef} & trivial
\\[\defaultaddspace] MMS $\fimplies$ APS
    & unit-demand & $\ge 0$ & -- & equal
    & \cref{thm:ud:mms,thm:ud:aps} & \textbf{new}
\\[\defaultaddspace] PMMS $\fimplies$ PAPS
    & unit-demand & $\ge 0$ & -- & equal
    & \cref{thm:ud:mms,thm:ud:aps} & \textbf{new}
\\[\defaultaddspace] GMMS $\fimplies$ GAPS
    & unit-demand & $\ge 0$ & -- & equal
    & \cref{thm:ud:mms,thm:ud:aps} & \textbf{new}
\\[\defaultaddspace] M1S $\fimplies$ APS
    & unit-demand & $\ge 0$ & -- & equal
    & \cref{thm:ud:m1s,thm:ud:aps} & \textbf{new}
\\[\defaultaddspace] EF1 $\fimplies$ GAPS
    & unit-demand & $\ge 0$ & -- & equal
    & \cref{thm:ud:m1s,thm:ud:aps} & \textbf{new}
\\ \bottomrule
\end{tabular}
\end{table}

\begin{remark}
\label[remark]{thm:ud:ef}
Consider a fair division instance $\fdInst{[n]}{[m]}{(v_i)_{i=1}^n}{\eqEnt}$ over goods
where $v_i$ is unit-demand for some agent $i$.
\WLoG, let $v_i(1) \ge \ldots \ge v_i(m)$.
Then agent $i$ is envy-free in allocation $A$ iff $v_i(A_i) = v_i(1)$.
\end{remark}

\begin{lemma}
\label[lemma]{thm:ud:m1s}
Consider a fair division instance $\fdInst{[n]}{[m]}{(v_i)_{i=1}^n}{\eqEnt}$ over goods
where $v_i$ is unit-demand for some agent $i$.
\WLoG, let $v_i(1) \ge \ldots \ge v_i(m)$.
Then allocation $A$ is M1S-fair to agent $i$ iff $v_i(A_i) \ge v_i(n)$.
\end{lemma}
\begin{proof}
Suppose $A$ is M1S-fair to agent $i$. Let $B$ be agent $i$'s M1S-certificate for $A$.
Let $G \defeq \{g \in [m]: v_i(g) > v_i(B_i)\}$.
Then $B$ is EF1-fair to agent $i$, so every agent $j \in [n] \setminus \{i\}$
has at most one good from $G$, so $|G| \le n-1$.
If $v_i(B_i) < v_i(n)$, then $G \supseteq [n]$, which is a contradiction.
Hence, $v_i(A_i) \ge v_i(B_i) \ge v_i(n)$.

Suppose $v_i(A_i) \ge v_i(n)$. Let $B$ be an allocation where
each agent in $[n] \setminus \{i\}$ gets exactly one good from $[n-1]$,
and agent $i$ gets $[m] \setminus [n-1]$.
Then $B$ is EF1-fair to $i$ and $v_i(A_i) \ge v_i(n) = v_i(B_i)$.
Hence, $B$ is an M1S-certificate for agent $i$ for allocation $A$.
Hence, $A$ is M1S-fair to agent $i$.
\end{proof}

\begin{lemma}
\label[lemma]{thm:ud:mms}
Consider a fair division instance $\fdInst{[n]}{[m]}{(v_i)_{i=1}^n}{\eqEnt}$ over goods
where $v_i$ is unit-demand for some agent $i$.
\WLoG, let $v_i(1) \ge \ldots \ge v_i(m)$.
Then agent $i$'s maximin share is $v_i(n)$.
\end{lemma}
\begin{proof}
In every $n$-partition $P$ of $[m]$, some bundle $P_k$ does not contain a good from $[n-1]$.
Then $v_i(P_k) \le v_i(n)$. Hence, $\MMS_i \le v_i(n)$.
Let $P$ be a partition where $P_j = \{j\}$ for $j \in [n-1]$ and $P_n = [m] \setminus [n-1]$.
Then $\min_{j=1}^n v_i(P_j) = v_i(n)$. Hence, $\MMS_i = v_i(n)$.
\end{proof}

\begin{lemma}
\label[lemma]{thm:ud:aps}
Consider a fair division instance $\fdInst{[n]}{[m]}{(v_i)_{i=1}^n}{\eqEnt}$ over goods
where $v_i$ is unit-demand for some agent $i$.
\WLoG, let $v_i(1) \ge \ldots \ge v_i(m)$.
Then agent $i$'s anyprice share is $v_i(n)$.
\end{lemma}
\begin{proof}
Set the price of the first $n-1$ goods to $1/(n-1)$ each
and price the remaining goods at 0.
Then $\APS_i \le v_i([m] \setminus [n-1]) = v_i(n)$.
For any non-negative price vector, some item from $[n]$ is affordable,
so $\APS_i \ge v_i(n)$.
\end{proof}

\begin{lemma}
\label[lemma]{thm:ud:prop1}
Consider a fair division instance $\fdInst{[n]}{[m]}{(v_i)_{i=1}^n}{\eqEnt}$ over goods
where $v_i$ is unit-demand for some agent $i$.
\WLoG, let $v_i(1) \ge \ldots \ge v_i(m)$.
Then every allocation is PROP1-fair to agent $i$.
\end{lemma}
\begin{proof}
$v_i(A_i \cup \{1\}) = v_i([m]) > v_i([m])/n$.
\end{proof}
\fi

\section{Non-Implications (Additive Valuations)}
\label{sec:cex-add-extra}

\subsection{Trivial Examples}
\label{sec:cex-extra:trivial}

\begin{example}[single item]
\label[example]{cex:single-item}
\label[example]{cex:single-item:goods}
\label[example]{cex:single-item:chores}
Consider a fair division instance with $n$ agents and one item
(which is either a good to everyone or a chore to everyone).
Then every allocation is EFX, EF1, APS, MMS, PROPx, GAPS, GMMS, PAPS, PMMS,
EEFX, EEF1, MXS, M1S, PROPm, and PROP1,
but not EF or PROP or GPROP or PPROP or EEF or MEFS.
\end{example}

\begin{lemma}
\label[example]{cex:share-vs-envy-goods}
Consider a fair division instance $\fdInst{[n]}{[m]}{(v_i)_{i=1}^n}{\eqEnt}$
with $n \ge 3$, $m = 2n-1$, and identical additive valuations,
where each item has value 1 to each agent.
Let $A$ be an allocation where agent $n$ gets $n$ goods,
and all other agents get 1 good each.
Then this allocation is APS+MMS+EEFX+PROPx, but not EF1.
\end{lemma}
\begin{proof}
By \cref{thm:impl:aps-to-pess,thm:impl:prop-to-aps},
$1 = \MMS_i \le \APS_i \le v([m])/n = 2 - 1/n$.
Hence, $APS_i = 1$, since the APS is the value of some bundle.
Hence, $A$ is APS+MMS+PROPx.
Agent $n$'s EEFX-certificate for $A$ is $A$ itself.
For $i \neq n$, agent $i$'s EEFX-certificate is $B$,
where $|B_i| = 1$ and $|B_j| = 2$ for $j \neq i$.
$A$ is not EF1: agents $[n-1]$ EF1-envy agent $n$.
\end{proof}

\begin{lemma}
\label[lemma]{cex:share-vs-envy-chores}
Consider a fair division instance $\fdInst{[n]}{[m]}{(v_i)_{i=1}^n}{\eqEnt}$
with $n \ge 3$, identical additive disutilities,
and $m = n+1$ chores, each of disutility 1.
Let $A$ be an allocation where agents 1 and 2 get 2 chores each,
agent $n$ gets 0 chores, and the remaining agents get 1 chore each.
Then this allocation is APS+MMS+EEFX+PROPx, but not EF1.
\end{lemma}
\begin{proof}
$-2 = \MMS_i \le \APS_i \le v([m])/n = - 1 - 1/n$.
Hence, $APS_i = -2$, since the APS is the value of some bundle.
Hence, $A$ is APS+MMS+PROPx.

Agents $[n] \setminus [2]$ do not EFX-envy anyone in $A$.
For $i \in \{1, 2\}$, agent $i$'s epistemic-EFX-certificate can be obtained by
transferring a chore from agent $3-i$ to agent $n$.
Hence, $A$ is epistemic EFX.
$A$ is not EF1 because agents 1 and 2 EF1-envy agent $n$.
\end{proof}

\subsection{From EEF, MEFS, PROP}
\label{sec:cex-extra:from-eef-mefs-prop}

\begin{lemma}[EEF $\nfimplies$ EF1]
\label[lemma]{cex:eef-not-ef1}
\label[lemma]{cex:eef-not-ef1:positive-bival}
\label[lemma]{cex:eef-not-ef1:binary}
\label[lemma]{cex:eef-not-ef1:negative-bival}
\label[lemma]{cex:eef-not-ef1:neg-binary}
Let $0 \le 2a < b$. Let $f_1, f_2, f_3: 2^{[12]} \to \mathbb{R}_{>0}$ be additive sets functions:

\noindent
{\renewcommand{\tabcolsep}{4pt}
\begin{tabular}{c|cccc|cccc|cccc}
& 1 & 2 & 3 & 4 & 5 & 6 & 7 & 8 & 9 & 10 & 11 & 12
\\ \hline $f_1$ & $a$ & $a$ & $b$ & $b$ & $b$ & $b$ & $b$ & $b$ & $a$ & $a$ & $a$ & $a$
\\ $f_2$ & $a$ & $a$ & $a$ & $a$ & $a$ & $a$ & $b$ & $b$ & $b$ & $b$ & $b$ & $b$
\\ $f_3$ & $b$ & $b$ & $b$ & $b$ & $a$ & $a$ & $a$ & $a$ & $a$ & $a$ & $b$ & $b$
\end{tabular}
}

Let $t \in \{-1, 1\}$ and let $\Ical \defeq \fdInst{[3]}{[12]}{(v_i)_{i=1}^3}{\eqEnt}$
be an instance where $v_i \defeq tf_i$ for all $i \in [3]$.
Then allocation
$A \defeq ([4], [8] \setminus [4], [12] \setminus [8])$ is EEF+PROP but not EF1.
\end{lemma}
\begin{proof}
For $t = 1$, agent 1 EF1-envies agent 2 in $A$,
and for $t = -1$, agent 1 EF1-envies agent 3 in $A$.
$B = ([4], \{5, 6, 9, 10\}, \{7, 8, 11, 12\})$ is agent 1's EEF-certificate.
A similar argument holds for agents 2 and 3 too.
\end{proof}

\begin{example}[PROP $\nfimplies$ MEFS]
\label[example]{cex:prop-not-mefs-goods}
Consider a fair division instance with 3 equally-entitled agents
having additive valuations over 3 goods:

\begin{tabular}{c|ccc}
& 1 & 2 & 3
\\ \hline $v_1$ & 10 & 20 & 30
\\ $v_2$ & 20 & 10 & 30
\\ $v_3$ & 10 & 20 & 30
\end{tabular}

Then the allocation $(\{2\}, \{1\}, \{3\})$ is PROP, but no allocation is MEFS
(every agent's min EF share is 30).
\end{example}

\begin{example}[PROP $\nfimplies$ MEFS]
\label[example]{cex:prop-not-mefs-chores}
Consider a fair division instance with 3 equally-entitled agents
having additive disutilities over 3 chores:

\begin{tabular}{c|ccc}
& 1 & 2 & 3
\\ \hline $-v_1$ & 30 & 20 & 10
\\ $-v_2$ & 20 & 30 & 10
\\ $-v_3$ & 30 & 20 & 10
\end{tabular}

Then the allocation $(\{2\}, \{1\}, \{3\})$ is PROP, but no allocation is MEFS
(every agent's min EF share is $-10$).
\end{example}

\begin{lemma}[MEFS $\nfimplies$ EEF]
\label[lemma]{cex:mefs-not-eef-goods}
Consider a fair division instance with 3 equally-entitled agents
having additive valuations over 6 goods:

\begin{tabular}{c|cccccc}
& 1 & 2 & 3 & 4 & 5 & 6
\\ \hline $v_1$ & 20 & 20 & 20 & 10 & 10 & 10
\\ $v_2$, $v_3$ & 20 & 10 & 10 &  1 &  1 &  1
\end{tabular}

Then the allocation $A \defeq (\{4, 5, 6\}, \{1\}, \{2, 3\})$ is MEFS, but no allocation is epistemic EF.
\end{lemma}
\begin{proof}
Agents 2 and 3 are envy-free in $A$.
Agent 1 has $B \defeq (\{1, 4\}, \{2, 5\}, \{3, 6\})$ as her MEFS-certificate for $A$.
Hence, $A$ is MEFS.

Suppose an epistemic EF allocation $X$ exists.
Let $Y^{(i)}$ be each agent $i$'s epistemic-EF-certificate.
For agent 2 to be envy-free in $Y^{(2)}$,
we require $Y^{(2)}_2 \supseteq \{1\}$ or $Y^{(2)}_2 \supseteq \{2, 3\}$.
Similarly, $Y^{(3)}_3 \supseteq \{1\}$ or $Y^{(3)}_3 \supseteq \{2, 3\}$.
Since $Y^{(i)}_i = X_i$ for all $i$, we get $X_2 \cup X_3 \supseteq \{1, 2, 3\}$.
Hence, $X_1 \subseteq \{4, 5, 6\}$.
But then no epistemic-EF-certificate exists for agent 1 for $X$,
contradicting our assumption that $X$ is epistemic EF.
Hence, no epistemic EF allocation exists.
\end{proof}

\begin{example}[MEFS $\nfimplies$ EEF]
\label[example]{cex:mefs-not-eef-chores}
Consider a fair division instance with 3 equally-entitled agents
having additive disutilities over 6 chores:

\begin{tabular}{c|cccccc}
& 1 & 2 & 3 & 4 & 5 & 6
\\ \hline $-v_1$ & 20 & 20 & 20 & 10 & 10 & 10
\\ $-v_2$, $-v_3$ & 20 & 10 & 10 & 10 & 10 & 10
\end{tabular}

Then the allocation $A \defeq (\{4, 5, 6\}, \{1\}, \{2, 3\})$ is MEFS
(agents 2 and 3 are EF, agent 1's MEFS-certificate is $(\{1, 4\}, \{2, 5\}, \{3, 6\})$).
Agent 1 is not EEF-satisfied by $A$.
\end{example}

\begin{lemma}[MEFS $\nfimplies$ EEF1]
\label[lemma]{cex:mefs-not-eef1-chores}
Consider an instance with 3 equally-entitled agents
having additive disutilities over 12 chores.
$v_1(1) = v_1(2) = v_1(3) = 70$ and $v_1(c) = 10$ for all $c \in [12] \setminus [3]$.
Agents 2 and 3 have disutility 10 for each chore.
Then $A \defeq ([12] \setminus [3], [2], \{3\})$ is a MEFS+PROP allocation
where agent 1 is not EEF1-satisfied.
\end{lemma}
\begin{proof}
$\PROP_1 = -100$ and $\PROP_2 = \PROP_3 = -40$.
$\MEFS_1 \le -100$ because of the allocation $(\{1, 4, 5, 6\},\allowbreak \{2, 7, 8, 9\},\allowbreak \{3, 10, 11, 12\})$.
$\MEFS_i \le -40$ for $i \in \{2, 3\}$ because of the allocation
$([4], [8] \setminus [4], [12] \setminus [8])$.
Agent 1 has disutility $90$ in $A$, so $A$ is MEFS-fair and PROP-fair to agent 1.
Agents 2 and 3 have disutility at most $20$ in $A$, so $A$ is MEFS-fair and PROP-fair to them.

Agent 1 is not EEF1-satisfied by $A$, since in any EEF1-certificate $B$,
some agent $j \in \{2, 3\}$ receives at most one chore of value $70$,
and agent 1 would EF1-envy $j$.
\end{proof}

\subsection{Two Equally-Entitled Agents}
\label{sec:cex-extra:2-eqEnt}

\begin{example}[EFX $\nfimplies$ MMS]
\label[example]{cex:efx-not-mms}
\label[example]{cex:efx-not-mms:goods}
\label[example]{cex:efx-not-mms:chores}
Let $t \in \{-1, 1\}$.
Consider a fair division instance with 2 equally-entitled agents having
an identical additive valuation function $v$ over 5 items.
$v(1) = v(2) = 3t$ and $v(3) = v(4) = v(5) = 2t$.
Then allocation $A \defeq (\{1, 3\}, \{2, 4, 5\})$ is EFX.
The MMS is $6t$, since $P = (\{3t, 3t\}, \{2t, 2t, 2t\})$ is an MMS partition.
But in $A$, some agent doesn't get her MMS.
\end{example}

\begin{example}[EF1 $\nfimplies$ PROPX or MXS]
\label[example]{cex:ef1-not-propx-mxs}
\label[example]{cex:ef1-not-propx-mxs:goods}
\label[example]{cex:ef1-not-propx-mxs:chores}
Let $t \in \{-1, 1\}$.
Consider an instance with 2 equally-entitled agents
having an identical additive valuation function $v$ over 5 items:
$v(1) = v(2) = 4t$ and $v(3) = v(4) = v(5) = t$.
Allocation $(\{1\}, [5] \setminus \{1\})$ is EF1 but not PROPx and not MXS.
\end{example}

\begin{lemma}[PROPx $\nfimplies$ M1S]
\label[lemma]{cex:propx-not-m1s}
\label[lemma]{cex:propx-not-m1s:goods}
\label[lemma]{cex:propx-not-m1s:chores}
Let $t \in \{-1, 1\}$ and $0 < \eps < 1/2$.
Consider a fair division instance with 2 equally-entitled agents having
an identical additive valuation function $v$ over 4 items.
Let $v(4) = (1+2\eps)t$ and $v(j) = t$ for $j \in [3]$.
Then allocation $A \defeq (\{4\}, [3])$ is PROPx but not M1S.
\end{lemma}
\begin{proof}
$v([m])/2 = (2+\eps)t$, so $A$ is PROPx.
For $t = 1$, in any allocation $B$ where agent 1 doesn't EF1-envy agent 2, she must have at least 2 goods.
But $v(A_1) = 1+2\eps$, so agent 1 doesn't have an M1S-certificate for $A$. Hence, $A$ is not M1S.
For $t = -1$, in any allocation $B$ where agent 2 doesn't EF1-envy agent 1, she must have at most 2 chores.
But $v(A_2) = -3$, so agent 1 doesn't have an M1S-certificate for $A$. Hence, $A$ is not M1S.
\end{proof}

\begin{example}[MXS $\nfimplies$ PROPx for $n=2$, \citet{caragiannis2022existence}]
\label[example]{cex:mxs-not-propx-n2}
\label[example]{cex:mxs-not-propx-n2:goods}
\label[example]{cex:mxs-not-propx-n2:chores}
Let $t \in \{-1, 1\}$.
Consider a fair division instance with 2 equally-entitled agents
having identical additive valuations over 7 items:
the first 2 items of value $4t$ and the last 5 items of value $t$.
Then the allocation $A = (\{1, 3\}, \{2, 4, 5, 6, 7\})$ is not PROPx or EFX,
but it is MXS because the agents have $([7] \setminus [2], [2])$
and $([2], [7] \setminus [2])$ as their MXS-certificates for $A$.
\end{example}

\begin{lemma}[M1S $\nfimplies$ PROP1]
\label[lemma]{cex:m1s-not-prop1}
\label[lemma]{cex:m1s-not-prop1:goods}
\label[lemma]{cex:m1s-not-prop1:chores}
Consider a fair division instance with 2 equally-entitled agents
having an identical additive valuation function $v$ over 9 items.
Let $t \in \{-1, 1\}$ and $v(9) = 4t$ and $v(j) = t$ for $j \in [8]$.
Then allocation $A \defeq (\{9\}, [8])$ is M1S but not PROP1.
\end{lemma}
\begin{proof}
$v([9])/2 = 6t$. Let $B \defeq ([4], [9] \setminus [4])$.
For $t = 1$, $B$ is agent 1's M1S-certificate for $A$,
but $A$ is not PROP1-fair to agent 1.
For $t = -1$, $B$ is agent 2's M1S-certificate for $A$,
but $A$ is not PROP1-fair to agent 2.
\end{proof}

\subsection{Three Equally-Entitled Agents}
\label{sec:cex-extra:3-eqEnt}

\begin{example}[GAPS $\nfimplies$ PROPx]
\label[example]{cex:gaps-not-propx}
\label[example]{cex:gaps-not-propx:additive}
\label[example]{cex:gaps-not-propx:unit-demand}
Consider a fair division instance with 3 equally-entitled agents
having identical additive or unit-demand valuations. There are 2 goods of values 50 and 10.
In every allocation, some agent doesn't get any good, and that agent is not PROPx-satisfied.
The allocation where the first agent gets the good of value 50
and the second agent gets the good of value 10 is a groupwise APS allocation
(set the price of the goods to $1.1$ and $0.9$).
\end{example}

\begin{lemma}[APS $>$ MMS]
\label[lemma]{thm:aps-gt-mms}
\label[lemma]{thm:aps-gt-mms:goods}
\label[lemma]{thm:aps-gt-mms:chores}
Let $t \in \{-1, 1\}$.
Consider a fair division instance with 3 equally-entitled agents
having identical additive valuations over 15 items. The items' values are
$65t$, $31t$, $31t$, $31t$, $23t$, $23t$, $23t$, $17t$, $11t$, $7t$, $7t$, $7t$, $5t$, $5t$, $5t$.
Then the AnyPrice share is at least $97t$, the proportional share is $97t$,
and the maximin share is less than $97t$.
\end{lemma}
\begin{proof}
For $t = 1$, this follows from Lemma C.1 of \citet{babaioff2023fair}.
For $t = -1$, a similar argument tells us that the AnyPrice share is at least $-97$.
If the maximin share is at least $-97$, then there must exist a partition $P$ of the chores
where each bundle has disutility 97. But then $P$ would prove that the maximin share
in the corresponding goods instance is at least 97, which is a contradiction.
Hence, for $t = -1$, the maximin share is less than $-97$.
\end{proof}

\begin{example}[GMMS $\nfimplies$ APS, \citet{babaioff2023fair}]
\label[example]{cex:gmms-not-aps}
\label[example]{cex:gmms-not-aps:goods}
\label[example]{cex:gmms-not-aps:chores}
For the fair division instance in \cref{thm:aps-gt-mms},
the leximin allocation is GMMS (since on restricting to any subset of agents,
the resulting allocation is still leximin, and is therefore MMS).
However, no APS allocation exists, because APS $>$ MMS,
and the minimum value across all bundles is at most the MMS.
\end{example}

\begin{example}[PMMS $\nfimplies$ MMS, Example 4.4 of \citet{caragiannis2019unreasonable}]
\label[example]{cex:pmms-not-mms}
\label[example]{cex:pmms-not-mms:goods}
\label[example]{cex:pmms-not-mms:chores}
Let $t \in \{-1, 1\}$.
Consider an instance with 3 equally-entitled agents
having an identical valuation function $v$ over 7 items where
$v(1) = 6t$, $v(2) = 4t$, $v(3) = v(4) = 3t$, $v(5) = v(6) = 2t$, $v(7) = t$.
Each agent's maximin share is $7t$ ($(\{1, 7\}, \{2, 3\}, \{4, 5, 6\})$ is a maximin partition).
Allocation $(\{1\}, \{3, 4, 5\}, \{2, 6, 7\})$ is PMMS but not MMS.
\end{example}

\begin{example}[APS $\nfimplies$ PROPm]
\label[example]{cex:aps-not-propm}
\label[example]{cex:aps-not-propm:goods}
\label[example]{cex:aps-not-propm:chores}
Consider a fair division instance with 3 equally-entitled agents
having an identical additive valuation function $v$ over 6 goods:
$v(1) = 60$, $v(2) = 30$, and $v(3) = v(4) = v(5) = v(6) = 10$.
The allocation $A \defeq (\{2\}, \{3, 4, 5\}, \{1, 6\})$ is APS+MMS, since the MMS is 30,
and the APS is at most 30 because of the price vector $(4, 3, 1, 1, 1, 1)$.
However, $A$ is not PROPm-fair to agent 1, because the proportional share is $130/3 > 40$.
\end{example}

\begin{example}[APS $\nfimplies$ PROP1]
\label[example]{cex:aps-not-prop1-chores}
Consider an instance with 3 equally-entitled agents
having an identical additive valuation function $v$ over 6 chores:
the first chore has disutility 18 (large chore) and the remaining chores have disutility 3 each (small chores).
Then $X \defeq ([6] \setminus \{1\}, \{1\}, \emptyset)$ is MMS+APS,
since the MMS is $-18$, and the APS is at most $-18$ due to the price vector $(1, 0, 0, 0, 0, 0)$.
$X$ is not PROP1-fair to agent 1, since the proportional share is $-11$,
and agent 1's disutility in $X$ after removing any chore is $12$.
\ifVerbose
$X$ is not EEF1-fair to agent 1 because even after redistributing chores among the remaining agents,
someone will always have no chores.
\fi
\end{example}

\begin{lemma}[PROPm doesn't exist for mixed manna]
\label[lemma]{cex:propm-mixed-manna}
Consider an instance $\fdInst{[3]}{[6]}{(v_i)_{i=1}^3}{\eqEnt}$
where agents have identical additive valuations, and items have values
$(-3,\allowbreak -3,\allowbreak -3,\allowbreak -3,\allowbreak -3,\allowbreak 3\eps)$, where $0 < \eps < 1/2$.
Then there exists an EFX+GMMS+GAPS allocation but no PROPm allocation.
\end{lemma}
\begin{proof}
The proportional share is $v([m])/3 = -5 + \eps$.
\WLoG, assume agent 1 receives the most number of chores,
and agent 3 receives the least number of chores.
Then agent 1 has at least 2 chores, and agent 3 has at most 1.

\textbf{Case 1}: agent 1 receives at least 3 chores.
\\ Then even after removing one of her chores, and even if she receives the good,
her value for her bundle is at most $-6 + 3\eps < -5 + \eps$.
Hence, she is not PROPm-satisfied.

\textbf{Case 2}: agent 1 receives 2 chores.
\\ Then agent 2 also receives 2 chores, and agent 3 receives 1 chore.
Without loss of generality, assume agent 1 does not receive the good.
Then this allocation is EFX and groupwise MMS.
On pricing each chore at $-3$ and pricing the good at 3,
we get that the allocation is groupwise APS.
However, if agent 1 adds the good to her bundle, her value becomes $-6 + 3\eps < -5 + \eps$,
so she is not PROPm-satisfied.
\end{proof}

\begin{example}[PROPm $\nfimplies$ PROPavg]
\label[example]{cex:propm-not-propavg}
Consider an instance with three equally-entitled agents
having an identical additive valuation function $v$ over three goods.
Let $v(1) = v(2) = 60$ and $v(3) = 30$.
Then allocation $A \defeq (\{3\}, \{1, 2\}, \emptyset)$ is PROPm,
but agent 3 is not PROPavg-satisfied by $A$.
\end{example}

\subsection{Unequal Entitlements}
\label{sec:cex-extra:uneqEnt}

\begin{example}[PROP1 $\nfimplies$ M1S]
\label[example]{cex:prop1-not-m1s-n2}
\label[example]{cex:prop1-not-m1s-n2:goods}
\label[example]{cex:prop1-not-m1s-n2:chores}
Consider an instance with 2 agents having identical additive valuations.
Let $t \in \{-1, 1\}$. Let there be 2 items, each of value $t$.
Let the entitlement vector be $(2/3, 1/3)$.
The allocation where the first agent gets both items is PROP1 but not M1S.
\end{example}

\begin{lemma}[GAPS $\nfimplies$ PROPx for $n=2$]
\label[lemma]{cex:gaps-not-propx-n2}
\label[lemma]{cex:gaps-not-propx-n2:goods}
\label[lemma]{cex:gaps-not-propx-n2:chores}
Let $t \in \{-1, 1\}$. Consider a fair division instance $\fdInst{[2]}{[12]}{(v, v)}{(0.4, 0.6)}$,
where $v(1) = v(2) = 10t$ and $v(j) = t$ for all $j \in [12] \setminus [2]$.
Then allocation $A = (\{1\}, [12] \setminus \{1\})$ is APS but not PROPx.
\end{lemma}
\begin{proof}
Agents 1 and 2 have proportional shares $12t$ and $18t$, respectively.

For $t = 1$, agent 1's APS is at most $10$
(price items 1 and 2 at $0.45$ each and the remaining items at $0.01$ each),
and agent 2's APS is at most $18$ (price each item at its value).
Hence, $A$ is APS but not PROPx-fair to agent 1.

For $t = -1$, agent 2's APS is at most $-20$
(price items 1 and 2 at $-0.45$ each and the remaining items at $-0.01$ each),
and agent 1's APS is at most $-12$ (price each item at its value).
Hence, $A$ is APS but not PROPx-fair to agent 2.
\end{proof}

\begin{example}[EF1 $\nfimplies$ EFX for binary mixed manna]
\label[example]{cex:ef1-not-efx-mixed-ue}
Let $([2], [5], (v, v), (4/5, 1/5))$ be a fair division instance
where $v(g) = 1$ for all $g \in [4]$ and $v(5) = -1$.
Let $A = ([5] \setminus \{1\}, \{1\})$. Then $A$ is EF1 but not EFX.
\end{example}

\begin{example}[WMMS and M1S are incompatible for chores]
\label[example]{cex:wmms-plus-m1s-chores}
\label[example]{cex:wmms-plus-m1s-chores:mms-not-m1s}
\label[example]{cex:wmms-plus-m1s-chores:efx-not-mms}
Let $([2], [2], (v, v), (0.9, 0.1))$ be a fair division instance, where $v(1) = v(2) = -1$.
Then for all $i \in [2]$, we have
\[ \frac{-\WMMS_i}{w_i} = \min_X \max\left(\frac{10|X_1|}{9}, 10|X_2|\right), \]
so $-\WMMS_1 = 2$ and $-\WMMS_2 = 2/9$.
Hence, allocation $X$ is WMMS iff $|X_2| = 0$.
On the other hand, $X$ is M1S iff $|X_1| = |X_2| = 1$ iff $X$ is EFX.
\end{example}

\begin{lemma}[PROP1+M1S is infeasible for goods]
\label[lemma]{cex:prop1-plus-m1s-ue}
Consider a fair division instance $\Ical \defeq \fdInst{[3]}{[7]}{(v_i)_{i=1}^3}{w}$,
where the entitlement vector is $w \defeq (7/12, 5/24, 5/24)$,
the agents have identical additive valuations, and each good has value 1. Then
\begin{tightenum}
\item $\APS_1 = 4$ and $\APS_2 = \APS_3 = 1$.
\item $X$ is APS $\iff$ $X$ is groupwise-APS (GAPS) $\iff$ $X$ is PROP1.
\item $\WMMS_1 = 3$ and $\WMMS_2 = \WMMS_3 = 15/14$.
\item $X$ is WMMS $\iff$ $X$ is groupwise-WMMS (GWMMS) $\iff$ $X$ is EFX $\iff$ $X$ is M1S.
\end{tightenum}
Therefore,
\begin{tightenum}
\item M1S+PROP1 is infeasible for this instance.
\item GWMMS+EFX doesn't imply PROP1.
\item GAPS doesn't imply M1S.
\end{tightenum}
\end{lemma}
\begin{proof}
By \cref{thm:impl:tribool:aps}, $\APS_1 = \floor{\frac{7 \times 7}{12}} = 4$
and $\APS_2 = \APS_3 = \floor{\frac{5 \times 7}{24}} = 1$.
By \cref{thm:impl:tribool:aps,thm:impl:tribool:prop1},
an allocation is APS iff it is PROP1.

Any GAPS allocation is also APS by definition.
We will now show that any APS allocation is also GAPS.
Formally, let $A$ be an APS allocation for $\Ical$.
The cardinality vector of $A$, i.e., $c \defeq (|A_1|, |A_2|, |A_3|)$,
can have three possible values: $(5, 1, 1)$, $(4, 2, 1)$, $(4, 1, 2)$.
For every possible value of $c$ and $S \subseteq [3]$,
we show that $(\Icalhat, \Ahat) \defeq \restrict(\Ical, A, S)$ is APS
(c.f.~\cref{defn:restricting}).
\begin{tightenum}
\item $c = (5, 1, 1)$ and $S = \{1, 2\}$:
    $\Icalhat$ has 6 goods and entitlement vector $(14/19, 5/19)$.
    $\APS_1 = \floor{\frac{14 \times 6}{19}} = 4$ and $\APS_2 = \floor{\frac{5 \times 6}{19}} = 1$.
    Hence, $\Ahat$ is APS for $\Icalhat$.
\item $c = (5, 1, 1)$, $S = \{1, 3\}$:
    Same as the $S = \{1, 2\}$ case.
\item $c = (5, 1, 1)$ and $S = \{2, 3\}$:
    $\Icalhat$ has 2 goods and entitlement vector $(1/2, 1/2)$.
    $\APS_2 = \APS_3 = 1$, so $\Ahat$ is APS for $\Icalhat$.
\item $c = (4, 2, 1)$ and $S = \{1, 2\}$:
    $\Icalhat$ has 6 goods and entitlement vector $(14/19, 5/19)$.
    $\APS_1 = \floor{\frac{14 \times 6}{19}} = 4$ and $\APS_2 = \floor{\frac{5 \times 6}{19}} = 1$.
    Hence, $\Ahat$ is APS for $\Icalhat$.
\item $c = (4, 2, 1)$ and $S = \{1, 3\}$:
    $\Icalhat$ has 5 goods and entitlement vector $(14/19, 5/19)$.
    $\APS_1 = \floor{\frac{14 \times 5}{19}} = 3$ and $\APS_2 = \floor{\frac{5 \times 5}{19}} = 1$.
    Hence, $\Ahat$ is APS for $\Icalhat$.
\item $c = (4, 2, 1)$ and $S = \{2, 3\}$:
    $\Icalhat$ has 3 goods and entitlement vector $(1/2, 1/2)$.
    $\APS_2 = \APS_3 = \floor{\frac{1 \times 3}{2}} = 1$.
    Hence, $\Ahat$ is APS for $\Icalhat$.
\item $c = (4, 1, 2)$:
    Similar to the $c = (4, 2, 1)$ case.
\end{tightenum}
Hence, any APS allocation for $\Ical$ is also GAPS.

For any allocation $X$, define
\[ f(X) \defeq \min_{j=1}^3 \frac{|X_j|}{w_j}. \]
Then $\WMMS_i = w_i\max_X f(X)$ for all $i \in [3]$.
If $|X_1| \le 2$, then $f(X) \le |X_1|/w_1 \le 24/7 = 3 + 3/7$.
If $|X_2| \le 1$ or $|X_3| \le 1$, $f(X) \le 24/5 = 4 + 4/5$.
Otherwise, $|X_1| = 3$ and $|X_2| = |X_3| = 2$,
so $f(X) = \min(3 \times 12/7, 2 \times 24/5) = 36/7 = 5 + 1/6$.
Hence, $\max_X f(X) = 36/7$, so $\WMMS_1 = 3$ and $\WMMS_2 = \WMMS_3 = 15/14$.
So, an allocation is WMMS iff it has cardinality vector $(3, 2, 2)$.

Any GWMMS allocation is also WMMS by definition. We now prove the converse.
For every possible value of $S \subseteq [3]$,
we show that $(\Icalhat, \Ahat) \defeq \restrict(\Ical, A, S)$ is WMMS
(c.f.~\cref{defn:restricting}).
\begin{tightenum}
\item $S = \{1, 2\}$:
    $\Icalhat$ has 5 goods and entitlement vector $(14/19, 5/19)$.
    If $|X_1| \le 2$, then $f(X) \le |X_1|/w_1 \le 38/14 = 2 + 10/14$.
    If $|X_2| \le 1$, then $f(X) \le |X_2|/w_2 \le 19/5 = 3 + 10/14$.
    Otherwise, $|X_1| = 3$ and $|X_2| = 2$, so $f(X) = \min(3/w_1, 2/w_2) = \min(57/14, 38/5) = 57/14 = 4 + 1/14$.
    Hence, $\WMMS_1 = 3$ and $\WMMS_2 = \WMMS_3 = 57/14 \times 5/19 = 1 + 19/266$.
    Hence, $\Ahat$ is WMMS for $\Icalhat$.
\item $S = \{1, 3\}$:
    Similar to the $S = \{1, 2\}$ case.
\item $S = \{2, 3\}$:
    $\Icalhat$ has 4 goods and entitlement vector $(1/2, 1/2)$.
    Then $\WMMS_1 = \WMMS_2 = 2$.
    Hence, $\Ahat$ is WMMS for $\Icalhat$.
\end{tightenum}
Hence, any WMMS allocation for $\Ical$ is also GWMMS.

Any GWMMS allocation is EFX by \cref{thm:impl:mms-to-efx-n2},
and any EFX allocation is M1S by \cref{thm:impl:efx-to-ef1}.
We will now show that any M1S allocation is WMMS.

Let $X$ be an M1S allocation.
Let $A$ be agent 1's M1S certificate.
If $|A_1| \le 2$, then $|A_j| \ge 3$ for some $j \in \{2, 3\}$.
Since agent 1 has higher entitlement, she would EF1-envy agent $j$,
which is a contradiction. Hence, $|X_1| \ge |A_1| \ge 3$.

Let $B$ be agent 2's M1S certificate.
Suppose $|B_2| \le 1$. Since $2$ doesn't EF1-envy $3$, we get $|B_3| \le 2$.
Since $2$ doesn't EF1-envy $1$, we get
\[ \frac{|B_1|-1}{w_1} \le \frac{|B_2|}{w_2} \iff |B_1| \le 1 + \frac{w_1}{w_2} = 3 + \frac{4}{5}. \]
Hence, $|B_1| + |B_2| + |B_3| \le 3 + 1 + 2 = 6$, which is a contradiction.
Hence, $|X_2| \ge |B_2| \ge 2$.
Similarly, we can prove that $|X_3| \ge 2$.

Hence, $|X_i| \ge \WMMS_i$ for all $i$, so $X$ is WMMS.
This proves that any M1S allocation is WMMS.
\end{proof}

\citet{chakraborty2021weighted} also prove that EF1+PROP1 allocations may not exist for unequal entitlements.
We use a different counterexample in \cref{cex:prop1-plus-m1s-ue},
which allows us to also prove other non-implications.

\begin{lemma}
\label[lemma]{cex:gwmms-nimpl-prop1-m1s-ue-chores}
Consider a fair division instance $\Ical \defeq \fdInst{[3]}{[7]}{(v_i)_{i=1}^3}{w}$,
where the entitlement vector is $w \defeq (9/16,\allowbreak 7/32,\allowbreak 7/32)$,
the agents have identical additive valuations, and each chore has disutility 1.
Let $X$ be an allocation. Then
\begin{tightenum}
\item $\APS_1 = -4$ and $\APS_2 = \APS_3 = -2$.
\item $X$ is APS $\iff$ $X$ is PROP1.
\item $\WMMS_1 = -5$ and $\WMMS_2 = \WMMS_3 = -35/18$.
\item $X$ is WMMS $\iff$ $X$ is groupwise-WMMS $\iff$ $(|X_1| = 5$ and $|X_2| = |X_3| = 1)$.
\item $X$ is EFX $\iff$ $X$ is M1S $\iff$ $(|X_1| = 3$ and $|X_2| = |X_3| = 2)$.
\end{tightenum}
Therefore,
\begin{tightenum}
\item PROP1+WMMS is infeasible for this instance.
\item M1S+WMMS is infeasible for this instance.
\item GWMMS doesn't imply PROP1 or M1S.
\end{tightenum}
\end{lemma}
\begin{proof}
Note that $\floor{-x} = -\ceil{x}$ for all $x \in \mathbb{R}$.

By \cref{thm:impl:tribool:aps}, $-\APS_1 = \ceil{\frac{9 \times 7}{16}} = 4$,
and $-\APS_2 = -\APS_3 = \ceil{\frac{7 \times 7}{32}} = 2$.
By \cref{thm:impl:tribool:aps,thm:impl:tribool:prop1},
an allocation is APS iff it is PROP1.

For any allocation $X$, define
\[ f(X) \defeq \max_{j=1}^3 \frac{|X_j|}{w_j}. \]
Then $-\WMMS_i = w_i\min_X f(X)$ for all $i \in [3]$.
If $|X_1| \ge 6$, then $f(X) \ge 6/w_1 = 10 + 2/3$.
If $|X_2| \ge 2$ or $|X_3| \ge 2$, then $f(X) \ge 2/w_2 = 9 + 1/7$.
Otherwise, $|X_1| = 5$ and $|X_2| = |X_3| = 1$,
so $f(X) = \max(5 \times 16/9, 1 \times 32/7) = 80/9 = 9 - 1/9$.
Hence, $\min_X f(X) = 80/9$, so $-\WMMS_1 = 5$ and $-\WMMS_2 = -\WMMS_3 = 35/18 < 2$.
Hence, an allocation $X$ is WMMS iff $|X_1| = 5$ and $|X_2| = |X_3| = 1$.

Any GWMMS allocation is also WMMS by definition. We now prove the converse.
Let $A$ be a WMMS allocation. For every possible value of $S \subseteq [3]$,
we show that $(\Icalhat, \Ahat) \defeq \restrict(\Ical, A, S)$ is WMMS
(c.f.~\cref{defn:restricting}).
\begin{tightenum}
\item $S = \{1, 2\}$:
    $\Icalhat$ has 6 chores and entitlement vector $\what \defeq (18/25, 7/25)$.
    Let $X$ be an allocation for $\Icalhat$.
    Let $\fhat(X) \defeq \max_{j=1}^2 |X_j|/\what_j$.
    If $|X_1| \ge 6$, then $\fhat(X) \ge 6/\what_1 = 25/3 = 8 + 1/3$.
    If $|X_2| \ge 2$, then $\fhat(X) \ge 2/\what_2 = 50/7 = 7 + 1/7$.
    Otherwise, $|X_1| = 5$ and $|X_2| = 1$, so
    $\fhat(X) = \max(5/\what_1, 1/\what_2) = 125/18 = 7 - 1/18$.
    Hence, $-\WMMS_1 = 5$ and $-\WMMS_2 = 35/18 < 2$.
    Hence, $\Ahat$ is WMMS for $\Icalhat$.
\item $S = \{1, 3\}$:
    Similar to the $S = \{1, 2\}$ case.
\item $S = \{2, 3\}$:
    $\Icalhat$ has 2 chores and entitlement vector $(1/2, 1/2)$.
    Then $-\WMMS_1 = -\WMMS_2 = 1$.
    Hence, $\Ahat$ is WMMS for $\Icalhat$.
\end{tightenum}
Hence, $A$ is GWMMS.

Suppose $X$ is M1S. Let $A$ be agent 1's M1S certificate.
Then $|X_1| \le |A_1|$ and $A$ is EF1-fair to agent 1.
Suppose $|A_1| \ge 4$. Then $|A_2| \le 1$ or $|A_3| \le 1$. \WLoG{}, let $|A_2| \le 1$. Then
\[ \frac{|A_1|-1}{w_1} \ge \frac{3}{w_1} = 5 + \frac{1}{3} > 5 - \frac{3}{7} = \frac{1}{w_2} \ge \frac{|A_2|}{w_2}. \]
Hence, agent 1 EF1-envies agent 2, which is a contradiction.
Hence, $|X_1| \le |A_1| \le 3$.

Let $B$ be agent 3's M1S certificate.
Then $|X_3| \le |B_3|$ and $B$ is EF1-fair to agent 3.
Suppose $|B_3| \ge 3$. If $|B_2| \le 1$, then agent 3 EF1-envies agent 2, so $|B_2| \ge 2$.
Hence, $|B_1| \le 2$. Then
\[ \frac{|B_3|-1}{w_3} \ge \frac{2}{w_3} = 9 + \frac{1}{7} > 4 - \frac{4}{9} = \frac{2}{w_1} \ge \frac{|B_1|}{w_1}. \]
Hence, agent 3 EF1-envies agent 2, which is a contradiction.
Hence, $|X_3| \le |B_3| \le 2$.
We can similarly prove that $|X_2| \le 2$.
Hence, if $X$ is M1S, then $|X_1| = 3$ and $|X_2| = |X_3| = 2$. Moreover,
\ifColsTwo\else\providecommand{\tempWrap}{&}\fi
\begin{align*}
\frac{|X_1|-1}{w_1} &= 4 - \frac{4}{9} < 9 + \frac{1}{7} = \frac{|X_2|}{w_2} = \frac{|X_3|}{w_3},
\tempWrap \frac{|X_3|-1}{w_3} &= \frac{|X_2|-1}{w_2} = 5 - \frac{3}{7} < 5 + \frac{1}{3} = \frac{|X_1|}{w_1}.
\end{align*}
Hence, $X$ is EFX. Also, for additive valuations, EFX implies M1S by \cref{thm:impl:efx-to-ef1}.
Hence, $X$ is EFX iff $X$ is M1S.
\end{proof}

\begin{lemma}[PROP $\nfimplies$ M1S for chores]
\label[lemma]{cex:prop-not-m1s-chores}
\label[lemma]{cex:prop-not-m1s-chores:neg-binary}
\label[lemma]{cex:prop-not-m1s-chores:negative-bival}
Let $0 \le \eps \le 2/7$.
Let $\fdInst{[3]}{[7]}{(v_i)_{i=1}^3}{w}$ be a fair division instance,
where $w = (4/7, 3/14, 3/14)$ and $v$ is given by the following table:

\begin{tabular}{cccccccc}
\toprule $c$ & 1 & 2 & 3 & 4 & 5 & 6 & 7
\\\midrule $-v_1(c)$ & 1 & 1 & 1 & 1 & 1 & 1 & 1
\\ $-v_2(c) = -v_3(c)$ & 1 & 1 & 1 & 1 & $\eps$ & $\eps$ & $\eps$
\\\bottomrule
\end{tabular}

Let $A$ be an allocation where $A_1 = [4]$. Then $A$ is PROP but not M1S-fair to agent 1.
\end{lemma}
\begin{proof}
$-v_2(A_2) \le 3\eps \le 4w_2$, so $A$ is PROP-fair to agents 2 and 3.
$-v_1(A_1) = 4 = 7w_1$, so $A$ is PROP-fair to agent 1.

Suppose $A$ is M1S-fair to agent 1.
Let $X$ be agent 1's M1S certificate for $A$.
Then $|X_1| \ge 4$, so $|X_2| \le 1$ or $|X_3| \le 1$. Hence,
\begin{align*}
\frac{|X_1|-1}{w_1} &\ge \frac{3}{w_1} = 5 + \frac{1}{4} > 5 - \frac{1}{3} = \frac{1}{w_2}
    \wrapIfTwoCols\ge \min\left(\frac{|X_2|}{w_2}, \frac{|X_3|}{w_3}\right).
\end{align*}
Hence, $X$ is not EF1-fair to agent 1, which is a contradiction.
Hence, $A$ is not M1S-fair to agent 1.
\end{proof}

\begin{example}[PROP $\nfimplies$ MEFS for binary valuations]
\label[example]{cex:prop-not-mefs:binary-goods}
Let $\fdInst{[3]}{[6]}{(v_i)_{i=1}^3}{w}$ be a fair division instance,
where $w = (1/2, 1/4, 1/4)$ and $v$ is given by the following table:

\begin{tabular}{cccccccc}
\toprule $g$ & 1 & 2 & 3 & 4 & 5 & 6
\\\midrule
   $v_1(g)$ & 1 & 1 & 1 & 1 & 1 & 1
\\ $v_2(g)$ & 0 & 0 & 0 & 1 & 1 & 0
\\ $v_3(g)$ & 0 & 0 & 0 & 0 & 0 & 1
\\\bottomrule
\end{tabular}

Let $A \defeq (\{1, 2, 3\}, \{4, 5\}, \{6\})$. Then $A$ is PROP.
$A$ is not MEFS for agent 1, because in any allocation $B$ such that $v_1(B_1) \le v_1(A_1) = 3$,
some agent $j \in \{2, 3\}$ must receive at least 2 goods, and then agent 1 would envy $j$ in $B$.
\end{example}

\section{Non-Implications (Non-Additive Valuations)}
\label{sec:cex-nonadd-extra}

\subsection{Supermodular Valuations}
\label{sec:cex-extra:supmod}

\begin{lemma}[EF+GAPS $\nfimplies$ PROP1]
\label[lemma]{cex:ef-not-prop-supmod}
\label[lemma]{cex:ef-not-prop-supmod:n2-positive-bival}
\label[lemma]{cex:ef-not-prop-supmod:n2-binary}
\label[lemma]{cex:ef-not-prop-supmod:n3-positive-bival}
\label[lemma]{cex:ef-not-prop-supmod:n3-binary}
Let $0 \le \eps < 1/27$ and $k \in \mathbb{N}$. For any $x \ge 0$, let
$f(x) \defeq x\eps + \max(0, x-k)$.
Consider an instance with $n$ equally-entitled agents
having an identical valuation function $v$ over $4n$ goods,
where $v(X) = f(|X|)$ for all $X \subseteq [4n]$. Then $v$ is supermodular.

Let $A$ be the allocation where each agent gets 4 goods. Then $A$ is EF+GAPS+GMMS.
For $n = 2$ and $k = 5$, no PROP1 or PROPm allocation exists.
If $n \ge 3$ and $k = 8$, then $A$ is PPROP but no PROP1 or PROPx allocation exists.
Additionally if $\eps = 0$, then $A$ is PROPavg.
\end{lemma}
\begin{proof}
For any $X \subseteq [4n]$ and $g \in [4n] \setminus X$, we have
$v(g \mid X) = \eps$ if $|X| < k$ and $v(g \mid X) = 1+\eps$ otherwise.
Since $v(g \mid X)$ is non-decreasing in $|X|$, $v$ is supermodular.

Since all agents have the same bundle, $A$ is EF and GMMS.
Set the price of each good to 1 to get that $A$ is GAPS.

For $n = 2$ and $k = 5$, the proportional share is $4-k/n+4\eps$.
For $k = 5$, the proportional share is at least $3/2$,
but $f(5) = 5\eps$ and $f(6) = 1+6\eps < 3/2$.
In any allocation, some agent will have at most 4 goods,
and that agent would not be PROP1 or PROPm satisfied.

If $n \ge 3$ and $k = 8$, the proportional share is at least $4/3$.
Since $f(8) = 8\eps$ and $f(4) = 4\eps$, $A$ is PPROP.
In any allocation, some agent gets at most 4 goods,
and $f(5) = 5\eps < 4/3$, so that agent is not PROP1-satisfied,
and $f(9) = 1 + 9\eps < 4/3$, so that agent is not PROPx-satisfied.
Additionally, if $\eps = 0$, for any two agents $i$ and $j$,
we have $v(A_j \mid A_i) = f(8) - f(4) = 0$. Hence, $A$ is PROPavg.
\end{proof}

\subsection{Unit-Demand Valuations}
\label{sec:cex-extra:unit-demand}

It is known that unit-demand functions are submodular and cancelable.
We first show how to perturb a unit-demand function such that
submodularity and cancelability remain preserved and all marginals become positive.

\begin{lemma}
\label[lemma]{thm:ud-perturb}
Let $u: 2^{[m]} \to \mathbb{R}_{\ge 0}$ be a function and $\eps \in \mathbb{R}_{\ge 0}$.
For any $X \subseteq [m]$, define $\uhat(X) \defeq u(X) + |X|\eps$.
If $u$ is submodular, then $\uhat$ is also submodular.
If $m \le 3$ and $u$ is cancelable, then $\uhat$ is also cancelable.
\end{lemma}
\begin{proof}
If $u$ is submodular, then for any $S, T \subseteq [m]$, we get
$\uhat(S) + \uhat(T) - \uhat(S \cup T) - \uhat(S \cap T)
= u(S) + u(T) - u(S \cup T) - u(S \cap T) \ge 0$.
Hence, $\uhat$ is also submodular.

Let $u$ be cancelable.
For every $T \subseteq [m]$ and $S_1, S_2 \subseteq [m] \setminus T$,
if $\uhat(S_1 \cup T) > \uhat(S_2 \cup T)$,
we need to show that $\uhat(S_1) > \uhat(S_2)$
to prove that $\uhat$ is cancelable.
This is trivial when $T = \emptyset$ or $\eps = 0$.
If $|S_1| = |S_2|$, then $u(S_1 \cup T) > u(S_2 \cup T)$.
Since $u$ is cancelable, we get $u(S_1) > u(S_2)$, which implies $\uhat(S_1) > \uhat(S_2)$.
Let $\eps > 0$. If $S_1 = \emptyset$, we get
$\uhat(S_1 \cup T) > \uhat(S_2 \cup T) \implies u(S_2 \mid T) < - |S_2|\eps$,
which is a contradition since $u$ is non-negative. Hence, $S_1 \neq \emptyset$.
If $S_2 = \emptyset$, we get $\uhat(S_1) \ge |S_1|\eps > 0 = \uhat(S_2)$.
When $m \le 3$, we have considered all cases.
Hence, if $m \le 3$ and $u$ is cancelable, then $\uhat$ is also cancelable.
\end{proof}

\begin{lemma}[PROP $\nfimplies$ M1S]
\label[lemma]{cex:prop-not-m1s-unit-demand}
\label[lemma]{cex:prop-not-m1s-unit-demand:nonneg}
\label[lemma]{cex:prop-not-m1s-unit-demand:positive}
Consider a unit-demand function $v$ over 3 goods such that $v(1) = v(2) = 4$ and $v(3) = 3$.
Let $0 \le \eps < 1/6$ and $\vhat(X) \defeq v(X) + 2|X|\eps$ for all $X \subseteq [3]$.
Consider a fair division instance with 3 goods and
2 equally-entitled agents having valuation function $\vhat$.
Then the allocation $A = (\{1, 2\}, \{3\})$ is PROP but not M1S.
\end{lemma}
\begin{proof}
$A$ is PROP since $\vhat(A_1) = 4+4\eps$, $\vhat(A_2) = 3+2\eps$, and the PROP share is $2+3\eps$.
Suppose $A$ is M1S and agent 2's M1S certificate for $A$ is $B$.
$B_2 \neq \emptyset$, since $B$ is EF1-fair to agent 2.
Moreover, $\vhat(B_2) \le \vhat(A_2) = 3 + 2\eps$, so $B_2 = \{3\}$. Hence, $B_1 = \{1, 2\}$.
However, $\min_{g \in B_1} \vhat(B_1 \setminus \{g\}) = 4 + 4\eps > 3 + 2\eps = \vhat(B_2)$.
Hence, agent 2 is not EF1-satisfied by $B$, which contradicts the fact that $B$
is agent 2's M1S certificate for $A$. Hence, $A$ is not M1S.
\end{proof}

\begin{example}[PROP1 $\nfimplies$ PROPm]
\label[example]{cex:ud:prop1-not-propm}
\label[example]{cex:ud:prop1-not-propm:nonneg}
\label[example]{cex:ud:prop1-not-propm:positive}
Consider a unit-demand function $v$ over 2 goods such that $v(1) = 30$ and $v(2) = 10$.
Let $0 \le \eps < 1$ and $\vhat(X) \defeq v(X) + |X|\eps$ for all $X \subseteq [2]$.
Consider a fair division instance with 2 goods and
2 equally-entitled agents having valuation function $\vhat$.
Then the allocation $A \defeq (\{1, 2\}, \emptyset)$ is PROP1 but not PROPm.
\end{example}

\begin{example}[PROP $\nfimplies$ PPROP]
\label[example]{cex:ud:prop-not-pprop}
\label[example]{cex:ud:prop-not-pprop:positive}
\label[example]{cex:ud:prop-not-pprop:nonneg}
Consider a unit-demand function $v$ over 3 goods such that $v(1) = 30$ and $v(2) = v(3) = 11$.
Let $0 \le \eps < 1$ and $\vhat(X) \defeq v(X) + |X|\eps$ for all $X \subseteq [3]$.
Consider a fair division instance with 3 goods and
3 equally-entitled agents having valuation function $\vhat$.
Then the allocation $A \defeq (\{1\}, \{2\}, \{3\})$ is PROP but no allocation is PPROP.
\end{example}

\begin{example}[MMS $\nfimplies$ EF1, PROPm]
\label[example]{cex:ud:mms-not-ef1-propm}
\label[example]{cex:ud:mms-not-ef1-propm:nonneg}
\label[example]{cex:ud:mms-not-ef1-propm:positive}
Consider a unit-demand function $v$ over 5 goods such that
$v(1) = 200$ and $v(2) = 30$ and $v(3) = v(4) = v(5) = 10$.
Let $0 \le \eps < 1$ and $\vhat(X) \defeq v(X) + |X|\eps$ for all $X \subseteq [5]$.
Consider a fair division instance with 5 goods and
4 equally-entitled agents having valuation function $\vhat$.
Then the allocation $A \defeq (\{1, 2\}, \{3\}, \{4\}, \{5\})$
is MMS but not EF1 and not PROPm.
(The maximin share is $10 + \eps$. The proportional share is $50 + (5/4)\eps$.)
\end{example}

\begin{example}[PROPm $\nfimplies$ PROPavg]
\label[example]{cex:ud:propm-not-propavg}
\label[example]{cex:ud:propm-not-propavg:positive}
\label[example]{cex:ud:propm-not-propavg:nonneg}
Consider a unit-demand function $v$ over 3 goods such that
$v(1) = 75$, $v(2) = 30$, and $v(3) = 10$.
Let $0 \le \eps < 1$ and $\vhat(X) \defeq v(X) + |X|\eps$ for all $X \subseteq [3]$.
Consider a fair division instance with 3 goods and
3 equally-entitled agents having valuation function $\vhat$.
Then the allocation $A \defeq (\{1, 3\}, \{2\}, \emptyset)$ is PROPm but not PROPavg.
(The proportional share is $25+\eps$.)
\end{example}

\subsection{Submodular Valuations}
\label{sec:cex-extra:submod}

\begin{definition}[PMRF]
\label[definition]{defn:pmrf}
Let $P = (P_1, \ldots, P_k)$ be a partition of a finite set $M$.
Define $\PMRF(P)$ as the function $f: 2^M \to [|M|]$ where
$f(X) \defeq \sum_{i=1}^k \boolone(X \cap P_i \neq \emptyset)$.
(Here $\boolone(C)$ is 1 if condition $C$ is true and 0 otherwise.)
For any $\eps \in \mathbb{R}_{\ge 0}$, define $\PMRF_{+\eps}(P)$
as the function $f$ where $f(X) = \PMRF(P)(X) + |X|\eps$.
\end{definition}

\begin{lemma}
\label[lemma]{thm:pmrf-submod}
For any partition $P$ over a finite set $M$,
$f \defeq \PMRF_{+\eps}(P)$ is submodular, and $f(g \mid X) \in \{\eps, 1+\eps\}$
for any $X \subseteq M$ and $g \in M \setminus X$.
\end{lemma}
\begin{proof}[Proof sketch]
$\PMRF(P)$ is the rank function of a partition matroid, and hence, is submodular.
The submodularity of $f$ follows from \cref{thm:ud-perturb}.
\end{proof}

\begin{example}[EF $\nfimplies$ MMS]
\label[example]{cex:ef-not-mms-pmrf}
\label[example]{cex:ef-not-mms-pmrf:binary}
\label[example]{cex:ef-not-mms-pmrf:positive-bival}
Let $P \defeq (\{r_1, r_2\},\allowbreak \{g_1, g_2\})$, $\eps \in \mathbb{R}_{\ge 0}$,
and $v \defeq \PMRF_{+\eps}(P)$.
Consider a fair division instance with 2 equally-entitled agents
having identical valuation function $v$.
Then the maximin share is at least $2 + 2\eps$,
as exemplified by the partition $(\{r_1, g_1\}, \{r_2, g_2\})$.
The allocation $P$ is EF but not MMS,
since each agent's value for her own bundle is $1 + 2\eps$.
\end{example}

\begin{example}[EF1 $\nfimplies$ MXS]
\label[example]{cex:ef1-not-mxs-pmrf}
Let $P \defeq (\{r_1, r_2\},\allowbreak \{g_1, g_2\},\allowbreak \{b\})$ and $v \defeq \PMRF(P)$.
Consider a fair division instance with 2 equally-entitled agents
having identical valuation function $v$.
Then the allocation $A \defeq (\{r_1, r_2\},\allowbreak \{g_1, g_2, b\})$ is EF1
but agent 1 is not MXS-satisfied by $A$.
\end{example}

\begin{example}[MEFS $\nfimplies$ EF1]
\label[example]{cex:mefs-not-ef1-pmrf}
\label[example]{cex:mefs-not-ef1-pmrf:binary}
\label[example]{cex:mefs-not-ef1-pmrf:positive-bival}
Let $P \defeq (\{a\},\allowbreak \{b\},\allowbreak \{c\},\allowbreak \{d\},\allowbreak
\{e_1, e_2\},\allowbreak \{f_1, f_2\},\allowbreak \{g_1, g_2\})$.
Let $v \defeq \PMRF_{+\eps}(P)$ for $0 \le \eps \le 1/2$.
Consider a fair division instance with 2 equally-entitled agents
having identical valuation function $v$.
Let $A \defeq (\{a, e_1, f_1, g_1\},\allowbreak \{b, c, d, e_2, f_2, g_2\})$.
Since $v(A_2) = 6 + 6\eps$ and $v(A_1) = 4 + 4\eps$,
agent 2 is envy-free in $A$ but agent 1 is not EF1-satisfied by $A$.
However, agent 1 is MEFS-satisfied by $A$, and her MEFS-certificate for $A$ is
$B \defeq (\{a, b, c, d\},\allowbreak \{e_1, e_2, f_1, f_2, g_1, g_2\})$,
since $v(B_1) = v(A_1) = 4 + 4\eps$ and $v(B_2) = 3 + 6\eps$.
\end{example}

\begin{example}
\label[example]{cex:eef-not-ef1-pprop-pmrf}
\label[example]{cex:eef-not-ef1-pprop-pmrf:binary}
\label[example]{cex:eef-not-ef1-pprop-pmrf:positive-bival}
Let $M \defeq \{a_j\}_{j=1}^6 \cup \{b_j\}_{j=1}^6 \cup \{c_j\}_{j=1}^4$.
Let $P$ be a partition of $M$ where $X \in P$ iff $X = \{a_j, b_j\}$ for some $j \in [6]$
or $X = \{c_j\}$ for some $j \in [4]$.
Let $v \defeq \PMRF_{+\eps}(P)$ for $\eps \in [0, 1/6)$.
Consider an instance with 3 equally-entitled agents
having identical valuation function $v$.

Let $A \defeq (\{a_j\}_{j=1}^6,\allowbreak \{b_j\}_{j=1}^6,\allowbreak \{c_j\}_{j=1}^4)$.
Then $v(A_1) = v(A_2) = 6(1+\eps)$ and $v(A_3) = 4(1+\eps)$.
Then $A$ is not EF1 and not PPROP.
But $A$ is EEF since agents 1 and 2 are envy-free,
and agent 3's EEF-certificate is $B$ where
$B_1 \defeq \{a_1,\allowbreak a_2,\allowbreak a_3,\allowbreak b_1,\allowbreak b_2,\allowbreak b_3\}$,
$B_2 \defeq \{a_4,\allowbreak a_5,\allowbreak a_6,\allowbreak b_4,\allowbreak b_5,\allowbreak b_6\}$, and $B_3 = A_3$.
Then $v(B_3) = 4(1+\eps)$ and $v(B_1) = v(B_2) = 3 + 6\eps < 4$.
\end{example}

\begin{example}[PROP $\nfimplies$ M1S]
\label[example]{cex:prop-not-m1s-uniform-matroid}
\label[example]{cex:prop-not-m1s-uniform-matroid:binary}
\label[example]{cex:prop-not-m1s-uniform-matroid:positive-bival}
Let $v: 2^{[10]} \to \mathbb{R}_{\ge 0}$ be a function defined as $v(X) \defeq \min(6, |X|)$.
Then $v$ is submodular (since $v$ is the rank function of a uniform matroid)
and $v$ has binary marginals (i.e., $v(g \mid X) \in \{0, 1\}$
for all $X \subseteq [10]$ and $g \in [m] \setminus X$).
Let $0 \le \eps < 1$ and $\vhat(X) \defeq v(X) + |X|\eps$.

Consider a fair division instance with 2 equally-entitled agents
having identical valuation function $\vhat$.
Then $A \defeq ([4], [10] \setminus [4])$ is PROP but agent 1 is not M1S-satisfied.
\end{example}

\begin{lemma}[MMS $\nfimplies$ APS for $n=2$, Remark 2 of \citet{babaioff2023fair}]
\label[lemma]{cex:mms-not-aps-n2-submod}
\label[lemma]{cex:mms-not-aps-n2-submod:nonneg}
\label[lemma]{cex:mms-not-aps-n2-submod:positive}
Let $C \defeq \{\{1, 2, 3\},\allowbreak \{1, 5, 6\},\allowbreak \{2,\allowbreak 4,\allowbreak 6\},\allowbreak \{3, 4, 5\}\}$.
For $\eps \in [0, 1/3)$, define $v: 2^{[6]} \to \mathbb{R}_{\ge 0}$ as
\[ v(X) \defeq |X|\eps + \begin{cases}
2|X| & \text{ if } |X| \le 2
\\ 5 & \text{ if } |X| = 3 \text{ and } X \not\in C
\\ 6 & \text{ if } X \in C \text{ or } |X| \ge 4
\end{cases}. \]
Then $v$ is submodular.
Consider an instance with 2 equally-entitled agents
having identical valuation function $v$.
The maximin share is $5 + 3\eps$, but the APS is $6 + 3\eps$.
Moreover, there exists an MMS allocation but no allocation is APS.
\end{lemma}
\ifVerbose
\begin{proof}
For any $X \subseteq [6]$, and $g \in [6] \setminus X$, we have
\begin{tightenum}
\item $|X| \le 1 \implies v(g \mid X) = 2 + \eps$.
\item $|X| = 2 \implies v(g \mid X) \in \{1 + \eps, 2 + \eps\}$.
\item $|X| = 3 \implies v(g \mid X) \in \{\eps, 1 + \eps\}$.
\item $|X| = 4 \implies v(g \mid X) = \eps$.
\end{tightenum}
Hence, $v(g \mid X)$ decreases with $|X|$, so $v$ is submodular.

The MMS is less than 6 because for every set in $C$,
its complement is not in $C$.
The APS is $6 + 3\eps$ by the dual definition of APS
(\cref{defn:aps-dual}, set $x_S = 1/4$ for $S \in C$).

For identical valuations, an MMS allocation always exists.
For identical valuations, in any allocation, some agent gets a bundle of value at most the MMS.
Since the APS is strictly greater than the MMS, that agent is not APS-satisfied.
Hence, no allocation is APS.
\end{proof}
\fi

\subsection{Subadditive Valuations}
\label{sec:cex-extra:subadd}

\begin{lemma}
\label[lemma]{thm:floor-ceil-value}
For any $k \in \mathbb{N}$, define the set functions $f_1, f_2: 2^{[m]} \to \mathbb{R}$ as
$f_1(S) = \floor{|S|/k}$ and $f_2(S) = \ceil{|S|/k}$, respectively.
Then $f_1$ is superadditive, $f_2$ is subadditive,
and for any $S \subseteq [m]$, $g \in [m] \setminus S$, and $i \in [2]$,
we have $f_i(g \mid S) \in \{0, 1\}$.
\end{lemma}

\begin{example}
\label[example]{cex:gaps-not-ef1-prop1-subadd}
Let $\fdInst{[2]}{[9]}{(v_i)_{i=1}^2}{\eqEnt}$ be an instance where
$v_i(S) \defeq \ceil{|S|/4}$ for all $i \in [2]$ and $S \subseteq [9]$.
Then $v_i$ is subadditive for $i \in [2]$ by \cref{thm:floor-ceil-value}.
Let $A$ be an allocation where agent 1 has 3 goods and agent 2 has 6 goods.
Then $A$ is GAPS+EFX+PROPx but not EF1 or PROP1.

Each agent's APS is $\le 1$, because if we price each good at $1/9$,
then a bundle is affordable iff its value is at most 1.
Each agent's proportional share is $3/2$.
Note that $A$ is not \EFXZero{} (c.f.~\cref{defn:efx0-goods});
the difference between EFX and \EFXZero{} is crucial here.
\end{example}

\begin{example}
\label[example]{cex:ef1-not-prop1-subadd}
Let $\fdInst{[3]}{[7]}{(v)_{i=1}^3}{\eqEnt}$ be a fair division instance where
$v(S) \defeq \ceil{|S|/2}$ for all $S \subseteq [7]$.
Then $v$ is subadditive by \cref{thm:floor-ceil-value}.
Let $A \defeq (\{1\}, \{2, 3, 4\}, \{5, 6, 7\})$.
Then $A$ is EF1 but is not PROP1 to agent 1.
\end{example}

\begin{lemma}[optimal bin packing is subadditive]
\label[lemma]{thm:bin-packing}
Let $(s_j)_{j=1}^m$ be a sequence where $0 \le s_j \le 1$ for all $j \in [m]$
($s_j$ is called the \emph{size} of item $j$).
For any $X \subseteq [m]$, define $\optBP_s(X)$ as the smallest integer $k$
such that a partition $(P_1, \ldots, P_k)$ of $X$ exists
such that $\sum_{j \in P_i} s_j \le 1$ for all $i \in [k]$.
Then $\optBP_s$ is subadditive and has binary marginals,
i.e., $\optBP_s(X \cup \{g\}) - \optBP_s(X) \in \{0, 1\}$.
\end{lemma}
\begin{proof}
Pick any two disjoint sets $X, Y \subseteq [m]$.
Let $k_X \defeq \optBP_s(X)$ and $k_Y \defeq \optBP_s(Y)$.
Let $(P_1, \ldots, P_{k_X})$ and $(Q_1, \ldots, Q_{k_Y})$ be the corresponding partitions.
Then $(P_1, \ldots,\allowbreak P_{k_X},\allowbreak Q_1,\allowbreak \ldots, Q_{k_Y})$ is a feasible partition of $X \cup Y$,
so $\optBP_s(X \cup Y) \le k_X + k_Y$.
\end{proof}

\begin{example}
\label[example]{cex:gaps-not-prop1-subadd}
\label[example]{cex:gaps-not-prop1-subadd:binary}
\label[example]{cex:gaps-not-prop1-subadd:positive-bival}
Let $s = (1/5, 1/5, 3/5, 3/5, 3/5, 3/5)$ and $\eps \in [0, 1/3)$.
Consider an instance with 3 equally-entitled agents
having an identical valuation function $v$ over 6 goods.
$v(X) \defeq \optBP_s(X) + |X|\eps$
($\optBP_s$ is defined in \cref{thm:bin-packing}).
Then by \cref{thm:bin-packing}, $v$ is subadditive,
and for every $g \in [6]$, we have $v(g \mid \cdot) \in \{\eps, 1+\eps\}$.

Then allocation $A = (\{1, 2\},\allowbreak \{3, 4\},\allowbreak \{5, 6\})$
is EFX, pairwise PROP, and groupwise APS
(set the price of the first two goods to $1/2$ and the last four goods to 1).
However, agent 1 is not PROP1-satisfied in $A$, since $v([6]) = 4 + 6\eps$,
and $v(A_1 \cup \{g\}) = 1 + 3\eps$ for all $g \in [6] \setminus A_1$
\end{example}

\begin{lemma}[PAPS+PPROP $\nfimplies$ PROPm or M1S for binary subadd]
\label[lemma]{cex:paps-pprop-not-propm-binary-subadd}
Let $A_1$, $A_2$, and $A_3$ be pairwise-disjoint sets where
$A_1$ contains 7 goods of size $0.6$, $A_2$ contains 21 goods of size $0.4$,
and $A_3$ contains 21 goods of size $0.4$. Let $M \defeq A_1 \cup A_2 \cup A_3$.
Let $v(X) \defeq \optBP(X)$ for all $X \subseteq M$.
Consider a fair division instance with 3 equally-entitled agents
having the same valuation function $v$ over $M$.
Then allocation $A \defeq (A_1, A_2, A_3)$ is EFX+PPROP+PAPS,
but $A$ is not PROPm-fair or M1S-fair to agent 1.
\end{lemma}
\begin{proof}
First, we show that $v(S) = \max(|S \cap A_1|, \ceil{|S|/2})$ for any $S \subseteq M$.
To see this, suppose $S$ contains $a$ items of size $0.6$ and $b$ items of size $0.4$.
If $b \le a$, we require $a$ to pack the items and $\ceil{(a+b)/2} \le a$.
If $b \ge a$, then we require $\ceil{(a+b)/2}$ bins to pack the items and $a \le (a+b)/2$.
As a corollary, we get that for all $S \subseteq M \setminus A_1$,
we have $v(S \mid A_1) = \max(0, \ceil{(|S|-|A_1|)/2})$.

$v(A_2) = v(A_3) = 11$ and $v(A_1) = 7$, so agents 2 and 3 are envy-free.
For any $S \subseteq A_2$, $v(S \mid A_1) > 0 \iff |S| \ge 8$,
and $|S| \ge 8 \implies |A_2 \setminus S| \le 13 \implies v(A_2 \setminus S) \le 7 = v(A_1)$.
Hence, agent 1 doesn't EFX-envy agent 2. Similarly, she doesn't EFX-envy agent 3.
Hence, $A$ is EFX.

$v(A_2 \cup A_3) = 21$, and $v(A_1 \cup A_2) = v(A_1 \cup A_3) = 14$.
Hence, $A$ is pairwise PROP. We will now show that $A$ is pairwise APS.
Set the price of each item to 1.
It's easy to see that the allocation restricted to agents 2 and 3 is APS.
Now suppose we restrict the allocation to agents 1 and 2.
Then $S \subseteq A_1 \cup A_2$ is affordable iff $|S| \le |A_1 \cup A_2|/2 = 14$.
For $|S| = 14$, we get $v(S) = 7$, so the APS is at most 7.
Hence, $A$ is pairwise APS satisfied.

$v(M)/3 = 25/3 = 8 + 1/3$.
For $S \subseteq A_2$ such that $|S| = 8$, we get that $v(S \mid A_1) = 1$,
but $v(A_1 \cup S) = 8 < v(M)/3$. Hence, $A$ is not PROPm-fair to agent 1.
Also, $v(A_1 \cup \{g\}) = v(A_1) = 7$ for all $g \in M \setminus A_1$,
so $A$ is also not PROP1-fair to agent 1.

Suppose $A$ is M1S-fair to agent 1. Let $B$ be her M1S-certificate for $A$.
Then $v(B_1) \le v(A_1) = 7$ and $B$ is EF1-fair to agent 1.
$v(B_1) \le 7$ implies $|B_1| \le 14$. Since $|M| = 49$, we get that
$\max(|B_2|, |B_3|) \ge \ceil{|M - B_1|/2} \ge 18$.
Hence, some agent $i \in \{2, 3\}$ has a bundle of value at least $9$.
If one good is removed from their bundle, its value can drop by at most $8$,
so $B$ is not EF1-fair to agent 1, which is a contradiction.
Hence, $A$ is not M1S-fair to agent 1.
\end{proof}

\begin{lemma}[GMMS $\nfimplies$ APS]
\label[lemma]{cex:gmms-not-aps-binary-subadd}
Let there be 15 items of sizes
$65/96$, $31/96$, $31/96$, $31/96$, $23/96$, $23/96$, $23/96$, $17/96$,
$11/96$, $7/96$, $7/96$, $7/96$, $5/96$, $5/96$, $5/96$.
Consider a fair division instance with 3 equally-entitled agents having
the same valuation function $v \defeq \optBP_s$ over the 15 items (c.f.~\cref{thm:bin-packing}).
Then $v$ is subadditive and has binary marginals.
Also, the AnyPrice share is at least 2, but the maximin share is 1.
Hence, no APS allocation exists, but the leximin allocation is GMMS.
\end{lemma}
\begin{proof}
The total size of the items is $3 \times 97/96$.
By Lemma C.1 of \citet{babaioff2023fair}, no $3$-partition of the items exists
such that the total size of each partition is at least $97/96$.
Hence, in every 3-partition $M$, some bundle has total size at most 1, so that bundle has value at most 1.
So, the maximin share is 1.

By Lemma C.1 of \citet{babaioff2023fair}, there exist 6 sets $(S_j)_{j=1}^6$ of items
such that each set has total size $97/96$ (and hence value 2) and each item appears in exactly two sets.
Hence, the APS is at least 2 (c.f.~\cref{defn:aps-dual}).
\end{proof}

\section{Feasibility of Fairness Notions}
\label{sec:feas}
\label{sec:feas-extra}

We list results regarding the feasibility and infeasibility of fairness notions
in \cref{table:feas,table:infeas}, respectively.

\begin{table*}[htb]
\centering
\caption{Feasibility of fairness notions}
\label{table:feas}
\bigTableSize
\setcounter{tabSerial}{0}
\begin{tabular}{clcccccr}
\toprule & \scriptsize notion & \scriptsize valuation & \scriptsize marginals & \scriptsize identical & \scriptsize $n$ & \scriptsize entitlements &
\\ \midrule \tabSn & EF1 & -- & dbl-mono\textsuperscript{\ref{foot:dbl-mono-2}} & -- & -- & equal
    & Theorem 4 of \shortcite{bhaskar2021approximate}
\\[\defaultaddspace] \tabSn & EF1 & additive & goods & -- & -- & --
    & Theorem 3.3 of \shortcite{chakraborty2021weighted}
\\[\defaultaddspace] \tabSn & EF1 & additive & chores & -- & -- & --
    & Theorem 19 of \shortcite{springer2024almost}
\\[\defaultaddspace] \tabSn & EF1+PO & additive & goods & -- & -- & --
    & \citet{caragiannis2019unreasonable}
\\[\defaultaddspace] \tabSn & EF1+PO & additive & chores & -- & -- & --
    & \citet{mahara2025existence}
\\[\defaultaddspace] \tabSn & EF1+PO & additive & -- & -- & $n=2$ & --
    & Theorem 14 of \shortcite{garg2024ef1}
\\[\defaultaddspace] \tabSn & MMS & additive & -- & -- & $n=2$ & equal
    & Cut-and-choose
\\[\defaultaddspace] \tabSn & WMMS & -- & -- & yes & -- & --
    & Trivial
\\[\defaultaddspace] \tabSn & APS & additive & -- & -- & $n=2$ & --
    & Proposition~6 of \shortcite{babaioff2023fair}
\\[\defaultaddspace] \tabSn & PROPx & additive & chores & -- & -- & --
    & Theorem 4.1 of \shortcite{li2022almost}\textsuperscript{\ref{foot:propx-li}}
\\[\defaultaddspace] \tabSn & PROPm & additive & goods & -- & -- & equal
    & \shortcite{baklanov2021propm}
\\[\defaultaddspace] \tabSn & PROPavg & additive & goods & -- & -- & equal
    & \shortcite{kobayashi2025proportional}\textsuperscript{\ref{foot:propavg}}
\\[\defaultaddspace] \tabSn & PROP1 & additive & -- & -- & -- & --
    & \shortcite{aziz2020polynomial}
\\[\defaultaddspace] \tabSn & EFX & additive & $\ge 0$, $\le 0$ & yes & -- & --
    & Theorem 5 of \shortcite{springer2024almost}
\\[\defaultaddspace] \tabSn & EEFX & cancelable & $\ge 0$, $\le 0$ & -- & -- & equal
    & \shortcite{caragiannis2022existence}
\\[\defaultaddspace] \tabSn & EEFX & -- & $\ge 0$, $\le 0$ & -- & -- & equal
    & \shortcite{akrami2025epistemic}
\\[\defaultaddspace] \tabSn & EFX & additive & bival goods & -- & -- & equal
    & Theorem 4.1 of \shortcite{amanatidis2021maximum}
\\[\defaultaddspace] \tabSn & MMS & additive & bival goods & -- & -- & equal
    & \shortcite{feige2022maximin}
\\[\defaultaddspace] \tabSn & MMS & additive & bival chores & -- & -- & equal
    & \shortcite{feige2022maximin}
\\[\defaultaddspace] \tabSn & MMS & submodular & $\{0, 1\}$ & -- & -- & equal
    & Theorem 3.4 of \shortcite{barman2021existence}
\\[\defaultaddspace] \tabSn & EFX+PO & submodular & $\{0, 1\}$ & -- & -- & equal
    & Theorem 1 of \shortcite{babaioff2021fairTruthful}
\\[\defaultaddspace] \tabSn & MMS & submodular & $\{0, -1\}$ & -- & -- & equal
    & Theorem 9 of \shortcite{barman2023fair}
\\[\defaultaddspace] \tabSn & GMMS & -- & -- & yes & -- & equal
    & Leximin is GMMS
\\ \bottomrule
\end{tabular}

\footnotesize
\begin{tightenum}
\item[*] \label{foot:dbl-mono-2}A function $v: 2^M \to \mathbb{R}$ is \emph{doubly monotone}
    if there is a partition $(G, C)$ of $M$ such that
    $v(g \mid \cdot) \ge 0$ $\forall g \in G$ and $v(c \mid \cdot) \le 0$ $\forall c \in C$.
\item[\textdagger] \label{foot:propx-li}Algorithm 2 of \citet{li2022almost} must be slightly modified
    to meet our slightly stricter definition of PROPx.
    Change line 5 from `if $|v_i(X_i)| > w_i$' to `if $|v_i(X_i)| \ge w_i$'.
\item[\textdaggerdbl] \label{foot:propavg}Algorithm 4.1 of \citet{kobayashi2025proportional}
    must be slightly modified to meet our slightly stricter definition of PROPavg.
    Change the definition of the PROPavg-graph to use a strict inequality instead.
\end{tightenum}
\end{table*}

\begin{table*}[htb]
\centering
\caption{Infeasibility of fairness notions}
\label{table:infeas}
\bigTableSize
\setcounter{tabSerial}{0}
\begin{tabular}{clcccccr}
\toprule & \scriptsize notion & \scriptsize valuation & \scriptsize marginals & \scriptsize identical & \scriptsize $n$ & \scriptsize entitlements &
\\ \midrule \tabSn & PROP & $m=1$ & $1$, $-1$ & yes & any & equal
    & \cref{cex:single-item}
\\[\defaultaddspace] \tabSn & APS & submod & $> 0$ & yes & $n=2$ & equal
    & \cref{cex:mms-not-aps-n2-submod}
\\[\defaultaddspace] \tabSn & APS & additive & $> 0$, $< 0$ & yes & $n=3$ & equal
    & Lemma C.1 of \shortcite{babaioff2023fair}
\\[\defaultaddspace] \tabSn & APS & subadd & $\{0, 1\}$ & yes & $n=3$ & equal
    & \cref{cex:gmms-not-aps-binary-subadd}
\\[\defaultaddspace] \tabSn & MMS & additive & $> 0$, $< 0$ & no & $n=3$ & equal
    & \shortcite{feige2022tight}
\\[\defaultaddspace] \tabSn & MMS & XOS & $\{0, 1\}$ & no & $n=2$ & equal
    & Theorem 3.5 of \shortcite{barman2021existence}
\\[\defaultaddspace] \tabSn & PROPx & additive & $> 0$ bival & yes & $n=3$ & equal
    & \cref{cex:propx}
\\[\defaultaddspace] \tabSn & PROPm & additive & mixed bival & yes & $n=3$ & equal
    & \cref{cex:propm-mixed-manna}
\\[\defaultaddspace] \tabSn & MXS & additive & $> 0$ & no & $n=2$ & unequal
    & \cref{cex:wmxs-goods}
\\[\defaultaddspace] \tabSn & MXS & additive & $< 0$ & no & $n=2$ & unequal
    & \cref{cex:wmxs-chores}
\\[\defaultaddspace] \tabSn & PROP1+M1S & additive & $1$ & yes & $n=3$ & unequal
    & \cref{cex:prop1-plus-m1s-ue}
\\[\defaultaddspace] \tabSn & PROP1 & supermod & $\ge 0$ bival & yes & $n=2$ & equal
    & \cref{cex:ef-not-prop-supmod}
\\[\defaultaddspace] \tabSn & PROPm & supermod & $\ge 0$ bival & yes & $n=2$ & equal
    & \cref{cex:ef-not-prop-supmod}
\\[\defaultaddspace] \tabSn & MMS & supermod & $\ge 0$ bival & no & $n=2$ & equal
    & \cref{cex:mms-supmod-goods}
\\[\defaultaddspace] \tabSn & MMS & supermod & $\le 0$ bival & no & $n=2$ & equal
    & \cref{cex:mms-supmod-chores}
\\[\defaultaddspace] \tabSn & MMS & submod & $> 0$ & no & $n=2$ & equal
    & \cref{cex:mms-submod-goods}
\\[\defaultaddspace] \tabSn & PMMS & -- & $\ge 0$ & no & $n=3$ & equal
    & Theorem 1 of \shortcite{byrka2025probing}
\\ \bottomrule
\end{tabular}
\end{table*}

\begin{example}[PROPx is infeasible]
\label[example]{cex:propx}
An instance with 3 equally-entitled agents having identical additive valuations over 2 goods
with values 10 and 1, respectively.
\end{example}

\begin{lemma}[WMXS is infeasible for goods, Theorem 8 of \citet{springer2024almost}]
\label[lemma]{cex:wmxs-goods}
Let $0 < \eps \le 1/4$, and $\phi \defeq (\sqrt{5}+1)/2$.
No MXS allocation exists for the fair division instance $\fdInst{[2]}{[4]}{(v_i)_{i=1}^2}{w}$,
where $w_1 = 1/(\sqrt{\phi}+1)$, $w_2 = \sqrt{\phi}/(\sqrt{\phi}+1)$,
and agents have additive valuations given by the following table:

\centering
\begin{tabular}{c|cccc}
$g$ & 1 & 2 & 3 & 4
\\ \hline $v_1(g)$ & $\eps$ & $1$ & $\phi$ & $\phi$
\\ $v_2(g)$ & $\eps$ & $\eps$ & $1$ & $1$
\end{tabular}
\end{lemma}
\ifVerbose
\begin{proof}[Proof sketch.]
Note that $1 + \eps < \sqrt{\phi} = w_2/w_1$.
The only bundles agent $2$ is EFX-satisfied with are
$\{1, 2, 3\}$, $\{1, 2, 4\}$, $\{3, 4\}$, and their supersets.
Hence, $\MXS_2 = 1 + 2\eps$.
The only bundles agent $1$ is EFX-satisfied with are
$\{1, 3\}$, $\{1, 4\}$, $\{2, 3\}$, $\{2, 4\}$, $\{3, 4\}$, and their supersets.
Hence, $\MXS_1 = \phi + \eps$.
One can check that no allocation is MXS.
\end{proof}
\fi

\begin{lemma}[WMXS is infeasible for chores, Theorem 17 of \citet{springer2024almost}]
\label[lemma]{cex:wmxs-chores}
Let $0 < \eps \le 1/4$, and $\phi \defeq (\sqrt{5}+1)/2$.
No MXS allocation exists for the fair division instance $\fdInst{[2]}{[4]}{(v_i)_{i=1}^2}{w}$,
where $w_1 = \sqrt{\phi}/(\sqrt{\phi}+1)$, $w_2 = 1/(\sqrt{\phi}+1)$,
and agents have additive disutilities given by the following table:

\centering
\begin{tabular}{c|cccc}
$c$ & 1 & 2 & 3 & 4
\\ \hline $-v_1(c)$ & $\eps$ & $1$ & $\phi$ & $\phi$
\\ $-v_2(c)$ & $\eps$ & $\eps$ & $1$ & $1$
\end{tabular}
\end{lemma}
\ifVerbose
\begin{proof}[Proof sketch.]
Note that $1 + \eps < \sqrt{\phi} = w_1/w_2$.
The only bundles agent $2$ is EFX-satisfied with are
$\{1, 2\}$, $\{3\}$, $\{4\}$, and their subsets.
Hence, $\MXS_2 = -1$.
The only bundles agent $1$ is EFX-satisfied with are
$\{2, 3\}$, $\{2, 4\}$, $\{1, 3\}$, $\{1, 4\}$, $\{1, 2\}$, and their subsets.
Hence, $\MXS_1 = -\phi^2$.
One can check that no allocation is MXS.
\end{proof}
\fi

\begin{example}[MMS is infeasible]
\label[lemma]{cex:mms-supmod-goods}
Let $0 \le a < b$. Let $\fdInst{[2]}{[4]}{(v_i)_{i=1}^2}{\eqEnt}$ be a fair division instance
where $v_1(S) = |S|a + (b-a)(\boolone(S \supseteq \{1, 2\}) + \boolone(S \supseteq \{3, 4\}))$
and $v_2(S) = |S|a + (b-a)(\boolone(S \supseteq \{1, 3\}) + \boolone(S \supseteq \{2, 4\}))$.
Then $v_1$ and $v_2$ are supermodular, each agent's MMS is $a+b$, and no MMS allocation exists.
In fact, the best multiplicative approximation to the MMS one can achieve is $2a/(a+b)$.
\end{example}

\begin{example}[MMS is infeasible]
\label[lemma]{cex:mms-supmod-chores}
Let $0 \le a < b$. Let $\fdInst{[2]}{[4]}{(v_i)_{i=1}^2}{\eqEnt}$ be a fair division instance
where $-v_1(S) = |S|b - (b-a)(\boolone(S \supseteq \{1, 2\}) + \boolone(S \supseteq \{3, 4\}))$
and $-v_2(S) = |S|b - (b-a)(\boolone(S \supseteq \{1, 3\}) + \boolone(S \supseteq \{2, 4\}))$.
Then $v_1$ and $v_2$ are supermodular, each agent has MMS $-(a+b)$, and no MMS allocation exists.
In fact, the best multiplicative approximation to the MMS one can achieve is $2b/(a+b)$.
\end{example}

\begin{lemma}[MMS is infeasible, Theorem 4.1 of \citet{ghodsi2018fair}]
\label[lemma]{cex:mms-submod-goods}
Let $\fdInst{[2]}{[4]}{(v_i)_{i=1}^2}{\eqEnt}$ be a fair division instance,
where $v_1$ and $v_2$ are as follows ($\eps \in \mathbb{R}_{\ge 0}$):
\[ \begin{array}{c|c|c}
S & v_1(S) - |S|\eps & v_2(S) - |S|\eps
\\\hline |S| \le 1 & |S| & |S|
\\ \{1, 2\}, \{3, 4\} & 2 & 3/2
\\ \{2, 3\}, \{1, 4\} & 3/2 & 2
\\ \{1, 3\}, \{2, 4\} & 3/2 & 3/2
\\ |S| \ge 3 & 2 & 2
\end{array} \]
Then $v_1$ and $v_2$ are submodular, each agent's MMS is $2+2\eps$,
and no MMS allocation exists.
\end{lemma}
\ifVerbose
\begin{proof}
For any $i \in [2]$, $S \subseteq [4]$, and $g \in [4] \setminus S$, we have
\begin{tightenum}
\item $S = \emptyset \implies v_i(g \mid S) = 1 + \eps$.
\item $|S| = 1 \implies v_i(g \mid S) \in \{1/2 + \eps, 1 + \eps\}$.
\item $|S| = 2 \implies v_i(g \mid S) \in \{\eps, 1/2 + \eps\}$.
\item $|S| = 3 \implies v_i(g \mid S) = \eps$.
\end{tightenum}
Hence, $v_i(g \mid S)$ decreases with $|S|$, so $v_i$ is submodular.
It is easy to verify that each agent's MMS is $2+2\eps$,
and that no allocation is MMS.
\end{proof}
\fi

\section{Implicit Representation of Set Families}
\label{sec:fd-set-family}

\Cref{sec:cpig} introduced the concept of \emph{conditional predicate implications},
and presented an algorithm for inferring additional implications and counterexamples.
In fair division, we defined $\Omega$ to be the set of all pairs $(\Ical, A)$,
where $\Ical$ is a fair division instance and $A$ is an allocation for $\Ical$.
We defined $\Fcal$ to be a set family over $\Omega$, and in fair division,
each set in $\Fcal$ represents a setting.
But $\Omega$ and the sets in $\Fcal$ are uncountable, so how do we compute with them?
This is the question we answer formally in this section.

\subsection{Mappings from Partial Orders}

\begin{definition}
\label{defn:set-family-repr}
A set family $\Fcal \subseteq 2^{\Omega}$ is \emph{represented by} a partial order $(P, \preceq)$
if there exists an order-preserving surjective mapping $f: P \to \Fcal$,
i.e., for all $S \in \Fcal$, there exists $x \in P$ such that $f(x) = S$,
and for all $x, y \in P$, we have $x \preceq y \implies f(x) \subseteq f(y)$.
\end{definition}

Note that the converse is not required to be true, i.e.,
$f(x) \subseteq f(y)$ need not imply $x \preceq y$.
Hence, if $P$ is an antichain, then $P$ trivially represents $\Fcal$.
However, such a representation is useless.
The more a representation captures the subset relations in $\Fcal$,
the better that representation is.

\begin{example}
Let $E$ and $O$ be the sets of even and odd integers, respectively, i.e.,
$E \defeq \{2i: i \in \mathbb{Z}\}$ and $O \defeq \{2i+1: i \in \mathbb{Z}\}$.
Then the set family $\Fcal \defeq \{E, O, \mathbb{Z}\}$ can be represented by
the partial order $(\{e, o, a\}, \{e \preceq e, a \preceq a, o \preceq o, e \preceq a, o \preceq a\})$
over three elements, where the corresponding mapping $f$ is given by
$f(e) = E$, $f(o) = O$, and $f(a) = \mathbb{Z}$.
\end{example}

Hence, for the conditional predicate implication problem,
if we can represent a set family $\Fcal \subseteq 2^{\Omega}$
by a finite partial order $(P, \preceq)$,
then we can indirectly specify the sets that implications and counterexamples
are conditioned on by elements in $P$.
In fact, for computation, we do not even need to know the set $\Fcal$ and the mapping $f$;
we can just work with elements in $P$ instead.
In the algorithm for inferring additional implications and counterexamples,
we perform several checks of the form $S \subseteq T$, where $S, T \in \Fcal$.
We replace them with checks of the form $x \preceq y$, where $f(x) = S$ and $f(y) = T$.

\subsection{Partial Order for Fair Division Settings}

We represent a fair division setting as a 5-tuple, as specified in \cref{sec:settings}.
We now explain how to define a partial order on these 5-tuples,
and how to map each 5-tuple to a subset of $\Omega$.

\begin{definition}[Product order]
Let $((P_i, \preceq_i))_{i=1}^k$ be a sequence of partial orders.
Their \emph{product} is another partial order $(P, \preceq)$, where
$P \defeq \prod_{i=1}^k P_i \defeq \{(p_i)_{i=1}^k: p_j \in P_j \forall j \in [k]\}$
and $(p_1, \ldots, p_k) \preceq (q_1, \ldots, q_k)$ iff $p_i \preceq_i q_i$ for all $i \in [k]$.
\end{definition}

\ifVerbose
\begin{example}
The product of $(\mathbb{N}, \le)$ with itself is $(\mathbb{N}^2, \preceq)$,
where $(a_1, a_2) \preceq (b_1, b_2)$ iff $a_1 \le b_1$ and $a_2 \le b_2$.
\end{example}
\fi

Recall the 5 features of fair division from \cref{sec:settings}:
whether entitlements are equal,
whether there are only two agents,
whether agents have identical valuations,
valuation function type,
and marginal values.
We define a partial order for each of these 5 features.
The first three features are represented by the \emph{boolean} partial order:
$(\{\mathrm{true},\allowbreak \mathrm{unknown}\},
\{\mathrm{true} \preceq \mathrm{unknown},\allowbreak
\mathrm{true} \preceq \mathrm{true},\allowbreak
\mathrm{unknown} \preceq \mathrm{unknown}\})$.
The partial orders for the last two features are represented as DAGs in
\cref{fig:dag-posets:valuation,fig:dag-posets:marginals}, respectively, in \cref{sec:settings-extra}.
(Formally, for a DAG $G = (V, E)$, the corresponding partial order is $(V, \preceq)$,
where $u \preceq v$ iff there is a path from $u$ to $v$ in $G$.)
Let $(P, \preceq)$ be the product of these 5 partial orders.
For a fair division setting $s \in P$, let
$f(s) \defeq \{(\Ical, A): \Ical$ is an instance consistent with $s$,
$A$ is an allocation for $\Ical\}$, and $\Fcal \defeq \{f(s): s \in P\}$.
It is easy to check that $f$ is order-preserving and surjective.
This completes our description of how to map fair division settings to subsets of $\Omega$.

Note that $f$ is not injective. The settings
$s_1 \defeq (\mathrm{unknown},\allowbreak \mathrm{unknown},\allowbreak \mathrm{true},\allowbreak \mathrm{additive},\allowbreak \{1\})$
and $s_2 \defeq (\mathrm{unknown},\allowbreak \mathrm{unknown},\allowbreak \mathrm{unknown},\allowbreak \mathrm{general},\allowbreak \{1\})$
map to the same set in $\Fcal$, because if
each item's marginal value is 1, then valuations are identical and additive.
Querying the inference engine with $s_2$
may fail to infer implications that rely on additivity or identical valuations.
Note that $s_1 \preceq s_2$.
Among equivalent settings, querying the inference engine with a minimal setting
gives the most informative results, provided that counterexamples fed to the engine
are also conditioned on minimal settings.

\end{document}